\documentclass[12pt]{article} 
\usepackage{float}
\usepackage{amsmath, amsthm}
\usepackage{amsmath}
\usepackage{amsfonts}
\usepackage{amssymb}
\usepackage{esint}
\usepackage{mathrsfs}
\usepackage{verbatim}
\usepackage{amscd}
\usepackage{amsthm}
\usepackage{enumerate}
\usepackage[pdftex]{graphicx}
\usepackage{wrapfig}

\newtheoremstyle{thm}{1.5ex}{1.5ex}{\itshape\rmfamily}{}
{\bfseries\rmfamily}{}{2ex}{}

\newtheoremstyle{rem}{1.3ex}{1.3ex}{\rmfamily}{} 
{\itshape}
{} {1.5ex}{}

\newtheoremstyle{theorem}{1.5ex}{1.5ex}{\itshape\rmfamily}{} {\bfseries\rmfamily}{}{2ex}{}

\newtheoremstyle{def}{1.5ex}{1.5ex}{\slshape\rmfamily}{} {\bfseries\rmfamily}{}{2ex}{}

\newtheoremstyle{rem}{1.3ex}{1.3ex}{\rmfamily}{} {\itshape}
{} {1.5ex}{}

\newenvironment{changemargin}[2]{%
\begin{list}{}{%
\setlength{\leftmargin}{#1}%
\setlength{\rightmargin}{#2}%
}%
\item[]}
{\end{list}}

\theoremstyle{theorem}
\newtheorem{theorem}{Theorem}[section]
\newtheorem{lemma}[theorem]{Lemma}
\newtheorem{nota}[theorem]{Notation}
\newtheorem{prop}[theorem]{Proposition}
\newtheorem{cor}[theorem]{Corollary}

\newtheorem*{Main Theorem}{Main Theorem.}

\theoremstyle{definition}
\newtheorem{defn}[theorem]{Definition}

\newtheorem{remark}[theorem]{Remark}

\numberwithin{equation}{section}

\usepackage{color}
\definecolor{purple}{rgb}{0.65, 0, 1}
\definecolor{orange}{rgb}{1,.5,0}
\definecolor{brown}{rgb}{.9,.73,.26}

\newcommand{\e}{\varepsilon}

\newcommand{\oo}{\circledcirc}

\DeclareSymbolFont{extraup}{U}{zavm}{m}{n}
\DeclareMathSymbol{\varheart}{\mathalpha}{extraup}{86}
\DeclareMathSymbol{\vardiamond}{\mathalpha}{extraup}{87}

\begin{document}

\title
{\Large On the Rate of Convergence for Critical Crossing Probabilities}
\author
{ I. Binder$^\star$, L. Chayes$^{\dagger}$ and H. K. Lei$^{\dagger\dagger}$}
\date{}
\maketitle

\vspace{-4mm}
\centerline{${}^\star$\textit{Department of Mathematics, University of Toronto}}
\centerline{${}^{\dagger}$\textit{Department of Mathematics, University of California at Los Angeles}}
\centerline{${}^{\dagger\dagger}$\textit{Department of Mathematics, California Institute of Technology}}

\abstract{For the site percolation model on the triangular lattice and certain generalizations for which Cardy's Formula has been established we acquire a power law estimate for the \emph{rate} of convergence of the crossing probabilities to Cardy's Formula.}

\section{Introduction}
Starting with the work 
\cite{stas_perc}
and continuing in: 
\cite{cnc}, \cite{werner_perc}
\cite{CL}, \cite{BCL1}, \cite{BCL2},
the validity of Cardy's formula
\cite{CARDY_himself}
-- which describes the limit of crossing probabilities
for certain percolation models --
and the subsequent consequence of an SLE$_{6}$
description for the associated limiting \textit{explorer process} has been well established.
The purpose of this work is to provide some preliminary quantitative estimates.  Similar work along these lines has already appeared in 
\cite{rlerw} (also see \cite{fnew})
in the context of the so--called loop erased random walk for both the observable and the process itself.  Here, our attention will be confined to the percolation observable as embodied by Cardy's Formula
for crossing probabilities.

While in the case of the loop erased random walk, estimates on the observable can be reduced to certain Green's function estimates, in the case of percolation the observables are not so readily amenable.
Instead of Green's functions, we shall have to consider the Cauchy integral representation of the complexified\textit{ crossing probability functions}, as first introduced in \cite{stas_perc}.  As demonstrated in \cite{stas_perc} (see also \cite{beffara} and \cite{CL}) these functions converge to conformal maps from the domain under consideration -- where the percolation process takes place -- to the equilateral triangle.  Thus, a combination of some analyticity property and considerations of boundary value should, in principal, yield a rate of convergence.  

However, the associated procedure requires a few domain deformations, each of which must be demonstrated to be ``small'', in a suitable sense.
While such considerations are not important for \emph{very regular} domains 
(which we will not quantify)
in order to consider general domains, a more robust framework for quantification is called for.  For this purpose, we shall introduce a procedure where all portions of the domain are explored via percolation crossing problems. This yields a multi--scale sequence of neighborhoods around each boundary point where
the nature of the boundary irregularities determines the sequence of successive scales. 
Thus, ultimately, we are permitted to measure the distances 
between regions by counting the number of neighborhoods which separate them.    This procedure is akin to the approach of Harris \cite{H} in his study of the critical state at a time when detailed information about the nature of the state was unavailable.

Ultimately we establish a power law estimate (in mesh size) for the rate of convergence in any domain with boundary dimension less than two. (For a precise statement see the Main Theorem below.)  As may or may not be clear to the reader at this point the hard quantifications must be done via percolation estimates -- as is perhaps not surprising since we cannot easily utilize continuum estimates before having reached the continuum in the first place.  The plausibility of a power law estimate then follows from the fact that most \emph{a priori} percolation estimates are of this form.

\bigskip

Finally, we should mention that this problem is also treated in the posting \cite{mnw}, which appeared at approximately the same time as (the preliminary version of) the present work.  The estimates in \cite{mnw} are more quantitative, however, the class of domains treated therein are restricted.  In the present work we make  no efforts towards \emph{precise} quantification, but we shall treat the problem for essentially arbitrary domains.  It is remarked that convergence to Cardy's Formula in a general class of domains is, most likely, an essential ingredient for acquiring a rate of convergence to SLE$_6$ for the percolation interfaces.
 
\section{Preliminaries}

\subsection{The Models Under Consideration}
We will be considering critical percolation models in the plane.  However in contrast to
the generality professed in \cite{BCL1}, \cite{BCL2} -- where, essentially, ``all'' that was required was a proof of Cardy's formula, here the mechanism of how Cardy's formula is established will come into play.  Thus, we must restrict attention to the triangular site percolation problem considered in \cite{stas_perc}
and the generalization provided in \cite{CL}.  These models can all be expressed 
in terms of random colorings (and sometimes double colorings) of hexagons.  As is traditional, the 
competing colors are designated by blue and yellow.  
We remind the reader that criticality implies that there are scale independent bounds in $(0,1)$ for crossing probabilities -- in either color -- between non--adjacent sides of regular polygons.  In this work, for the most part, we will utilize crossings in rectangles with particular aspect ratios.


\subsection{The Observable}
Consider a fixed domain $\Omega \subset \mathbb C$ that is a conformal rectangle with marked points (or prime ends)
$A$, $B$, $C$ and $D$
which, as written, are in cyclic order.  We let $\Omega_n$ denote the lattice approximation at scale
$\varepsilon = n^{-1}$ to the domain $\Omega$.  
The details of the construction of $\Omega_n$ 
-- especially concerning boundary values and explorer processes --
are somewhat tedious and have been described e.g.,  in
\cite{BCL2} \S 3 $\&$ \S4 and \cite{BCL1} \S 4.2.
For present purposes, it is sufficient to know that 
$\Omega_n$ consists of the maximal union of lattice hexagons
-- of radius $1/n$ -- whose closures lie entirely inside $\Omega$; we sometimes refer to this as the \emph{canonical approximation}.  (We shall also have occasions later to use \textit{other} discrete interior approximating domains which are a subset of $\Omega_n$.)
Moreover, boundary arcs can be appropriately colored 
and lattice points $A_n$ -- $D_n$ can be selected.  We consider \textit{percolation} problems in 
$\Omega_n$.  

The pertinent object to consider is a crossing probability: performing percolation on $\Omega_n$, we ask for the crossing probability -- say in yellow -- from $(A_n, B_n)$ to $(C_n, D_n)$.  Here and throughout this work, a colored \emph{crossing} necessarily implies the existence of a self--avoiding, connected path of the designated color with endpoints in the specified sets and/or that satisfies specific separation criteria.  
Below we list various facts, definitions and notations related to the observable that will be used throughout this work.
In some of what follows, we temporarily neglect the marked point $A_n$ and regard 
$\Omega_n$ with the three remaining marked points as a conformal triangle.  

\begin{itemize}

\item[$\circ$] Let us recall the functions introduced in \cite{stas_perc}, here denoted by $S_B, S_C, S_D$ where e.g., $S_D(z)$ with $z \in \Omega_n$ a lattice point, is the probability of a yellow crossing from $(C_n, D_n)$ to $(D_n, B_n)$ separating $z$ from $(B_n, C_n)$.  Note that it is implicitly understood that the $S_B, S_C, S_D$--functions are defined on the discrete level; to avoid clutter, we suppress the $n$ index for these functions.  Moreover, we will denote the underlying \emph{events} associated to these functions by $\mathbb S_B, \mathbb S_C, \mathbb S_D$, respectively.

\item[$\circ$] It is the case that the functions $S_B, S_C, S_D$ are invariant  under exchange of color (see \cite{stas_perc} and \cite{CL}).
While it is not essential to the arguments in this work, we sometimes may take  liberties regarding whether we are considering a yellow or blue version of these functions.  

\item[$\circ$] It is also easy to see that e.g., $S_B$ has boundary value 0 on $(C_n, D_n)$ and 1 at the point $B_n$.  Moreover, the complexified function $S_n = S_B + \tau S_C + \tau^2 S_D$, with $\tau = e^{2\pi i/3}$, converges to the conformal map to the equilateral triangle with vertices at $1, \tau, \tau^2$, which we denote by $\mathbb T$. (See \cite{stas_perc}, \cite{beffara}, \cite{BCL2}.) 

\item[$\circ$]For finite $n$, we shall refer to the object $S_{n}(z)$ as the \textit{Carleson--Cardy--Smirnov} function and sometimes abbreviated CCS--function.

\item[$\circ$] We will use 
$H_n: \Omega_n \rightarrow \mathbb T$
to denote
the unique conformal map which sends $(B_n, C_n, D_n)$ to $(1, \tau, \tau^2)$.  Similarly, $H : \Omega \rightarrow \mathbb T$
is the corresponding conformal map of the continuum domain.

\item[$\circ$] With $A_n$ reinstated, we will denote by $\mathscr C_n$ the crossing probability of the conformal rectangle 
$\Omega_n$ and $\mathscr C_{\infty}$ its limit in the domain $\Omega$; i.e., Cardy's Formula in the limiting domain.  

\item[$\circ$] 

Since $S_C(A_n) \equiv 0$, 
\[S_n(A_n) = S_B(A_n) + \tau^2 S_D(A_n) = [S_B(A_n) - \frac{1}{2} S_D(A_n)] - i \frac{\sqrt 3}{2} S_D(A_n).\]
Now we recall (or observe) that $\mathscr C_n$ can be realized as $S_D(A_n)$ and so from the previous display, $\mathscr C_n = -\frac{2}{\sqrt 3} \cdot \mbox{Im}[S_n(A_n)]$.  Since it is already known that $S_n$ converges to $H$ (see \cite{stas_perc}, \cite{beffara}, \cite{BCL2}) it is also the case that $\mathscr C_\infty = -\frac{2}{\sqrt 3}\cdot \mbox{Im}[H(A)]$.  Therefore to establish a \emph{rate} of convergence of $\mathscr C_n$ to $\mathscr C_\infty$, it is sufficient to show that there is some $\psi > 0$ such that 
\[ |S_n(A_n) - H(A)| \leq C_\psi \cdot n^{-\psi},\]
for some $\psi > 0$ and $C_\psi < \infty$ which may depend on the domain $\Omega$.

\item[$\circ$] The functions $S_n$ are not \emph{discrete analytic} but the associated contour integrals vanish with lattice spacing (see \cite{stas_perc}, \cite{beffara} and \cite{CL}.) In particular, this is exhibited
by the fact that the contour integral around some closed discrete contour 
$\Gamma_n$ behaves like the length of $\Gamma_n$ times $n$ to some negative power.  Also, the functions $S_n$ are H\"older continuous with estimates which are uniform for large $n$.  For details we refer the reader to Definition \ref{nhf}.

\end{itemize}

Our goal in this work is to acquire the following theorem on the rate of convergence of the finite volume crossing probability, 
$\mathscr C_n$, to its limiting value:

\begin{Main Theorem}
\label{MT}
Let $\Omega$ be a domain and $\Omega_n$ its canonical discretization.  Consider the site percolation model or the models introduced in \cite{CL} on the domain $\Omega_n$.  In the case of the latter we also impose the assumption that the boundary Minkowski dimension is less than 2 (in the former, this is not necessary).  Let $\mathscr C_n$ be described as before.  Then there exists some $\psi(\Omega) > 0$ (which may depend on the domain $\Omega$) such that $\mathscr C_n$ converges to its limit with the estimate
$$
|\mathscr C_n - \mathscr C_{\infty}| \lesssim n^{-\psi},
$$
provided $n \geq n_{0}(\Omega)$ is sufficiently large and the symbol
$\lesssim$ is described with precision in 
Notation \ref{AWZ}
below.  
\end{Main Theorem}

\begin{nota}
\label{AWZ}
In the above and throughout this work, we will be describing 
asymptotic behaviors of various quantities as a function of small or large parameters (usually $n$ in one form or another).  The
relation $X\lesssim Y$ relating two functions $X$ and $Y$ of large or small parameters (below denoted by $M$ and $m$, respectively) means that there exists a constant $c \in (0, \infty)$ independent of $m$ and $M$ such that for all $M$ sufficiently large and/or $m$ sufficiently small $X(m, M) \leq c \cdot Y(m, M)$. 
\end{nota}

\begin{remark}
The restrictions on the boundary Minkowski dimension for the models in \cite{CL} is not explicitly important in this work and will only be implicitly assumed as it was needed in order to establish convergence to Cardy's Formula.  
\end{remark}

\begin{figure}[H]
\vspace{.15 in}
\centering
\includegraphics[scale=.25]{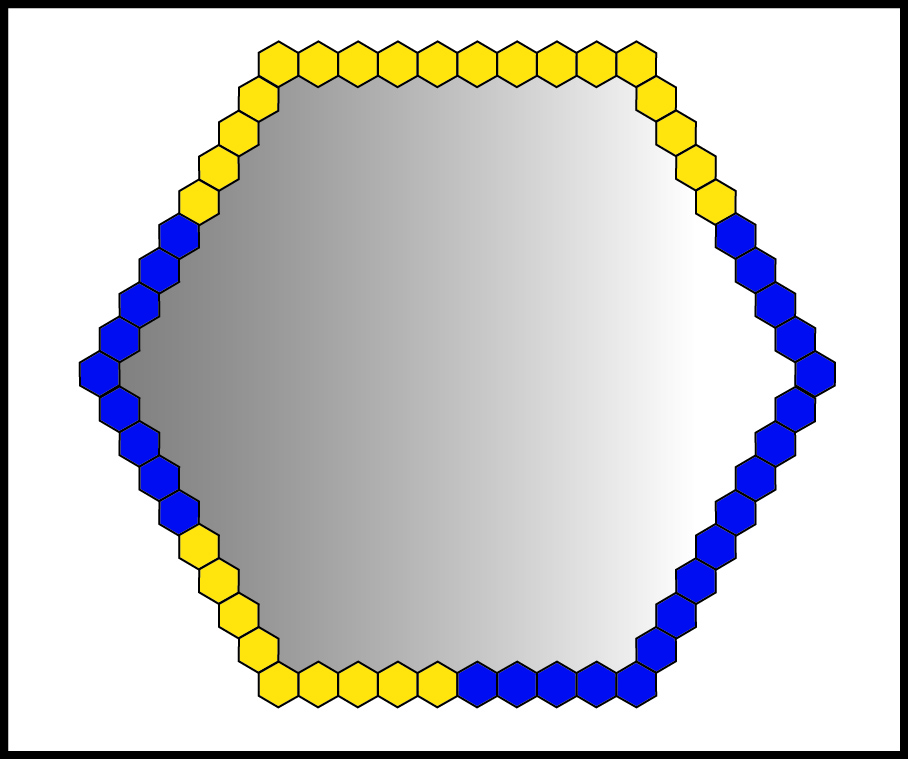}
\begin{changemargin}{.55in}{.55in}
\caption{{\footnotesize In certain conformal rectangles, the crossing probabilities at the discrete level will be identically $\frac{1}{2}$ independent of scale.}}
\label{DHU}
\end{changemargin}
\vspace{-.3 in}
\end{figure}

\begin{remark}
It would seem that complementary lower bounds of the sort presented in the Main Theorem are actually \textit{not} possible.  For example, in the triangular site model, the crossing probabilities for particular shapes are 
identically $\tfrac{1}{2}$ independently of $n$, as is demonstrated in Figure \ref{DHU}.
\end{remark}

We end this preliminary section with some notations and conventions: (i) the notation $\mbox{dist}(\cdot, \cdot)$ denotes the usual Euclidean distance while the notation $d_{\rm sup}( \cdot, \cdot)$ denotes the sup--norm distance between curves; (ii) we will make use of both \emph{macroscopic} and \emph{microscopic} units, with the former corresponding to an $\e \rightarrow 0$ approximation to shapes of fixed scale and the latter corresponding to $n \rightarrow \infty$, wherein distances are measured relative to the size of a hexagon.  So, even though analytical quantities are naturally expressed in macroscopic units, it is at times convenient to use microscopic units when performing percolation constructions; (iii) we will use $a_1, a_2, \dots $ to number the powers of $n$ appearing in the \emph{statements} of lemmas, theorems, etc.  Thus, throughout, $n = \e^{-1}$.  Constants used in the course of a proof are considered temporary and duly forgotten after the Halmos box.
  
\section{Proof of the Main Theorem}

Our strategy for the proof of the Main Theorem is as follows: recall that $H_n$ is the conformal map from $\Omega_n$ to $\mathbb T$ (the ``standard'' equilateral triangle) so that $B_n, C_n, D_n$ map to the three corresponding vertices, where it is reiterated that $\mathscr C_n$ corresponds to a boundary value of $S_n$.  Thus it is enough to uniformly estimate the difference between $S_n$ and $H_n$ and then the difference between $H_n$ and $H$. 

Foremost, the discrete domain may itself be a bit too rough so we will actually be working with an approximation to $\Omega_n$ which will be denoted by $\Omega_n^\square$ (see Proposition \ref{pr}).  Now, on $\Omega_n^\square$, we have the function 
$S_n^\square$ associated with the corresponding percolation problem \textit{on this domain} and, similarly, the conformal map $H_n^\square: \Omega_n^\square \to \mathbb T$.  Via careful consideration of Euclidean distances and distortion under the conformal map, we will be able to show that both $|S_n(A_n) - S_n^\square(A_n^\square)|$ (for an appropriately chosen $A_n^\square \in \partial \Omega_n^\square$) and $|H(A) - H_n^\square(A_n^\square)|$ are suitably small (see Theorem \ref{s1}).  Thus we are reduced to proving a power law estimate for the domain $\Omega_n^\square$. 

Towards this goal, we introduce
the \emph{Cauchy--integral extension} of $S_n^{\square}$, which we denote by
$F_n^{\square}$, so that
$$
F_n^{\square}(z) := \frac{1}{2\pi i} \oint_{\partial \Omega_n^{\square}} \frac{S_n^\square(\zeta)}{\zeta - z}~d\zeta.
$$

Now by using the H\"older continuity properties and the \emph{approximate} discrete analyticity properties of the $S_n$'s, we can show that, 
barring the immediate vicinity of the boundary, 
the difference between $F_n^\square$ and $S_n^\square$ is power law small (see Lemma \ref{ci}).  It follows then that in an even smaller domain,
$\Omega_n^\vardiamond$, which can be realized as the inverse image of a uniformly shrunken version of $\mathbb T$, 
the function $F_n^\square$ is in fact \emph{conformal} and thus it is uniformly close to $H_n^\vardiamond$, which is \emph{the} conformal map from 
 $\Omega_n^\vardiamond$ to $\mathbb T$ (see Lemma \ref{cdmd}).  

Now for $z \in \Omega_n^\square$ the  dichotomy we have introduced  is not atypical:  on the one hand, $F_n^\square(z)$ is manifestly analytic but does not necessarily embody the function $S_n^\square$ of current interest.  On the other hand, $S_n^\square(z)$ has the desired boundary values -- at least on $\partial \Omega_n^\square$ -- but is, ostensibly, lacking in analyticity properties.  Already the approximate discrete analyticity properties permit us to compare $F_n^\square$ to $S_n^\square$ in $\Omega_n^\vardiamond$.  In order to return to the domain $\Omega_n^\square$, we require some control on the ``distance'' between $\Omega_n^\vardiamond$ and $\Omega_n^\square$ (not to mention a suitable choice of some point $A_n^\vardiamond \in \partial \Omega_n^\vardiamond$ as an approximation to $A$).  It is indeed the case that \emph{if} $\Omega_n^\vardiamond$ is close to $\Omega_n^\square$ in the Hausdorff distance, then the proof can be quickly completed by using distortion estimates and/or H\"older continuity of the $S$ function.  However, such information translates into an estimate on the continuity properties of the \emph{inverse} of $F_n^\square$, which is not \emph{a priori} accessible (and, strictly speaking, not always true).   

Further thought reveals that we in fact require the domain $\Omega_n^\vardiamond$ to be close to $\Omega_n^\square$ in both the conformal sense and in the sense of ``percolation'' -- which can be understood as being measured via local crossing probabilities.  (In principle, given sufficiently strong control on boundary distortion of the relevant conformal maps, these notions should be directly equivalent; however, we do not explicitly address this question here, as this would entail overly detailed consideration of domain regularity.)  

While with a \textit{deliberate} choice of a point on the boundary corresponding to $A$ we may be able to ensure that either one or the other of the two criteria can be satisfied, it is not readily demonstrable that \textit{both} can be simultaneously satisfied without some additional detailed considerations; it is for this reason that we will introduce and utilize the notion of \emph{Harris systems} (see Theorem \ref{hreg}) in order to quantify the distances between $\Omega_n^\vardiamond$ and $\Omega_n^\square$.  

The Harris systems are collections of concentric topological rectangles (portions of annuli) of various scales centered on points of $\partial \Omega_n^\square$ and heading towards some ``central region'' of $\Omega_n^\square$; they are constructed so that uniform estimates are available for both the traversing of each annular portion and the existence of an obstructing ``circuit'' (in dual colors).  This leads to a natural choice of $A_n^\vardiamond$: it is a point on $\partial \Omega_n^\vardiamond$ which is in the Harris system of $A_n^\square$ (i.e., a point in one of the ``rings'').  Consequently, 
the distance between $A_n^\square$ and $A_n^\vardiamond$ -- and other such pairs as well -- can be measured vis a counting of Harris segments (see Lemma \ref{dp}).  

Specifically, we will make use of another auxiliary point, $A_n^\diamondsuit$, which is also in the Harris system centered at $A_n^\square$, chosen so that it is \emph{inside} the domain $\Omega_n^\vardiamond$.  The task of providing an estimate for $|S_n^\square(A_n^\diamondsuit) - S_n^\square(A_n^\square)|$ (and thus also $|F_n^\square(A_n^\diamondsuit) - S_n^\square(A_n^\square)|$) is immediately accomplished by the existence of suitably many Harris segments surrounding both $A_n^\square$ and $A_n^\diamondsuit$ (see Proposition \ref{s2}).   Also, considering $n$ to be \emph{fixed}, the domain $\Omega_n^\square$ can be approximated at scales $N^{-1} \ll n^{-1}$ and the estimates derived from the Harris systems remain \emph{uniform} in $N$ as $N$ tends to infinity and thus also immediately imply an estimate for $|H_n^\square(A_n^\square) - H_n^\square(A_n^\diamondsuit)|$ (see Proposition \ref{dia2}).

At this point what remains to be established is an estimate relating the conformal map $H_n^\vardiamond$, which is defined by percolation at scale $n$ via $F_n^\square$, and $H_n^\square$, the ``original'' conformal map.  It is here that we shall invoke a Marchenko theorem for the triangle $\mathbb T$ (see Lemma \ref{distort}): indeed, again considering $\Omega_n^\square$ to be a fixed domain and performing percolation at scales $N^{-1} \ll n^{-1}$, we have by convergence to Cardy's Formula that $S_{n, N}^\square(s) \rightarrow H_n^\square(s)$ as $N \rightarrow \infty$, for all $s \in \partial \Omega_n^\vardiamond$.  The inherent \emph{scale invariance} of the Harris systems permits us to establish that in fact $S_{n, N}^\square(s)$ is close to $\partial \mathbb T$, \emph{uniformly} in $N$ (see Lemma \ref{s3}) and thus, $H_n^\square(\partial \Omega_n^\vardiamond)$ is close to $\partial \mathbb T$ (in fact in the \emph{supremum norm}).  Armed with this input, the relevant Marchenko theorem applied at the point $A_n^\diamondsuit$ immediately gives that $H_n^\square(A_n^\diamondsuit) - H_n^\vardiamond(A_n^\diamondsuit)$ is suitably small.

The technical components relating to the Cauchy--integral estimate and the construction of the Harris systems are relegated to Section \ref{KBB} and 
Section \ref{SBT}, respectively.  As for the rest, we will divide the proof of the main theorem into three subsections, corresponding to:
\begin{itemize}
\item[(i)] the regularization of the boundary (introduction of $\Omega_n^\square$) and showing that crossing probabilities are close for the domains $\Omega_n^\square, \Omega_n~\&~\Omega$;
\item[(ii)] the construction of the Cauchy--integral $F_n^\square$ and of the domain $\Omega_n^\vardiamond$;
\item[(iii)] the establishment of the remaining estimates needed to show that the domains $\Omega_n^\vardiamond$ and $\Omega_n^\square$ are suitably close, by using the Harris systems of neighborhoods.
\end{itemize}

\subsection{Regularization of Boundary Length}
We now introduce the domain $\Omega_n^\square \subseteq \Omega_n$.  The primary purpose of this domain is to reduce the boundary length of the domain
that need be considered; in particular, this will be pivotal when estimating the discrete analyticity properties of $S_n^\square$ in the next section.

\begin{defn}
Let $1 > a_1 > 0$ and consider a square grid
whose elements are squares of (approximately) microscopic size $n^{a_{1}}$
and let $\Omega^\square_{n}$ denote the union of all (hexagons within the) squares of this grid that are entirely within the original domain 
$\Omega$.  
\end{defn}

We have:
\begin{prop}\label{pr}
Let $\Omega \subseteq \mathbb C$ be a domain with boundary Minkowski dimension at most $1 + \alpha'$ with $\alpha' \in [0, 1]$, which we write as $M(\partial \Omega) < 1 + \alpha$ for any $\alpha > \alpha'$.  Then the domain $\Omega_n^\square$ satisfies $\Omega_n^\square \subseteq \Omega_n$ and 
\[|\partial \Omega_n^\square| \lesssim n^{\alpha(1-a_1)} .\]
\end{prop}
\begin{proof}
Since $M(\partial \Omega) < 1 + \alpha$ we have (for all $n$ sufficiently large)  that the number of boxes required to cover $\partial \Omega$ is essentially bounded from above by $(n^{1-a_1})^{1+\alpha}$ which is then multiplied by $\frac{1}{n^{(1-a_1)}}$, the size of the box (in macroscopic units).  The fact that $\Omega_n^\square \subseteq \Omega_n$ is self--evident. 
\end{proof}
Next we will choose $A_n^\square, B_n^\square, C_n^\square, D_n^\square \in \partial \Omega_n^\square$ by some procedure to be outlined below and denote by $S_n^\square$ the corresponding CCS--function.  Particularly, this can be done so that the crossing probabilities do not change much:

\begin{theorem}
\label{s1}
Let $\Omega_{n}^\square \subseteq \Omega_n$ with marked boundary points $(A_{n}, \dots , D_{n})$ be as described, so particularly $\partial \Omega_n^\square$ is of distance at most $n^{1-a_1}$ from $\partial \Omega_n$.
Then there is an
$A_{n}^{\square}$ as well as $B_{n}^{\square}$, $C_{n}^{\square}$
and $D_{n}^{\square}$
such that the corresponding
$S_{n}^{\square}$ satisfies, for some 
$a_{2} > 0$ and for all $n$ sufficiently large,
$$
|S_n(A_n) - S_n^\square(A_n^{\square})| \lesssim n^{-a_2}
$$
and, moreover,
\[ |H(A) - H_n^\square(A_n^\square)| \lesssim n^{-a_2}.\]
\end{theorem}
\begin{remark}
In the case that the separation
between
$A_{n}$ and $\partial\Omega_{n}$
is the order of 
$n^{a_{1}}$ -- as is usually imagined --
facets of Theorem \ref{s1} are essentially trivial.  
However, the reader is reminded that 
$A_{n}$ could be deep inside a ``fjord'' and well separated from 
$\partial\Omega_{n}^{\square}$.  In this language,
the forthcoming arguments will demonstrate that, notwithstanding,
an $A_{n}^{\square}$ may be chosen near the mouth of the fjord
for which the above estimates hold.
\end{remark}

\noindent
\textit{Proof of Theorem \ref{s1}.}
For $\eta > 0$ and a subset $K \subset \overline{\Omega}$ we will denote by  $N_\eta(K)$ the $\eta$--neighborhood of $K$
intersected with $\Omega$.  Now let us first choose $\eta$ sufficiently small so that 
\[ \left[[B, C, D] \cup N_{4\eta}(B) \cup N_{4\eta}(D) \cup N_{4\eta}(C) \right] \cap N_{4\eta}(A) = \emptyset,\]
where $[B, C, D]$ denotes the closed boundary segment containing the prime ends $B, C, D$.

Next we assume that $n > n_{\circ}$ where  $n_{\circ}$
is large enough so that
for all $n > n_{\circ}$, 
$A_{n}\in N_{\eta}(A)$ , \dots, 
$D_{n}\in N_{\eta}(D)$.  Moreover, 
$\Omega_{n}^{\square}\cap N_{\eta}(A) \neq \emptyset$
and similarly for $\Omega_{n}^{\square}\cap N_{\eta}(B), \dots , 
\Omega_{n}^{\square}\cap N_{\eta}(D)$. Then, since
$$
0 < \text{dist}(([A,B]\setminus N_{\eta}(A)), ([D,A]\setminus N_{\eta}(A)))
$$
it is assumed that for $n > n_{\circ}$, the above is very large compared with $n^{-(1-a_{1})}$ and similarly for the other three marked points.  Finally, consider the 
uniformization map $\varphi: \mathbb D \rightarrow \Omega$.  Then taking $n_\circ$ larger if necessary, we assert that for all $n > n_{\circ}$, the distance (in the unit disc) 
between
$\varphi^{-1}(N_{\eta}(A))$ and
$[\varphi^{-1}(N_{4\eta}(A))]^{\text{c}}$
satisfies
\begin{equation}\label{saway}
\text{dist}[\varphi^{-1}(N_{\eta}(A)), 
[\varphi^{-1}(N_{4\eta}(A))]^{\text{c}}]
\gg
n^{-\frac{1}{2}}.
\end{equation}

We first state:

\noindent {\bf Claim.} For $n > n_{\circ}$,
\[ \mbox{dist}(N_{\eta}(A), [B_n, C_n, D_n]) > 0.\]

\bigskip

\noindent \emph{Proof of Claim.}  
 We note that the pre--image of $\partial \Omega_n$
 under uniformization
  has the following property: for $n$ sufficiently large as specified above, consider the pre--image of the boundary element $\varphi^{-1}([A_n, B_n])$.  Then starting at $\varphi^{-1}(A_n)$, once the segment enters $\varphi^{-1}(N_{\eta}(B_n))$, it must hit 
$\varphi^{-1}(B_n)$ before exiting $\varphi^{-1}(N_{4\eta}(B_n))$.

Indeed, if this were not true, then necessarily there would be three or more crossings of the ``annular region'' $\varphi^{-1}(N_{4\eta}(B_n)) \setminus \varphi^{-1}(N_{\eta}(B_n))$.
It is noted that all such crossings -- indeed all
of $\varphi^{-1}(\Omega_{n})$ -- lies within a distance 
of the order $n^{-1/2}$ of $\partial \mathbb D$.
This follows by standard distortion estimates (see e.g., \cite{lawler}, Corollary 3.19 together with Theorem 3.21) and the definition of canonical approximation: each point on $\partial \Omega_n$ is within distance $1/n$ of some point on $\partial \Omega$.  It is further noted, by the final stipulation
 concerning $n_{\circ}$, that the separation scale of the above mentioned ``annular region'' is large compared with 
the distance $n^{-1/2}$.

Envisioning $\partial \Omega$ to be the ``bottom'', consider now a point on the ``topmost'' of these crossings which is well away -- compared with 
$n^{-1/2}$ --
from the lateral boundaries of the annular region and also the pre--image of its associated hexagon.  
Since this point is the pre--image of one on $\partial \Omega_n$, the hexagon in question must intersect $\partial \Omega$ and therefore its pre--image must intersect $\partial \mathbb D$.  However, in order to intersect $\partial \mathbb D$, the pre--image of the hexagon in question must intersect \emph{all} the lower crossings, since our distortion estimate does not permit this pre--image to leave (a 
lower portion of) the annular region.  This 
necessarily implies it passes through the interior of $\Omega_n$, which is impossible for a boundary hexagon.

  The same argument also shows that once $\varphi^{-1}(\partial \Omega_n)$ exits $\varphi^{-1}(N_{4\eta}(B_n))$, it cannot re--enter $\varphi^{-1}(N_{\eta}(B_n))$ and so must be headed towards $\varphi^{-1}(C_n)$ and certainly cannot enter $\varphi^{-1}(N_{\eta}(A))$
since
\[\mbox{dist}(\varphi^{-1}(N_{\eta}(A)), \varphi^{-1}([B, C, D] \cup [N_{4\eta}(B) \cup N_{4\eta}(D) \cup N_{4\eta}(C)])) \gg n^{-1/2}\]
by assumption (by the choice of $\eta$, it is the case that $[B, C, D] \cup [N_{4\eta}(B) \cup N_{4\eta}(D) \cup N_{4\eta}(C)] \subseteq [N_{4\eta}(A)]^c$ from which the previous display follows from Equation \eqref{saway}).  

Altogether we then have that 
$ \mbox{dist}(\varphi^{-1}(N_{\eta}(A)), \varphi^{-1}([B_n, C_n, D_n])) > 0$,
and so the claim follows after applying $\varphi$.
\qed

\vspace{6 pt}

The above claim in fact implies that there exist points $A_n^p \in [A_n, B_n]$ and $A_n^g \in [A_n, D_n]$ such that 
\[ \mbox{dist}(A_n^p, A_n^g) < \frac{1}{n^{1-a_1}}\]
and 
\[ \mbox{dist}(A_n^p, \partial \Omega_n^\square), ~\mbox{dist}(A_n^g, \partial \Omega_n^\square) < \frac{1}{n^{1-a_1}}.\]
Indeed, consider squares of side length $n^{a_1}$ intersecting $\partial \Omega_n$ which share an edge with $\partial \Omega_n^\square$ and have non--trivial intersection with $N_{\eta}(A)$, then since $\partial \Omega_n$ passes through such boxes, we can unambiguously label them as either an $[A_n,B_n]$, an $[A_n,D_n]$ box, or both, and by the claim there are no other possibilities.  Therefore, a pair of such boxes of differing types must be neighbors or there is at least one single box of both types, so we indeed have points $A_n^p, A_n^g$ as claimed.  
Finally, by the stipulation 
$$
\frac{1}{n^{1-a_{1}}} \ll \text{dist}(([A,B]\setminus N_{\eta}(A)), ([D,A]\setminus N_{\eta}(A)))
$$
it is clear that these 
points must lie in $N_{\eta}(A)$.  

Thus
we choose $A_n^\square \in \partial \Omega_n^\square$ to be any point near the $(A_n^p, A_n^g)$ juncture (e.g., the nearest point).
Now consider the scale $n^{a_{3}}$
with $1 > a_{3} > a_{1}$.  We may surround the points
$A_n^{p}, A_n^{g}$ and $A_{n}^{\square}$
with the order of 
$\log_{2} n^{a_{3} - a_{1}}$
disjoint concentric annuli.  These annuli have the property that the \emph{fragment} consisting of its intersection with $\Omega_n$ forms a conduit between some portion of the $[A_n, D_n]$ boundary (which need not be connected) and a similar portion of the $[A_n, B_n]$ boundary.  Moreover, any circuit in the annulus necessarily provides a path which connects these two portions.  Thus the probability of a blue connected path between $[A_n, D_n]$ and $[A_n, B_n]$ within any particular annulus fragment is no less than the probability of a blue circuit in the corresponding full annulus, which is uniformly positive.  So letting $\mathcal A$ denote the event that in at least one of these fragments the desired blue connection occurs, we have 
%
\begin{equation}
\label{XTC}
\mathbb P(\mathcal A) \geq 1 - n^{-a_4}
\end{equation}
for some $a_{4} > 0$.  
Similar constructions may be enacted about the 
$B_{n}, B_{n}^{\square}; \dots ;  D_{n}, D_{n}^{\square}$
pairs leading, ultimately, to the events
$\mathcal B, \dots, \mathcal D$ (which are analogous to 
$\mathcal A$) with estimates on their probabilities as in Eq.(\ref{XTC}).
For future reference, we denote by
$\mathcal E$ the event $\mathcal A \cap \dots \cap \mathcal D$ and so $\mathbb P(\mathcal E) \gtrsim 1 - n^{a_4}$ (by the FKG inequality).

We are now in a position to verify that
$|S_{n}(A_{n}) - S_{n}^\square(A_{n}^{\square})|$
obeys the stated power law estimate.  
Indeed, the $C$--component
of both functions vanish identically
while the differences between the other two components
amount to comparisons of crossing probabilities 
on the ``topological'' rectangles
$[A_{n}, B_{n}, C_{n}, D_{n}]$
verses
$[A_{n}^{\square}, B_{n}^{\square}, C_{n}^{\square}, D_{n}^{\square}]$.  There are two crossing events contributing to the (complex) function $S_n(A_n)$ (and similarly for $S_n^\square(A_n^\square)$) but since the arguments are identical, it is sufficient to treat one such crossing event.  Thus we denote by
$\mathbb K_{n}$
the event of a crossing in $\Omega_{n}$
by a blue path 
between the 
$[A_{n}, D_{n}]$
and
$[B_{n}, C_{n}]$
boundaries (the event contributing to $S_B(A_n)$)
and similarly for the event
$\mathbb K_{n}^{\square}$
for a blue path in $\Omega_{n}^{\square}$.
It is sufficient to show that
$|\mathbb P(\mathbb K_{n}^{\square}) - \mathbb P(\mathbb K_{n})|$
has an estimate of the stated form. 

The greater portion of the following is rather standard in the context of 2D percolation theory so we shall be succinct.  The idea is to first ``seal'' together e.g., $A_n$ and $A_n^\square$ (and similarly for $B, C, D$) by circuits and then perform a continuation of crossings argument.

Without loss of generality we may assume that $S_B^\square(A_n^\square) > S_B(A_n)$ since otherwise the $S_D$ functions would satisfy this inequality and we may work with $S_D$ instead.
For ease of exposition, let us envision
$[A_{n}, B_{n}]$ and $[A_{n}^{\square}, B_{n}^{\square}]$
as the ``bottom'' boundaries and the
$D, C$ pairs as being on the ``top''.

Let $\Gamma$ denote a crossing between
$[A_{n}^{\square}, D_{n}^{\square}]$ and $[B_n^\square, C_n^\square]$ within
$\Omega_{n}^{\square}$ and let
$\Gamma_{\mathbb K_{n}^{\square}} \in \mathbb K_{n}^{\square}$
denote the event
that $\Gamma$ is the ``lowest'' (meaning
$[A_{n}^{\square}, B_{n}^{\square}]$--most) crossing.
These events form a disjoint partition so that 
$\mathbb P(\mathbb K_n^\square) = \sum_\Gamma ~\mathbb P(\mathbb K_{n}^\square
\mid \Gamma_{\mathbb K_{n}^{\square}}) \cdot \mathbb P(\Gamma_{\mathbb K_n^\square})$.  From previous discussions concerning $A_n^p, A_n^g$, we have that $\mathbb P(\mathcal E) \geq 1 - n^{-a_4}$, which we remind the reader, means that with the stated probability, these crossings do not go into any ``corners'' and hence there is ``ample space'' to construct a continuation.  

So let $a_{5} > a_{1}$ denote another constant which is less than unity (recall that in microscopic units, $\mbox{dist}(\partial \Omega_n^\square, \partial \Omega_n) \leq n^{a_1}$).
Then, to within tolerable error estimate (by the Russo--Seymour--Welsh estimates)
it is sufficient to consider only the crossings 
$\Gamma$
with right endpoint a distance in excess of
$n^{a_{5}}$ away from
$C_{n}^{\square}$ and with left endpoint similarly
separated from
$D_{n}^{\square}$.

Let $\Gamma_{D}$ and $\Gamma_{C}$
denote these left and right endpoints of $\Gamma$, respectively.
Consider a sequence of intercalated 
annuli starting at the scale $n^{a_{1}}$ -- or, if necessary, in slight excess -- and ending at scale
$n^{a_{5}}$ (where ostensibly they might run aground at
$C_{n}^{\square}$) around $\Gamma_C$.  A similar sequence should be considered
on the left.  Focusing on the right, it is clear that each
such annulus provides a conduit
between $\Gamma$ and 
$\partial \Omega_{n}$
that runs through
the $[B_{n}^{\square}, C_{n}^{\square}]$ 
boundary of
$\Omega_{n}^{\square}$.
Let $\bar \gamma_r$ denote an occupied blue circuit in one of these annuli
and similarly for $\bar \gamma_\ell$ on the left.  

The blue circuit $\bar \gamma_r$ must intersect $\Gamma$ and, since e.g., $\Gamma_C$ is at least $n^{a_5}$ away from $A_n^\square, D_n^\square$, these circuits must end on the $[D_n^\square, A_n^\square]$ boundary so that the portion of the circuit above $\Gamma$ forms a continuation to $\partial \Omega_n$; similar results hold for $\Gamma_D$ and $\bar \gamma_\ell$ and the crossing continuation argument is complete. As discussed before,  we may repeat the argument for the other crossing event contributing to the $S$--functions, so we now have that $|S_n(A_n) - S_n^\square(A_n^\square)| \leq n^{-a_2}$ for some $a_2 > 0$, concluding the first half of the theorem.
  
The second claim of this theorem, 
concerning the conformal maps $H_{n}(A_{n})$
and $H_{n}^{\square}(A_{n}^{\square})$
in fact follows readily from the \textit{arguments}
of the first portion.  In particular, we claim that the estimate on the difference can be acquired by an identical sequence of steps
by the realization of the fact that the $S$--function for a given percolative domain which is the canonical approximation to a conformal rectangle
converges to the conformal map of said domain to 
$\mathbb T$ (\cite{stas_perc} , \cite{beffara}, \cite{BCL2}).

Thus, while seemingly a bit peculiar, 
there is no reason why we may not consider 
$\Omega_{n}$ to be a fixed \emph{continuum} domain
and, e.g., for $N \geq n$, the domain
$\Omega_{n,N}$ to be its canonical approximation 
for a percolation problem at scale 
$N^{-1}$.  Similarly for
$\Omega_{n,N}^{\square}$.  
Of course here we underscore that 
e.g., 
$A_{n}^{\square}, \dots D_{n}^{\square}$ are regarded as fixed
(continuum) marked points which have their 
own canonical approximates 
$A_{n,N}^{\square}, \dots D_{n,N}^{\square}$
but there is no immediately useful relationship between them and the approximates
$A_{n,N} \dots D_{n,N}$.

It is now claimed that uniformly in $N$, with $N \geq n$ and
$n$ sufficiently large the entirety of the previous argument 
can be transcribed \emph{mutatis mutantis}
for the percolation problems on
$\Omega_{n,N}$ and $\Omega_{n,N}^{\square}$.  
Indeed, once all points were located,
the seminal ingredients all concerned
(partial) circuits in (partial) annuli 
and/or rectangular crossings
of uniformly bounded aspect ratios and dimensions not smaller than 
$n^{-1}$.  All such events enjoy uniform bounds
away from $0$ or $1$ (as appropriate)
which do not depend on the scale
and therefore apply 
to the percolation problems on
$\Omega_{n,N}$ and $\Omega_{n,N}^{\square}$.
We thus may state without further ado
that for all $N > n$ (and $n$ sufficiently large)
\begin{equation}
|S_{n,N}(A_{n,N}) - S_{n,N}^{\square}(A_{n,N}^{\square})|
\lesssim \frac{1}{n^{a_{2}}}
\end{equation}
and so $|H_n(A_n) - H_n^\square(A_n^\square)| \lesssim n^{-a_2}$ as well.  

Finally, since the relationship between $\Omega_n$ and $\Omega$ is the same as that between $\Omega_n^\square$ and $\Omega_n$ (both $\Omega_n$, $\Omega_n^\square$ are inner domains obtained by the union of shapes (squares or hexagons) of scale an inverse power of $n$ from $\Omega$, $\Omega_n$, respectively) the same continuum percolation argument as above gives the estimate that $|H_n(A_n) - H(A)| \leq n^{-a_2}$.

\qed

We remark that the idea of uniform estimates leading to ``continuum percolation'' statements will be used on other occasions in this paper.

\subsection{The Cauchy--Integral Extension}
\label{YTV}

We will now consider the Cauchy--integral version of 
the function $S_n^\square$.  
Ostensibly this is defined on the full
$\Omega_{n}^{\square}$ however as mentioned in the introduction to this section, its major 
r\^ole will be played on the subdomain $\Omega_n^\vardiamond$ which will emerge shortly.
\begin{lemma}
\label{ci}
Let $\Omega_n^\square$ and $S_n^\square$ be as in Proposition \ref{pr} so that 
\[|\partial \Omega_n^\square| \leq n^{\alpha(1-a_1)},\]
where $M(\partial \Omega) < 1 + \alpha$. 
For $z \in \Omega_n^\square$ (with the latter regarded as a continuum subdomain of the plane) let 
\begin{equation}
\label{FS}
F_n^{\square}(z) = \frac{1}{2\pi i} \oint_{\partial \Omega_n^\square} \frac{S_n^\square(\zeta)}{\zeta - z}~d\zeta.  
\end{equation}
Then for $a_1$ sufficiently close to 1 there exists some $\beta > 0$ and some $a_5 > 0$ such that for all $z \in \Omega_n^\square$ (meaning lying on edges and sites of $\Omega_n^\square$) with $\mbox{dist}(z, \partial \Omega_n^\square) > n^{-a_5}$ it is the case that
$$
\left|S_n^\square(z) - F_n^{\square}(z) \right| \lesssim n^{-\beta}.
$$
\end{lemma}
The proof of this lemma is postponed until Section \ref{CIEX} and we remark that while $S_n^\square$ is only defined on vertices of hexagons \emph{a priori}, it can be easily interpolated to be defined on all edges, as discussed in Section \ref{KBB}.  We will now proceed to demonstrate that $F_n^\square$ is conformal in a subdomain of $\Omega_n^\square$.  Let us first define a slightly smaller domain:

\begin{defn}
Let $\Omega_n^\square$ be as described.  Let $a_5 > 0$ be as in Lemma \ref{ci} and define, for temporary use,
\[ \Omega_n^\oo := \{ z \in \Omega_n^\square: \mbox{dist}(z, \partial \Omega_n^\square) \geq n^{-a_5}\}.\]
\end{defn}

We immediately have the following:

\begin{prop}\label{t1}
For $n$ sufficiently large, there exists some $\beta > a_3 > 0$ (with $\beta$ as in Lemma \ref{ci}) such that 
\[ d_{\rm sup}(F_n^\square(\partial \Omega_n^\oo), \partial \mathbb T) \lesssim n^{-a_3}.\]
Here $d_{\rm sup}$ denotes the \emph{supremum distance} between curves, i.e., 
\[ d_{\rm sup}(\gamma_1, \gamma_2) = \inf_{\varphi_1, \varphi_2} \sup_t |\gamma_1(\varphi_1(t)) - \gamma_2(\varphi_2(t))|,\]
where the infimum is over all possible parameterizations.  
\end{prop}
\begin{proof}
Let us first re--emphasize that $S_n^\square$ maps $\partial \Omega_n^\square$ to $\partial \mathbb T$.  This is in fact fairly well known (see e.g., \cite{beffara} or \cite{BCL2}, Theorem 5.5) but a quick summary proceeds as follows: by construction $S_n^\square$ is continuous on $\partial \Omega_n^\square$ and e.g., takes the form $\lambda \tau + (1-\lambda) \tau^2$ on one of the boundary segments, where $\lambda$ represents a crossing probability which increases monotonically -- and continuously -- from 0 to 1 as we progress along the relevant boundary piece.  Similar statements hold for the other two boundary segments.

Now by Lemma \ref{ci},  $F_n^{\square}(z)$ is at most the order $n^{-\beta}$ 
away from $S_n^\square(z)$ for any $z \in \partial \Omega_n^{\oo}$, so the curve $F_n^{\square}(\partial \Omega_n^{\oo})$ is in fact also that close to $S_n^\square(\partial \Omega_n^\oo)$ in the supremum norm.  Finally, by the H\"older continuity of $S_n^\square$ up to $\partial \Omega_n^\square$ (see Proposition \ref{sfp}) and the fact that $\partial \Omega_n^\oo$ is a distance which is an inverse power  of $n$ to $\partial \Omega_n^\square$, it follows that $S_n^\square(\partial \Omega_n^\oo)$ is also close to $\partial \mathbb T$ and the stated bound emerges.
\end{proof}

Equipped with this proposition, we can now introduce the domain $\Omega_n^\vardiamond$:

\begin{defn}\label{t2}
Let $a_4 > 0$ be such that $\beta > a_3  > a_4$ (with $a_3 > 0$ as in Proposition \ref{t1}) and let us denote by
$$
\mathbb T^{\vardiamond} = (1- n^{-a_4}) \cdot \mathbb T 
$$
the uniformly shrunken version of $\mathbb T$.  Finally, let 
$$
\Omega_n^{\vardiamond} := (F_n^{\square})^{-1} (\mathbb T^\vardiamond)
$$
and denote by $(B_n^\vardiamond, C_n^\vardiamond, D_n^\vardiamond)$ the preimage of $(1 - n^{-a_4}) \cdot (1, \tau, \tau^2)$ under $F_n^\square$.
\end{defn}

The purpose of introducing $\Omega_n^\vardiamond$ is illustrated in the next lemma:

\begin{lemma}
\label{cdmd}
Let $F_n^\square$ and $\Omega_n^\vardiamond$, etc., be as described.  Then $F_n^\square$ is conformal in $\Omega_n^\vardiamond$.  Next let $H_n^\vardiamond: \Omega_n^\vardiamond \longrightarrow \mathbb T$ be the conformal map which maps $(B_n^\vardiamond, C_n^\vardiamond, D_n^\vardiamond)$ to $(1, \tau, \tau^2)$.  Then for all $z \in \overline{\Omega_n^\vardiamond}$, 
\[ |F_n^\square(z) - H_n^\vardiamond(z)| \lesssim n^{-a_4}.\]
\end{lemma}

\begin{proof}
Since $F_n^\square$ is manifestly holomorphic in order to deduce conformality it is only necessary to check that it is 1--to--1.
Let $K_n: = F_n^{\square}(\partial \Omega_n^{\oo})$ and let us start with the following observation on the winding of $K_n$:

\bigskip

\noindent \textbf{Claim.}
If $w \in \mathbb T^{\vardiamond}$, then the winding of $K_n$ around $w$ is equal to one:
\[W(K_n, w) = \frac{1}{2\pi i} \int_{K_n} \frac{dz}{z - w} = 1.\]

\bigskip

\noindent \textit{Proof of Claim.}
The result is elementary and is, in essence, Rouch\'e's Theorem so we shall be succinct and somewhat informal.  Foremost, 
by continuity, the winding is constant for any 
$w \in \mathbb T^\vardiamond$.  (This is easily proved using the displayed formula and the facts that the winding is integer valued and that $K_{n}$ is rectifiable.)  
Clearly, since
$\partial \mathbb T$ and $K_{n}$ are close in the supremum norm, it follows, by construction 
that 
$\partial \mathbb T^\vardiamond$ and $K_{n}$ are also close in this norm.  

Let $z_{K}(t)$ and $z_{\vardiamond}(t)$, $0 \leq t \leq 1$ denote counterclockwise moving parameterizations of $K_n$ and $\partial \mathbb T^\vardiamond$ that are uniformly close.  For $z_{\vardiamond}$, this starts and ends on the positive real axis
and we let $\theta_{\vardiamond}(t)$ denote the evolving argument of 
$z_{\vardiamond}(t)$ (with respect to the origin as usual): $0 \leq \theta_{\vardiamond}(t) \leq 2\pi$.
We similarly define $\theta_{K}(t)$: in this case, we stipulate that $|\theta_{K}(0)|$
is as \textit{small} as possible -- and thus approximately zero -- but of course $\theta_{K}(t)$
evolves continuously with $z_{K}(t)$ and therefore ostensibly could lie anywhere in
$(-\infty, \infty)$.
But 
$|z_{\vardiamond}(t)|$ and $|z_{K}(t)|$ are both of order unity (and in particular not small) and they are close to each other.  So it follows that 
$|\theta_{\vardiamond}(t) - \theta_{K}(t)|$
must be uniformly small, e.g., within some $\vartheta$ with
$0 < \vartheta \ll \pi$ for all $t\in [0,1]$.   
Now, since 
$\theta_{\vardiamond}(1) - \theta_{\vardiamond}(0) = 2\pi$,
we have
$$
|W(K_{n}, 0) -1| = 
\left|
\frac{\theta_{K}(1) - \theta_{K}(0)}{2\pi} - 
\frac{\theta_{\vardiamond}(1) - \theta_{\vardiamond}(0)}{2\pi}\right|
\leq\frac{2\vartheta}{2\pi} \ll 1,
$$
so we are forced to conclude that $W(K_n, 0) = 1$ by the integer--valued property of winding.  The preceding claim has been established.
\qed

The above implies that $F_n^{\square}$ is in fact 1-1 in $\Omega_n^\vardiamond$: from Definition \ref{t2} we see that $a_4$ is chosen so that (for $n$ sufficiently large) $n^{-a_4}$ is large compared with $n^{-a_3}$ (from the conclusion of Proposition \ref{t1}) so that $K_n$ (which is clearly a continuous and possibly self--intersecting curve) lies outside $\mathbb T^{\vardiamond}$.  Now fix some point 
$\xi \in \Omega_n^{\vardiamond}$
and consider the function $h_\xi(z) := F_n^{\square}(z) - F_n^{\square} (\xi)$.  Next parametrizing $\partial \Omega_n^\square := \gamma$
as $\gamma: [0, 1] \rightarrow \mathbb C$, noting that $F_n^{\square}(\xi) \in \mathbb T^{\vardiamond}$ and using the chain rule we have that 
\[ \begin{split}
1 = W(K_n, F_n^{\square}(\xi))
&= \frac{1}{2\pi i} \oint_{F_n^{\square} \circ \gamma} ~
\frac{1}{\zeta - F_n^{\square}(\xi)}~d\zeta\\
&= \frac{1}{2\pi i} \int_0^1 \frac{(F_n^{\square})'(\gamma(t)) \gamma'(t)}{F_n^{\square}(\gamma(t)) - F_n^{\square}(\xi)}~dt
= \frac{1}{2\pi i} \oint_\gamma h_\xi'/h_\xi~dz.  
\end{split}\]
By the argument principle, the last quantity is equal to the number of zeros of $h_\xi$ in the region enclosed by 
$\gamma$, i.e., in $\Omega_{n}^{\square}$.  The desired 1--to--1 property is established.

We have now that $F_n^{\square} \mid_{\Omega_n^{\vardiamond}}$ is analytic and maps $\Omega_n^{\vardiamond}$ in a one--to--one fashion onto $\mathbb T^{\vardiamond}$.  Therefore $F_n^{\square} \mid_{\Omega_n^{\vardiamond}}$ is \emph{the} conformal map from $\Omega_n^{\vardiamond}$ to $\mathbb T^{\vardiamond}$ (mapping $B_n^{\vardiamond}, C_n^{\vardiamond}, D_n^{\vardiamond}$ to $(1 - n^{-a_4}) \cdot (1, \tau, \tau^2)$, the corresponding vertices of $\mathbb T^\vardiamond$).  Thus by uniqueness of conformal maps we have that $H_n^\vardiamond = \frac{1}{1 - n^{-a_4}} \cdot \left(F_n^\square \mid_{\Omega_n^\vardiamond}\right)$ and the stated estimate immediately follows.
\end{proof}

\subsection{Harris Systems}

We will now introduce the notion of \emph{Harris systems}; proofs will be postponed until Section \ref{SBT}.

\begin{theorem}[Harris Systems.]\label{hreg}
Let $\Omega_n^\square \subseteq \Omega$ be as described with marked boundary points (prime ends) $A, B, C, D \in \partial \Omega$ and let $z$ be an arbitrary point on $\partial \Omega_n$. Further, let $2\Delta$ denote the supremum of the side--length
of all circles contained in $\Omega$, and let 
$\mathcal D_{\Delta}$ denote a circle of side $\Delta$ with the same center as a circle for which the supremum is realized.  

Then there exists some $\Gamma > 0$ such that for all $n \geq n(\Omega)$ sufficiently large, the following holds:  around each boundary point
$z \in \partial \Omega_n^\square$ there is a nested sequence of at least $\Gamma \cdot \log n$ neighborhoods
the boundaries of which are segments (lattice paths) separating 
$z$ from $\mathcal D_{\Delta}$.  We call this sequence of segments 
the \emph{Harris system} stationed at $z$.  The regions between these segments (inside $\Omega_n^\square$)
are called Harris ring fragments (or just Harris rings).  

Further, there exists some $0 < \vartheta < 1/2$ such that  in each Harris ring, the probability of a blue path separating $z$ from $\mathcal D_{\Delta}$ is uniformly bounded from below by $\vartheta$.

Also, let $J$ denote the $d_\infty$--distance (see the definitions in Subsection \ref{dfss}) between successive segments forming a Harris ring -- of course, $J$ depends on the particulars of the ring under consideration -- and let $B > 0$ be such that the probability of a hard way crossing of a $B$ by $1$ topological rectangle (in both yellow and blue; see Proposition \ref{btt}) is less than $\vartheta^2$.  The following properties hold:

\begin{enumerate}

\item\label{hbox} for $r > 0$ sufficiently large (particularly, $2^{-r} < B^{-1}$) the Harris rings can be tiled with boxes of scale $2^{-2r} \cdot J$ and there is an aggregation of full boxes (unobstructed by the boundary of the domain) which connect the segments forming the Harris rings;

\item \label{hhss} successive segments $Y, Y_Q$ satisfy 
\[ B^{-1} \cdot J \leq \|Y\|_\infty \leq 2^{2r+1}(\kappa B)  \cdot J,~~~B^{-1} \cdot J  \leq \|Y_Q\|_\infty \leq \kappa B \cdot J,\]
where e.g., $\|Y\|_\infty$ denotes the diameter of the segment $Y$;

\item\label{hseal2} let $a$ be a point in the Harris system centered at $A_n^\square$ such that the number of Harris rings between $a$ and $\mathcal D_\Delta$ is of order $\log n$. 
Let $\mathcal A(a)$ denote the event of a blue (or yellow) path surrounding both $a$ and $A_n^\square$ with endpoints on $[A_n^\square, B_n^\square]$ and $[D_n^\square, A_n^\square]$.  Then there exists some constant $\lambda > 0$ such that 
\[\mathbb P(\mathcal A(a)) \geq 1 - n^{-\lambda}; \]
similar estimates hold at the points $B_n^\square, C_n^\square, D_n^\square$ and hence the estimate also holds for the intersected event, by FKG type inequalities
(or just independence); 

\item\label{hunif} finally, all estimates are uniform in lattice spacing in the sense of considering $\Omega_n^\square$ to be a fixed domain and performing percolation at scale $N^{-1}$.

\end{enumerate}
\end{theorem}

\begin{proof}
The constructions required for the establishment of this theorem is the content of Section \ref{SBT}.  
That there exists at least of order $\log n$ such neighborhoods follows from the fact that each point on $\partial \Omega_n^\square$ is a distance at least $\Delta$ from $\mathcal D_\Delta$ and so proceeding ``directly'' towards $\mathcal D_\Delta$ and increasing the scale each time by the maximum allowed while fixing the aspect ratio already leads to of order $\log n$ such neighborhoods.

As for the various statements, items 1, 2 are consequences of the full Harris construction (see Subsection \ref{fcont}); item 3 follows from Lemma \ref{ped} and item 4 is a direct consequence of the fact that at criticality, crossing probabilities of rectangles with bounded aspect ratios remain bounded away from 0 and 1 uniformly in lattice spacing.  
\end{proof}

Let us start with the quantification of the ``distance'' between the corresponding marked points of $\Omega_n^\vardiamond$ and $\Omega_n^\square$: 

\begin{prop}\label{pts}
$B_n^\vardiamond$ is in the Harris system stationed at $B_n^\square$.  Moreover, there exists some $\kappa > 0$ such that there are at least $\kappa \cdot \log n$ Harris rings from this Harris system which enclose $B_n^\vardiamond$.  Similar statements hold for $C_n^\vardiamond, D_n^\vardiamond$.
\end{prop}
\begin{proof}
The argument that $B_n^\vardiamond$ is indeed in the Harris system stationed at $B_n^\square$ and the argument that there are many Harris rings enclosing $B_n^\vardiamond$ are essentially the same.

\begin{figure}[ ]
\vspace{.15 in}
\centering
\includegraphics[scale=.27]{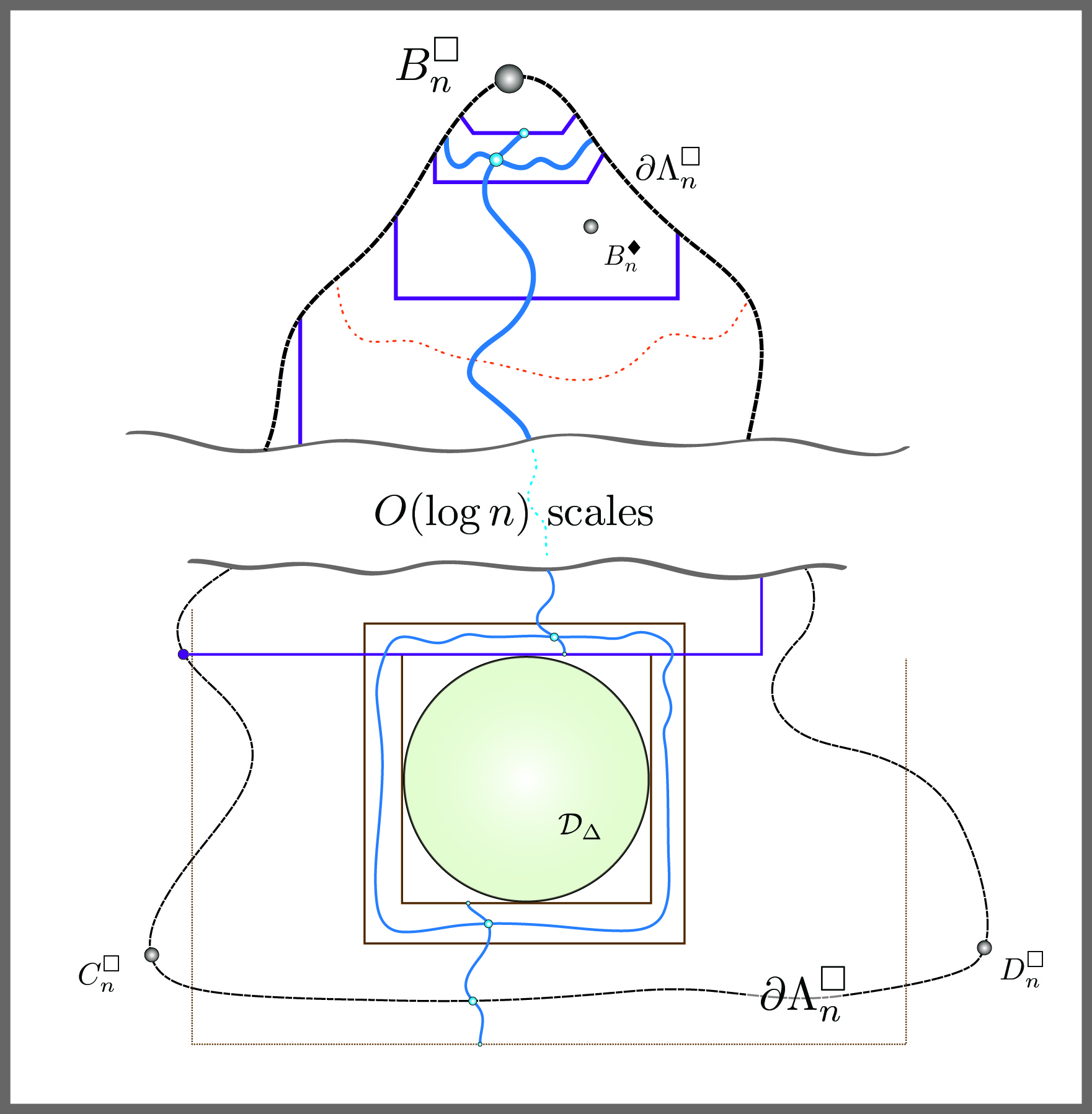}
\begin{changemargin}{.55in}{.55in}
\caption{{\footnotesize Preliminary circuit element ``inside'' $B_{n}^{\vardiamond}$ plus connection to 
$[C_n^{\square},D_n^{\square}]$ boundary -- first via intermediate scales then to the vicinity of $\mathcal D_{\Delta}$ and the rest by large--scale events -- prevents the occurrence of of the yellow circuit described.  If there are fewer than $\gamma\log n$ intermediate scales then this circuit would have substantive probability.}}
\label{XAB}
\end{changemargin}
\vspace{-.3 in}
\end{figure}

First we have that by Lemma \ref{ci} and Definition \ref{t2} that e.g., $|S_B^\square(B_n^\vardiamond)| \gtrsim 1 - n^{-a_4} - n^{-\beta} \gtrsim 1 - n^{-a_4}$. (Recall that $\beta > a_3 > a_4$ and $S_B^\square(B_n^\vardiamond)$ is the probability of a yellow crossing from $(B_n^\square, C_n^\square)$ to $(D_n^\square, B_n^\square)$ separating $B_n^\vardiamond$ from $(C_n^\square, D_n^\square)$.)  On the other hand, let us consider the ``last'' Harris ring separating $B_n^\vardiamond$ from $B_n^\square$ which forms a conduit between $[D_n^\square, B_n^\square]$ and $[B_n^\square, C_n^\square]$, c.f., Theorem \ref{hreg}, item \ref{hseal2}; we may enforce a crossing in this conduit with probability $\vartheta$ (as in Theorem \ref{hreg}) and then via a box construction and a ``large scale'' crossing (as appears below in the proof of Lemma \ref{dp}) the said crossing can be connected to $[C_n^\square, D_n^\square]$ in blue.  This construction procedure is illustrated in Figure \ref{XAB}.  Therefore, if the number of Harris rings enclosing $B_n^\vardiamond$ were \emph{less} than $\gamma \cdot \log n$, then there would be some $V > 0$ such that the journey from the vicinity of $B_n^\vardiamond$ to $[C_n^\square, D_n^\square]$ can occur at a probabilistic cost in excess of $n^{-\gamma V}$. 

Since such a blue connection renders a yellow version of the event $\mathbb S_B^\square(B_n^\vardiamond)$ impossible, we conclude  that there must be more than $a_4/V$ Harris rings enclosing $B_n^\vardiamond$, for $n$ sufficiently large.  Similar arguments yield the result also for $C_n^\vardiamond, D_n^\vardiamond$.  
\end{proof}

More generally, we have the following description of the distance between $\partial \Omega_n^\square$ and $\partial \Omega_n^\vardiamond$:

\begin{lemma}\label{dp}
Let $s \in \partial \Omega_n^\vardiamond$ and $z \equiv z(s)$ the point on $\partial \Omega_n^\square$ which is closest to $s$ (in the Euclidean distance).  Then there exists some $\kappa > 0$ such that in the Harris system stationed at $z$, there are at least $\kappa \cdot \log n$ Harris rings that enclose $s$.
\end{lemma}

\begin{figure}[ ]
\centering
\includegraphics[scale=.27]{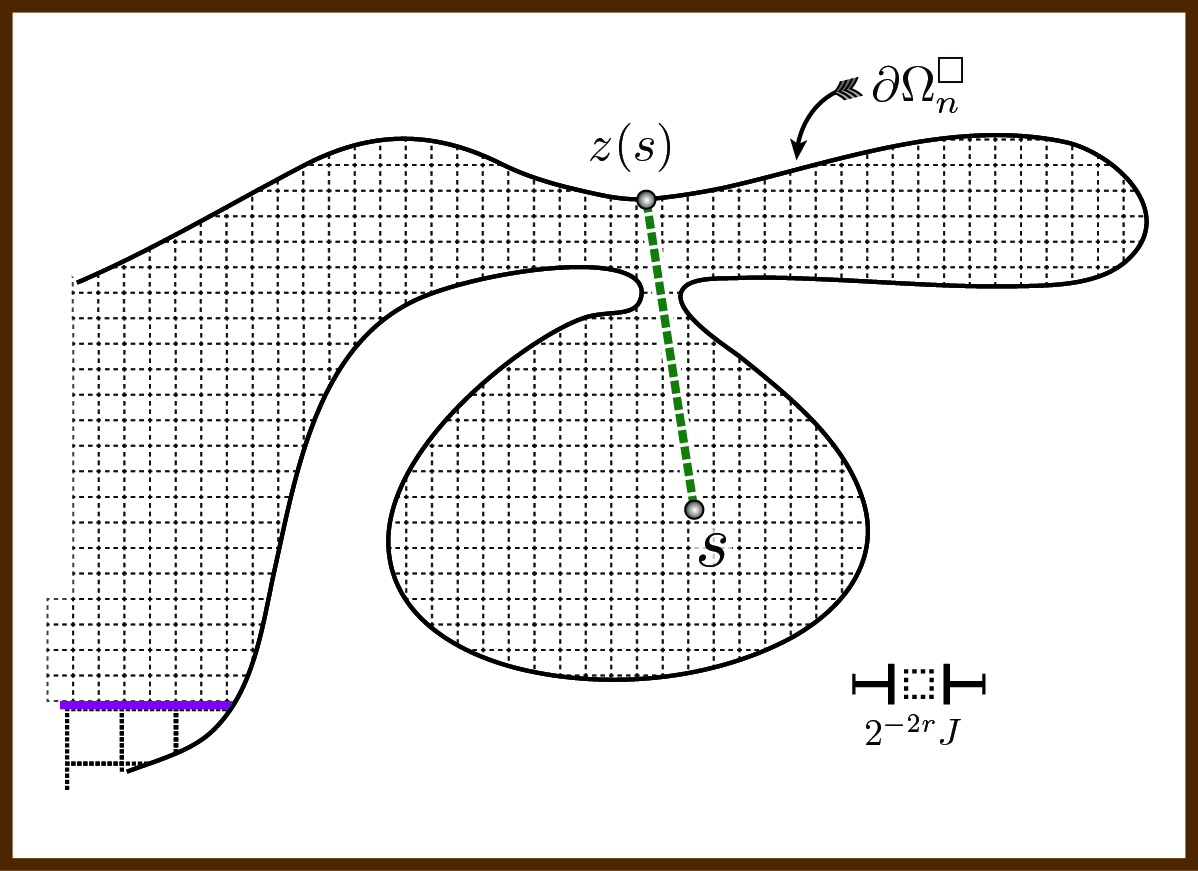}
\begin{changemargin}{.55in}{.55in}
\caption{{\footnotesize Case 3 of Lemma \ref{dp}:  if the box containing $s$ is separated from the main percolating cluster of full boxes (of scale $2^{-2r} \cdot J$) associated with its ring by a partial box then, necessarily, $z$ could not be the point on 
$\partial \Omega_n^\square$ which is closest to $s$.}}
\label{VFA}
\end{changemargin}
\vspace{-.3 in}
\end{figure}

\begin{proof}
Let us set $\lambda := \mbox{dist}(s, z)$.  First, logically speaking, we must rule out the possibility that $s$ is outside the Harris system stationed at $z$ altogether: if this were true, then it would  imply that $\mbox{dist}(s, \partial \Omega_n^\square) = \lambda > \frac{1}{2}\Delta$ (since Harris circuits plug into $\partial \Omega_n^\square$ the point $s$ can only be outside the Harris system at $z$ altogether if it is ``beyond'' the last Harris segment which parallels $\partial \mathcal D_\Delta$; see Theorem \ref{hreg}) which then readily implies that all of the $S$--functions are of order unity: indeed, in this case $S_B^\square(s), S_C^\square(s)$ and $S_D^\square(s)$ can all be bounded from below by large scale events of order unity  (consider e.g., the crossing of a suitable annulus whose aspect ratio is order unity with $s$ on the boundary of the inner square and the outer square touching $\partial \Omega_n^\square$ (from inside $\Omega_n^\square$) together with yet another couple of crossings from the inner square of this annulus to a larger rectangle which encloses all of $\Omega_n^\square$) which would place $s$ well away from the boundary of $\Omega_n^\vardiamond$ by Definition \ref{t2} and Lemma \ref{ci}.  Thus $s$ is in a Harris ring of $z$.

If the separation -- measured in number of Harris rings -- between $s$ and $\mathcal D_\Delta$ is not so large, then we will show that $|S_n^\square(z)|$ is larger than a small inverse power of $n$.  We will accomplish this by constructing configurations which lead to the occurrence of all three events corresponding to $S_B^\square, S_C^\square, S_D^\square$ with sufficiently large probability.  To this end we will make detailed use of the Harris system. 

Let $J$ denote the separation distance of the Harris segments which form the ring containing $s$ and let $r > 0$ be as given in Theorem \ref{hreg}.
Now note that if the statement of the lemma were false, then there would be an abundance of Harris rings separating $z$ from $s$, which will enable us to construct a path ``beneath'' $s$ to yield the events $\mathbb S_B^\square, \mathbb S_C^\square, \mathbb S_D^\square$.  Consider the boxes of size $2^{-2r} \cdot J$ which grid the ring containing $s$.  Let us observe that there are three cases: 1) the \emph{main} type, $s$ is contained in a full box which is connected to the cluster which percolates through the ring (see Theorem \ref{hreg}, item \ref{hbox}); 2) the \emph{partial} type, meaning that $s$ is in a partial box, i.e., a box intersected by $\partial \Omega_n^\square$; 3) $s$ is in a full box which is separated from the cluster of main types of percolating boxes by a partial box.  

Let us rule out the possibility of 2) and 3).  Case 2) is impossible since it implies that $\mbox{dist}(s, z) = \mbox{dist}(s, \partial \Omega_n^\square) \leq 2^{-2r} \cdot J$ which (see Theorem \ref{hreg}, items \ref{hbox} and \ref{hhss}) necessarily implies that $z$ and $s$ are in the \textit{same} ring.  But, supposing they do reside in the same ring then with probability in excess of (some constant times) $1 - n^{-\Gamma}$, with $\Gamma$ as in Theorem \ref{hreg}, the occurrence or not of the events contributing to $S_B^\square, S_C^\square, S_D^\square$ would be the same for both $s$ and $z$ (c.f., the proof of Proposition \ref{s2} below).  Then by Lemma \ref{ci} and Definition \ref{t2}, it would the case that $|S_n^\square(z) - S_n^\square(s)| \gtrsim n^{-a_4} - n^{-\beta}$, which is a contradiction if $a_4, \beta$ are appropriately chosen relative to $\Gamma$. 

Similar reasoning shows that 3) is also not possible: indeed, since $z$ is the closest point to $s$, $z$ and $s$ must lie along a straight line segment which lies in $\Omega_n^\square$ and this segment must pass through the partial box in question (i.e., the ``bottleneck''; we emphasize here that we are considering the Harris system centered \emph{at} $z$) which separates $s$ from \emph{the} connected component of boxes which percolate through the ring.  From previous considerations regarding $2^{-2r}J$ (the scale of the boxes) versus $\mbox{dist}(s, z)$, it is clear that there is a point on $\partial \Omega_n^\square$ within this partial box which is closer to $s$ than $z$, a contradiction.  These considerations are illustrated in Figure \ref{VFA}.

Thus, we find $s$ in the main percolating component of boxes.  For convenience, we focus on the sub--case where the box containing $s$ is separated from 
$\partial\Omega_{n}^{\square}$
by at least one layer of full boxes.  Indeed, the complementary sub--cases are easily handled by arguments similar to those which dispensed with cases 2) and 3).

We shall now proceed to construct, essentially by hand, any of the events 
$\mathbb S_{B}^{\square}(s)$, 
$\mathbb S_{C}^{\square}(s)$
or $\mathbb S_{D}^{\square}(s)$
corresponding to the functions $S_B^\square, S_C^\square, S_D^\square$, respectively, with ``unacceptably large'' probability.

It is understood that the constructions that follow utilize the main body of boxes percolating through a given Harris ring fragment, as detailed in Theorem \ref{hreg}, item \ref{hbox}.  Ultimately we will be constructing two (disjoint) paths.
E.g., for the $\mathbb S_{B}^{\square}(s)$ event, one path from the vicinity of $s$ to the $[B^\square_n, C^\square_n]$ boundary and the other from the vicinity of $s$ to the $[D^\square_n, B^\square_n]$ boundary.  While not strictly necessary, it is slightly more convenient to construct the ``bulk'' of both paths at once.  Therefore, we shall undertake a 
\textit{double bond} construction.
For further convenience, we will base our construction on
$3\times 1$ bond events which will be described in the next paragraph.

We remark, again, that arguments of this sort have appeared before, e.g., at least as far back as \cite{accfr}, so we will be succinct in our descriptions.  The events are described as follows:  let us assume, for ease of exposition, that three neighboring boxes form a horizontal 
$3\times 1$ rectangle.  The \textit{bond event} -- in yellow -- would then consist
of two disjoint left--right yellow crossings of the 
$3\times 1$ rectangle together with two disjoint top--bottom yellow crossings in each of the outer two squares, 
as is illustrated in Figure \ref{WLM},
it is seen that if a pair of such rectangles overlap on an end--square, and the bond event occurs for both of them, then, regardless of the orientations, there are two disjoint yellow paths which transmit
from the beginning of one to the end of the other.  I.e., these ``bonds'' have the same connectivity properties as the bonds of
$\mathbb Z^{2}$ and provide us with double paths.  

Starting with the square containing $s$ we may suppose there is (or construct) a yellow ring in the eight boxes immediately surrounding and encircling this square.  Via the bond events just described, we connect this encircling ring  to the outward boundary of the Harris annulus to which $s$ belongs.  
Each of these events -- which are positively correlated --
incurs a certain probabilistic cost.  However, 
it is observed, with emphasis, that since the \emph{relative} scales of the Harris ring and the bonds used in the construction are fixed 
independent of the actual scale, the cost may be bounded by a number independent of the \emph{actual} scale.  

Similarly, we may use the bonds to acquire a double path across the next (outward) ring and the two double paths may be connected to form a continuing double path by explicit use of a ``patch'' consisting of the smaller of the two bond types.  Again, since the ratio of scales of (boxes of) successive Harris rings are uniformly bounded above and below, the probabilistic cost 
does not depend on the actual scale.  
The procedure of double crossing via bond events and patches can be continued till the boundary of $\mathcal D_{\Delta}$ is reached; thereupon, 
treating $\mathcal D_{\Delta}$ and its vicinity as an annulus in its own right, the two paths can be connected to separate boundaries at an additional cost of order unity.  

Now let us assume for the moment that $s \in [B_n^\vardiamond, C_n^\vardiamond]$, so that by Lemma \ref{ci} and Definition \ref{t2} it is the case that $S_D^\square(s) \leq C \cdot(n^{-a_4} + n^{-\beta})$ for some constant $C> 0$, so denoting by $\text{e}^{-V}$ (for some $V > 0$) the uniform bound on the cost of one patch and one annular crossing via the double bonds, if 
$\kappa > 0$ is sufficiently small so that
$\text{e}^{-\kappa V\log n} = n^{-\kappa V} > C\cdot (n^{-a_4} + n^{-\beta})$, then it is not possible that $s \in [B_n^\vardiamond, C_n^\vardiamond]$.  By cyclically permuting the relevant $B, C, D$ labels, the cases where $s \in [C_n^\vardiamond, D_n^\vardiamond]$ and $s \in [D_n^\vardiamond, B_n^\vardiamond]$ follow similarly.
\end{proof}

\begin{figure}[ ]
\centering
\includegraphics[scale=.28]{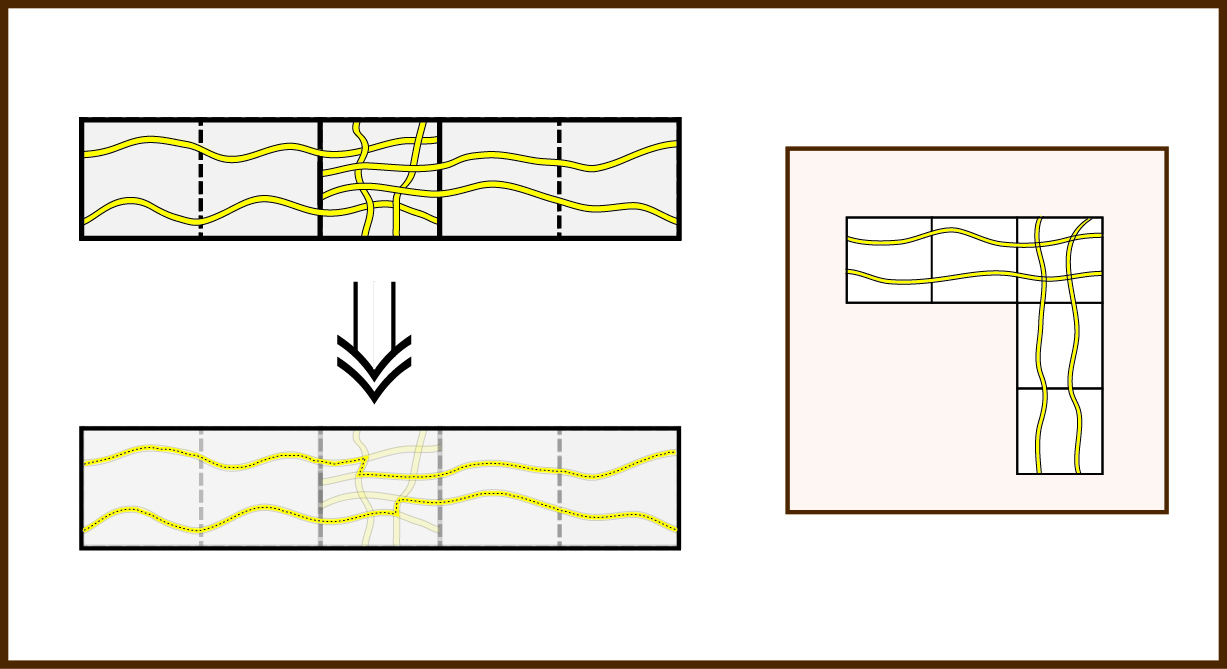}
\begin{changemargin}{.55in}{.55in}
\caption{{\footnotesize Two pairs of disjoint left--right crossings in successive overlapping $3\times1$ blocks together with two disjoint top--bottom crossings in the common square allow for the continuation of two disjoint paths.  Not all crossings described in the bond event are shown.  Note, as illustrated in the insert that, in the case of right angle continuations, the additional paths in the overlap block are superfluous.}}
\label{WLM}
\end{changemargin}
\vspace{-.3 in}
\end{figure}

The ensuing arguments will require an auxiliary point somewhat inside $\Omega_n^\vardiamond$, which we will denote $A_n^\diamondsuit$:

\begin{defn}\label{admd}
Let $\Omega_n^\square, \Omega_n^\vardiamond$, etc., be as described.  Let $\eta > 0$ be a number to be specified in Proposition \ref{dia2}.  Then we let $A_n^\diamondsuit$ be a point in the Harris ring of the Harris system stationed at $A_n^\square$ which is separated from $\mathcal D_\Delta$ by $\eta \cdot \log n$ Harris segments.  Moreover, $A_n^\diamondsuit$ is in the center of a box which belongs to the connected component of the boxes which percolate through the relevant ring (see the description in Theorem \ref{hreg}, item 1) as in the proof of Lemma \ref{dp}.  
\end{defn}

\begin{prop}\label{dia2} 
There exists some $\eta > 0$ such that if $A_n^\diamondsuit$ is as in Definition \ref{admd}, then there exists some $\gamma > 0$ such that 
\begin{enumerate}
\item $|S_n^\square(A_n^\square) - S_n^\square(A_n^\diamondsuit)| \lesssim n^{-\gamma}$;
\item $|H_n^\square(A_n^\square) - H_n^\square(A_n^\diamondsuit)| \lesssim n^{-\gamma}$;

\end{enumerate}
In particular, with appropriate choice of $\gamma$, $A_n^\diamondsuit$ is strictly inside $\Omega_n^\vardiamond$.
\end{prop}

\begin{proof}
First let us establish item 1.  It is claimed that for any configuration in which the event 
$\mathcal A(A_n^\diamondsuit)$ -- of a blue circuit 
connecting $[D_{n}^{\square}, A_{n}^{\square}]$ to $[A_{n}^{\square}, B_{n}^{\square}]$
which surrounds both $A_n^\square$ and $A_n^\diamondsuit$ (as described in Theorem \ref{hreg}, item \ref{hseal2}) -- occurs, the indicator function of the yellow version of
$\mathbb S_{n}^{\square}(A_{n}^{\square})$ is \textit{equal} to that of 
$\mathbb S_{n}^{\square}(A_{n}^{\diamondsuit})$. Indeed, for the $S_C^\square$--component,
which always vanishes for $A_{n}^{\square}$,
the requisite event in yellow is directly obstructed
by the blue paths of $\mathcal A(A_n^\diamondsuit)$.  
As for the rest, for either of the differences in the $B$ or $D$ components to be non--zero, there must be a long yellow path separating
$A_{n}^{\square}$ from $A_{n}^{\diamondsuit}$ heading to a distant boundary,
but this separating path is preempted by the blue event $\mathcal A(A_n^\diamondsuit)$.  We may thus conclude that
\begin{equation}
\mathbb E(|\mathbb I_{\mathbb S_{n}^\square(A_{n}^{\square})} -  \mathbb I_{\mathbb S_n^\square(A_{n}^{\diamondsuit})}|\mid\mathcal A(A_n^\diamondsuit))  =  0
\end{equation}
(where $\mathbb I_{(\bullet)}$ denotes the indicator)
which together with Lemma \ref{dp} and Theorem \ref{hreg}, item \ref{hseal2} gives the result.

As for item 2, recalling the discussion near the end of the proof of Theorem \ref{s1}, 
we may consider 
$\Omega_{n}^\square$ to be a fixed \emph{continuum} domain
and, e.g., for $N \geq n$, the domain
$\Omega_{n,N}^\square$ to be its canonical approximation (together with appropriate approximations for the marked points $A_n^\square, B_n^\square$, etc.) for a percolation problem at scale 
$N^{-1}$.  We will consider the corresponding CCS--functions $S_{n, N}^\square$ on the domains $\Omega_{n, N}^\square$.  

Let us now argue that the arguments for item 1 persist, uniformly, for all $N$ sufficiently large.  First, it is emphasized that all the results follow from the occurrence of paths in each Harris ring, which has probability uniformly bounded from below.  We claim that this remains the case for percolation performed at scale $N^{-1}$.  Indeed, while the scales of the Harris rings were constructed existentially to ensure uniform bounds on crossings at scale $n^{-1}$, it is recalled that these rings are gridded  by boxes of scale $2^{-2r}$ \emph{relative} to the rings themselves (see Theorem \ref{hreg}, item 1).  Thence, using uniform probability crossings in squares/rectangles, etc., the necessary crossings can be constructed by hand as in e.g., the proof of Lemma \ref{dp}.

For the last statement, we invoke an argument similar to that in the proof of Lemma \ref{dp}.  Recapitulating the construction, we acquire a lower bound on the probability of occurrence of any of the events associated with the $S$--functions for $A_n^\diamondsuit$.  Finally, since $S_n^\square$ is close to $F_n^\square$ by Lemma \ref{ci} the latter of which is used to \textit{define} $\partial \Omega_n^\vardiamond$, with an appropriate choice of power of $n$ (i.e., $\gamma$) $A_n^\diamondsuit$ can be placed in the interior of $\Omega_n^\vardiamond$.
\end{proof}

\begin{prop}\label{s2}
There exists some $a_5 > 0$ such that 
\[ |F_n^\square(A_n^\diamondsuit) - S_n^\square(A_n^\square)| \leq n^{-a_5}.\]

\end{prop}
\begin{proof}
This follows immediately from Proposition \ref{dia2}, item 1 and Lemma \ref{ci}.  
\end{proof}

Finally, we will need a result concerning the conformal maps $H_n^\vardiamond$ and $H_n^\square$.  First we state a distortion estimate:

\begin{lemma}\label{distort}
Let $\epsilon > 0$ and let $K \subseteq \mathbb T$ be a domain whose boundary is a Jordan curve such that the sup--norm distance between $\partial K$ and $\partial \mathbb T$ is less than $\epsilon$.   We consider $K$ to be a conformal triangle with some marked points $K_B, K_C, K_D$ such that  $|K_B-1|<\epsilon$, $|K_C-\tau|<\epsilon$, $|K_D-\tau^2|<\epsilon$, and let $g_K$ denote the conformal map from $K$ to $\mathbb T$ mapping $(K_B, K_C, K_D)$ to $(1, \tau, \tau^2)$.  Then for $z \in K$ it is the case that
\[ |g_K(z) - z| \lesssim [\epsilon \cdot \log(1/\epsilon)]^{1/3}.\]
\end{lemma}
\begin{proof}

The result for the disk (without the power of 1/3) is a classical result going back to Marchenko (for a statement see \cite{W}, Section 3) and of course, we can transfer our hypotheses to the disk by applying a conformal map $\phi$, which maps $\mathbb T$ to the unit disk such that $\phi(0) = 0$. The map $\phi$ does not increase the distances, because it is smooth up to the boundary everywhere but at $1$, $\tau$, and $\tau^2$, where it behaves locally like $\epsilon^3$, which in fact only \emph{decreases} the distances. 

We are almost in a position to directly apply Marchenko's Theorem except for a few caveats.  First of all Marchenko's Theorem requires a certain geometric condition on the tortuosity of the boundary of $K$, which is manifestly satisfied under the assumption that $\partial K$ and the boundary of the \emph{triangle} are close in the \emph{sup--norm} distance.

Secondly, Marchenko's Theorem is stated for some map $f_K$ with $f_K(0) =0$ and $f'_K(0)>0$, and we have a possibly different normalization.  Specifically, we have some map $G_K: \phi(K) \rightarrow \mathbb D$ so that $\phi^{-1} \circ G_K \circ \phi = g_K$, so it suffices to check that $G_K$ has approximately the correct normalizations (indeed, the conformal self--map of the unit disc mapping a point $a$ to the origin takes the form ~$e^{i\theta}\cdot \left(\frac{z - a}{1 - \bar a z}\right)$).  

Since $\phi(0) = 0$ and $1 + \tau + \tau^2 = 0$ it is the case that $\phi^{-1}((1-\epsilon) \cdot \phi(K_B + K_C + K_D))$ is close to 0 and also close to $w:= \phi^{-1}((1-\epsilon) \cdot \phi(K_B)) + \phi^{-1}((1-\epsilon) \cdot \phi(K_C)) + \phi^{-1}((1-\epsilon) \cdot \phi(K_D))$; since it is also the case that $g_K(K_B) + g_K(K_C) + g_K(K_D)$ is close to 0, we have that $G_K (w)$ is close to 0.  So we now have that $G_K(z)$ is close to some $e^{i\theta}z$ for some fixed $\theta$. But since $\phi (K_B)$ is close to $\phi(1)$, and so $z_0:=\phi^{-1}((1-\epsilon)\phi(K_B))$ is close to both 1 and $e^{-i\theta} \cdot G_K(1)$, it follows that $|e^{i\theta}-1|\lesssim \epsilon \cdot \log(1/\epsilon)$.  

Finally, in transferring the result back to the triangle, the behavior near the vertices of the triangle requires us to replace the distances by their cube roots. 
\end{proof}

\begin{remark}
We remark that for our purposes, we can in fact avoid the fractional power: indeed, we shall only use this result at the point $A_n^\diamondsuit$, which we remind the reader is chosen to be in the Harris system stationed at $A_n^\square$ and by Lemma \ref{ped} we may assert that it is within a fixed small neighborhood of $A_n^\square$ and therefore outside fixed neighborhoods of the other marked points.  
\end{remark}

\begin{lemma}\label{s3}
There exists some $a_6 > 0$ such that for all $n$ sufficiently large, 
\[ |H_n^\vardiamond(A_n^\diamondsuit) - H_n^{\square}(A_n^{\diamondsuit})| \lesssim n^{-a_6}.\]
\end{lemma}

\begin{proof}
Denoting by $G_n$ the conformal map mapping $H_n^\square(\Omega_n^\vardiamond)$ to $\mathbb T$ with the points 
$(H_n^\square(B_n^\vardiamond), H_n^\square(C_n^\vardiamond)$, $H_n^\square(D_n^\vardiamond))$ mapping to the points $(1, \tau, \tau^2)$, we have by uniqueness of conformal maps that 
\[H_n^\vardiamond = G_n \circ H_n^\square.\]
The stated result will follow from Lemma \ref{distort}, and in order to utilize this lemma, we need to verify that $(H_n^\square(B_n^\vardiamond), H_n^\square(C_n^\vardiamond), H_n^\square(D_n^\vardiamond))$ is close to $(1, \tau, \tau^2)$ and to show that the sup--norm distance between $\partial[H_n^\square(\partial \Omega_n^\vardiamond)]$ and $\partial \mathbb T$ is less than $n^{-\gamma}$ for some $\gamma > 0$.  The first statement is a direct consequence of Proposition \ref{pts}: since $O(\log n)$ Harris rings surround both $B_n^\vardiamond$ and $B_n^\square$, by an argument as in the proof of Proposition \ref{dia2}, their $S_n^\square$ values differ by an inverse power of $n$ and the result follows since $S_n^\square(B_n^\square) \equiv 1$; similar arguments yield the result for $C_n^\vardiamond, D_n^\vardiamond$. 

As for the second statement, first we have by Lemma \ref{ci} and Lemma \ref{cdmd} that the distance between $\partial [S_n^\square(\partial \Omega_n^\vardiamond)]$ and $\partial \mathbb T$ is less than (some constant times) $n^{-a_4} + n^{-\beta}$; we emphasize that here we in fact have closeness in the sup--norm since both lemmas yield pointwise estimates.  Next, as near the end of the proof of Theorem \ref{s1}, 
we may consider 
$\Omega_{n}^\square$ to be a fixed \emph{continuum} domain
and, e.g., for $N \geq n$, the domain
$\Omega_{n,N}^\square$ to be its canonical approximation (together with appropriate approximations for the marked points $A_n^\square, B_n^\square$, etc.) for a percolation problem at scale 
$N^{-1}$.  We will consider the corresponding CCS--functions $S_{n, N}^\square$ on the domains $\Omega_{n, N}^\square$.  

We claim that there exists some $\gamma > 0$ such that uniformly in $N$ for $N$ sufficiently large, the sup--norm distance between $\partial[S_{n, N}^\square(\partial \Omega_n^\vardiamond)]$ and $\partial \mathbb T$ is less than $n^{-\gamma}$.  Indeed, from Lemma \ref{dp}, we know that for each point $s$ on $\partial \Omega_n^\vardiamond$, there are $\kappa \cdot \log n$ Harris rings stationed at $z(s)$ which separate it from the central region $\mathcal D_\Delta$.  While by \emph{fiat} $S_{n, N=n}^\square(\partial \Omega_n^\vardiamond)$ is close to $\partial \mathbb T$, we shall reprove this using the Harris systems since we require an estimate which is uniform in $N$.  We start with the following observation concerning the central region $\mathcal D_\Delta$:  

\bigskip

\noindent \textbf{Claim.} For $n$ sufficiently large, with probability of order unity independent of $n$, there are monochrome percolative connections (in blue or yellow) between $\mathcal D_\Delta$ and any or all of the three boundary segments.  

\bigskip

\noindent \emph{Proof of Claim.} Consider the domain $\Omega$ with marked points $B, C, D$, viewed as a conformal triangle.  It is recalled that $\mathcal D_\Delta$ is roughly half the size of the largest circle which can be fit into $\Omega$.  Let us focus on two of the three marked points, say $B$ and $D$.  We now mark two boundary points on $\mathcal D_\Delta$ and denote them by $b$ and $d$ and consider two disjoint curves which join $B$ to $b$ and $D$ to $d$, thereby forming a conformal rectangle.  Since the aspect ratio of the said rectangle is fixed, it therefore follows, by convergence to Cardy's Formula, that for $n$ sufficiently large, there is a uniform lower bound on the probability of a discrete realization of the desired connection.  Similar arguments apply to the other two boundary segments.
\qed

\bigskip

\noindent \textbf{Claim.} Consider $s \in \partial \Omega_n^\vardiamond$ and the Harris rings from the Harris system stationed at $z(s) \in \partial \Omega_n^\square$ which also enclose $s$ as in Lemma \ref{dp}.  Without loss of generality, we may assume that $z(s) \in [B_n^\square, D_n^\square]$.  Then there exists some fixed constant $\Upsilon < \infty$ such that all but $\Upsilon$ of the Harris segments have at least one endpoint on $[B_n^\square, D_n^\square]$.  Moreover, among these, \emph{either} the other endpoint of the Harris segment is also on $[B_n^\square, D_n^\square]$ \emph{or} the existence of the corresponding path event within this Harris segment achieves $\mathbb S_B^\square$ \emph{or} $\mathbb S_D^\square$ (which, we remind the reader, are the percolation events defining $S_D^\square$ and $S_B^\square$, respectively) for \emph{both} $s$ and $z(s)$.  Similar statements hold if $z$ belongs to the other boundary segments.

\bigskip

\noindent \emph{Proof of Claim.}  Let us first rule out the possibility that too many Harris segments have endpoints on $[B_n^\square, C_n^\square, D_n^\square]$.  It is noted that each Harris segment of this type in fact separates all of $[B_n^\square, D_n^\square]$ from $\mathcal D_\Delta$.  Thus, if there are say $\Upsilon$ such Harris segments, then the probability of a connection between $\mathcal D_\Delta$ and $[B_n^\square, D_n^\square]$ would be less than $(1 - \vartheta)^\Upsilon$, with $\vartheta > 0$ as in Theorem \ref{hreg}.  It follows from the previous claim that $\Upsilon$ cannot scale with $n$.

Finally, if there are too many Harris segments with one endpoint on $[B_n^\square, D_n^\square]$, but accomplishes \emph{neither} $\mathbb S_B^\square$ \emph{nor} $\mathbb S_D^\square$, then necessarily the other endpoint must be on $[B_n^\square, C_n^\square]$ or $[C_n^\square, D_n^\square]$ in such a way that the Harris segment separates $\mathcal D_\Delta$ from $[B_n^\square, C_n^\square]$ or $[C_n^\square, D_n^\square]$.  The same reasoning as in the above paragraph then implies that this also cannot occur ``too often''.  For illustrations of some of these cases, see Figure \ref{FZX}.
\qed 

\begin{figure}[ ]
\centering
\includegraphics[scale=.4]{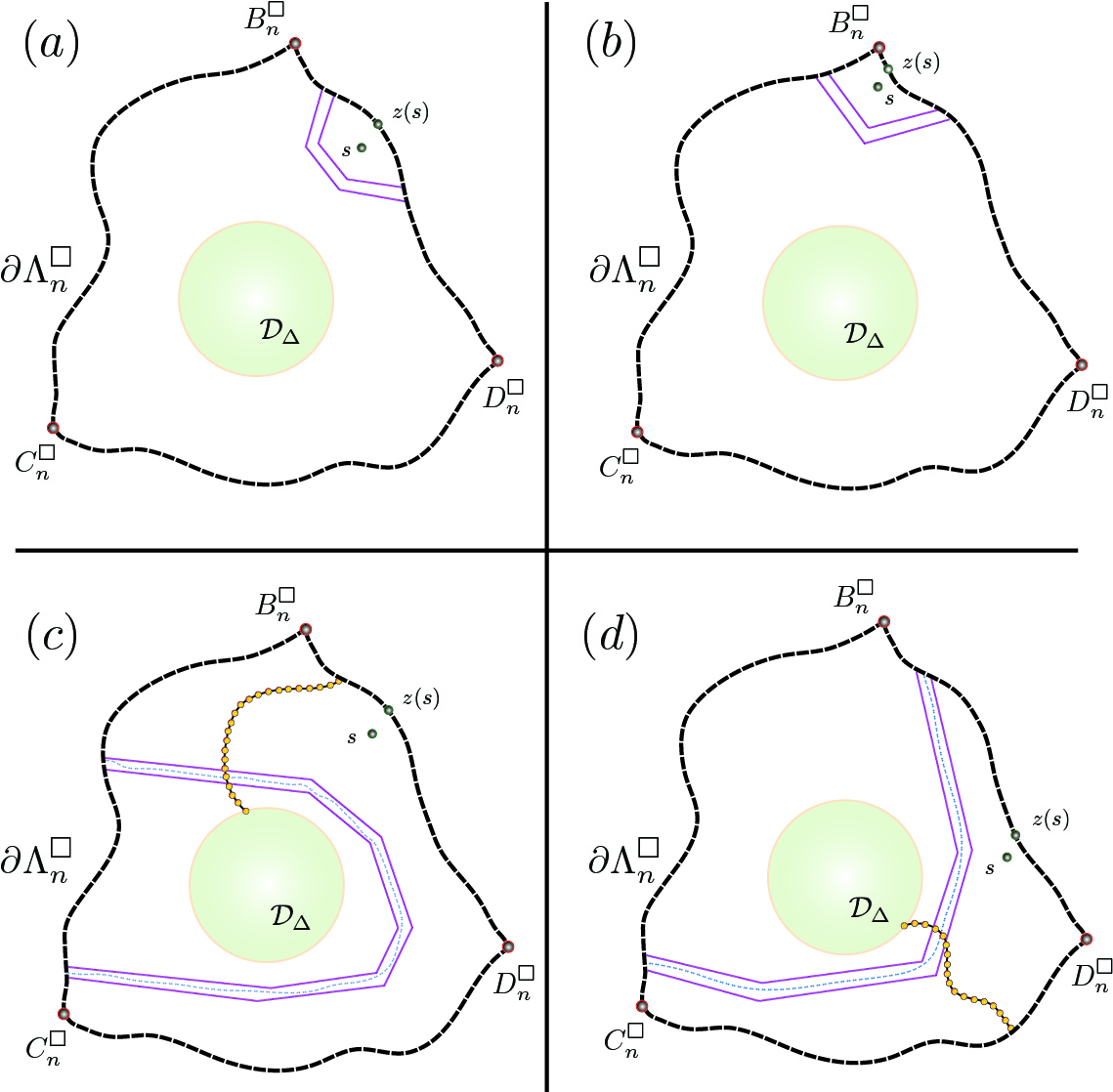}
\begin{changemargin}{.55in}{.55in}
\caption{{\footnotesize (a) and (b): Harris annular segment of the type envisioned for $s \in \partial \Omega_n^\vardiamond$ and $z(s) \in [B_n^\square, D_n^\square]$; in case (b), it so happens that $z(s)$ and $B_n^\square$ are close.
(c)  and (d):  since yellow paths indicated in these illustrations represent typical large scale events, there cannot be too many Harris rings of the contrary type.}}
\label{FZX}
\end{changemargin}
\vspace{-.3 in}
\end{figure}

\bigskip

We also note that there cannot be Harris segments of conflicting ``corner types'' (e.g., $[D_n^\square, B_n^\square]$ to $[B_n^\square, C_n^\square]$ \emph{and} $[B_n^\square, C_n^\square]$ to $[C_n^\square, D_n^\square]$) since the Harris segments are topologically ordered and cannot intersect one another.

\bigskip

We can now acquire the needed conclusion that the Harris rings themselves force $S_{n, N}^\square(s)$ to be close to $\partial \mathbb T$.  The essence of the argument can be captured by the (redundant) case $N = n$, so let us proceed.  Consider then $s \in \partial \Omega_n^\vardiamond$ and the Harris system stationed at $z(s) \in \partial \Omega_n^\square$ as above which, without loss of generality, we assume to be in $[B_n^\square, D_n^\square]$.  Then we claim that $|S_n^\square(z(s)) - S_n^\square(s)| \lesssim n^{-\kappa}$.  Indeed, from the previous claim, all but $\Upsilon$ of the Harris segments have beginning and ending points on $\partial \Omega_n^\square$ which are such that conditioned on the existence of paths of the appropriate color within these segments, the indicator functions of all $S_n^\square$--events are the same value for both $s$ and $z(s)$.

Let us now argue that the above argument persists, uniformly, for all $N$ sufficiently large.  First, it is emphasized that all arguments follow from the occurrence of paths in each Harris ring, which has probability uniformly bounded from below.  We claim that this remains the case for percolation performed at scale $N^{-1}$.  
Indeed, let us again recall that these rings are gridded  by boxes of scale $2^{-2r}$ \emph{relative} to the rings themselves (see Theorem \ref{hreg}) and using uniform probability of crossings in squares/rectangles, etc., which is characteristic of critical 2D percolation problems, the necessary crossings can be constructed by hand as in e.g., the proof of Lemma \ref{dp}.

Now by convergence to Cardy's Formula (or rather, the statement that the CCS--function converges uniformly on compact sets to the conformal map to $\mathbb T$)  it is the case that $S_{n, N}^\square(s) \rightarrow H_n^\square(s)$.  Uniformity in $s$ follows from the fact that $\overline{\Omega_n^\vardiamond} \subseteq \Omega_n^\square$ is a fixed (for $n$ fixed) compact set, c.f.,~Section 5 in \cite{BCL2}.  We conclude therefore that each point on $\partial \Omega_n^\vardiamond$ maps to a point sufficiently close to $\partial \mathbb T$, and since $\partial[H_n^\square(\partial \Omega_n^\vardiamond)]$ is a curve, it easily follows that the \emph{Hausdorff distance} is small.

However, we require the stronger statement that the relevant objects are close in the \emph{sup--norm} (i.e., in $d_{\rm sup}( \cdot, \cdot)$).  We will now strengthen the above arguments to acquire this conclusion. Let us define the set of all points which are chosen as the $z(s)$ (the closest point to $s$) for some $s$ in $\langle \partial \Omega_n^\vardiamond\rangle_N$ (the approximation to $\partial \Omega_n^\vardiamond$ at scale $N^{-1}$):
\[ \mathcal Z_N:= \{ z \in \partial \Omega_{n, N}^\square \mid \exists s \in \langle\partial \Omega_n^\vardiamond\rangle_N, z = z(s)\}.\]

Let us first observe that \emph{a priori} $S_{n, N}^\square(\mathcal Z_N)$ is a discrete set of points on $\partial \mathbb T$ which we may consider to be a curve by linear interpolation.  For simplicity let us consider the portion of $\partial \mathbb T$ corresponding to the $[C, D]$ boundary, i.e., the vertical segment connecting $\tau$ and $\tau^2$.  Let us focus attention on $S_{n, N}^\square([C_{n, N}^\square, D_{n, N}^\square] \cap \mathcal Z_N)$.  By monotonicity of crossing probabilities, it is the case that these points are ordered along the vertical segment.  

Now our contention is that there are no substantial gaps between successive points: 

\bigskip

\noindent \textbf{Claim.} Let $s \in \langle \partial \Omega_n^\vardiamond \rangle_N$ and $z(s)$ be as described.  Let $\nu > 0$ be such that the inequality $n^{-\nu} \gg n^{-\kappa}$ is sufficiently strong, as will emerge in the proof, where $\kappa$ as above is such that $|S_{n, N}^\square(s) - S_{n, N}^\square(z(s))| \lesssim n^{-\kappa}$.  Then for all $N > n$, it is the case that the maximum separation between successive points of $S_{n, N}^\square([C_{n, N}^\square, D_{n, N}^\square] \cap \mathcal Z_N)$ is less than $n^{-\nu}$, with $\mathcal Z_N$ as described.

\bigskip

\noindent \emph{Proof of Claim.}  Suppose there are two points $x_1, x_2 \in [C_{n, N}^\square, D_{n, N}^\square] \cap \mathcal Z_N$, say with $S_{n, N}^\square(x_1)$ below $S_{n, N}^\square(x_2)$, which are separated by a gap in excess of $n^{-\nu}$.  Let us denote by $s_1, s_2 \in \langle \partial \Omega_n^\vardiamond \rangle_N$ the points corresponding to $x_1, x_2$, respectively.  Next consider the $\frac{1}{4} \cdot n^{-\nu}$ neighborhoods of $S_{n, N}^\square(s_1)$ and $S_{n, N}^\square(s_2)$ and consider the points in $\partial \Omega_n^\vardiamond$ ``between'' $s_1$ and $s_2$. There must be points between $s_1$ and $s_2$ since $|S_{n, N}^\square(s_1) - S_{n, N}^\square(s_2)| \gtrsim n^{-\nu} - 2 \cdot n^{-\kappa}$.  So if these points were neighbors, by standard critical percolation arguments, the difference between their $S_{n, N}^\square$ values must be small and the above inequality would render this difference unacceptably large, for $\nu$ appropriately chosen.     

If these points all have $S_{n, N}^\square$--value which lie in the $\frac{1}{4} \cdot n^{-\nu}$ neighborhoods described above, then there would be a neighboring pair whose $S_{n, N}^\square$ values are separated by $\frac{1}{2} \cdot n^{-\nu}$, which would again be unacceptably large.  We conclude therefore that there exists some point between $s_1$ and $s_2$ with $S_{n, N}^\square$ value outside these neighborhoods and therefore a point in $\mathcal Z_N$ whose $S_{n, N}^\square$ value lies between those of $x_1$ and $x_2$.  This is a contradiction.
\qed

Finally, let us describe the parametrization. First we denote by $U_N$ the number of points in $\mathcal Z_N$ and then we may parametrize say the vertical portion of $\partial \mathbb T$ by having, for $t = j$, the curve on the $j^{\text{th}}$ site of $\mathcal Z_N$ and linearly interpolating for the non--integer times.  Similarly, we parametrize the corresponding portion of $S_{n, N}^\square(\langle \partial \Omega_n^\vardiamond\rangle_N$, so that pairs of points at integer times correspond to their $s, z(s)$ pair.  The above claim then implies that with this parametrization, the two curves are within $n^{-\nu}$ at all times.  We have verified that $S_{n, N}^\square(\langle \partial \Omega_n^\vardiamond\rangle_N)$ is sup--norm close to $\partial \mathbb T$, uniformly in $N$.

The stated result now follows from Lemma \ref{distort}.
\end{proof}

\bigskip

\noindent \emph{Proof of the Main Theorem.} The required power law estimate for the rate of convergence of crossing probabilities now follows by concatenating the various theorems, propositions and lemmas we have established.  Let us temporarily use the notation $A \sim B$ to mean that $A$ and $B$ differ by an inverse power of $n$.  

Starting with $S_n(A_n)$, we have that $S_n(A_n) \sim S_n^\square(A_n^\square)$ by Theorem \ref{s1}; $S_n^\square(A_n^\square) \sim S_n^\square(A_n^\diamondsuit)$ by Proposition \ref{dia2}, item 1; $S_n^\square(A_n^\diamondsuit) \sim F_n^\square(A_n^\diamondsuit)$ by Lemma \ref{ci}; $F_n^\square(A_n^\diamondsuit) \sim H_n^\vardiamond(A_n^\diamondsuit)$ by Lemma \ref{cdmd}; $H_n^\vardiamond(A_n^\diamondsuit) \sim H_n^\square(A_n^\diamondsuit)$ by Lemma \ref{s3}; $H_n^\square(A_n^\diamondsuit) \sim H_n^\square(A_n^\square)$ by Proposition \ref{dia2}, item 2; finally, $H_n^\square(A_n^\square) \sim H(A)$ by Theorem \ref{s1}.

\qed

\section{$\sigma$--Holomorphicity}
\label{KBB}

The main goal in this section is to establish the so--called Cauchy integral estimates which is one of the more technical aspects required for the proof of Lemma \ref{ci}.  We will address such issues in somewhat more generality than strictly necessary by extracting the two properties of functions of the type $S_n(z)$ which are of relevance: i) H\"older continuity and ii) that their discrete (closed) contour integrals are asymptotically zero as the lattice spacing tends to zero.  As for the latter, it should be remarked that the details of how 
our particular
$S_n(z)$ exhibits its cancelations on the \emph{microscopic scale} can be directly employed to provide the Cauchy--integral estimates.  

\subsection{$(\sigma, \rho)$--Holomorphicity}

As a starting point -- and also to fix notation -- let us review the concept of a discrete holomorphic function on a hexagonal lattice.  Let $\mathbb H_{\e}$ denote the hexagonal lattice at scale $\e$, e.g., the length of the sides of each hexagon is $\e$, so we envision $\e = n^{-1}$, where the hexagons are oriented horizontally (i.e., two of the sides are parallel to the $x$--axis).  For now, let 
$\Lambda_{\e}$ denote any collection of hexagons and 
$Q: \Lambda_{\e} \to \mathbb C$ a function on the vertices of $\Lambda_{\e}$.  
For each pair of adjacent vertices in $\Lambda_{\e}$ let us linearly interpolate $Q$ on the edges.  In particular, $Q$ as a function on \emph{edges} when integrated with respect to arc length yields the average of the values of $Q$ at the two endpoints.  Hence all integrals may be regarded as taking place in the continuum.

\begin{itemize}
\item[~] We say that $Q$ is \emph{discrete holomorphic} on 
$\Lambda_{\e}$
if for any hexagon
$h_\e \in \Lambda_{\e}$ with vertices $(v_{1}, \dots v_{6})$  -- in counterclockwise order with $v_{1}$ the leftmost of the lowest two -- the following holds:
$$
0 = \left(\frac{Q(v_{1}) + Q(v_{2})}{2} + \dots + 
e^{i\frac{5}{3}\pi} \cdot \frac{Q(v_{6}) + Q(v_{1})}{2}\right)
= \varepsilon^{-1} \cdot \oint_{\partial h_\e} Q~dz.
$$
\end{itemize}
That is, the usual discrete contour integral (by this or any equivalent) definition vanishes.  By way of contrast, we have the following mild generalization pertaining to \textit{sequences} of functions.   

\begin{defn}
\label{nhf}
Let $\Lambda \subseteq \mathbb C$ be a simply connected domain and denote by $\Lambda_\e$ the (interior) discretized domain given as $\Lambda_\e := \bigcup_{h_\e \subseteq \Lambda} h_\e$ and let $(Q_{\e}: \Lambda_{\e} \to \mathbb C)$ be a 
sequence of functions defined on the vertices of $\Lambda_\e$.  Here $\e$ is tending to zero and, without much loss, may be taken as a discrete sequence.
We say that the sequence $(Q_{\e})$ is \emph{$\sigma$--holomorphic} if there exist constants $0 < \sigma, \rho \leq 1$ such that for all $\e$ sufficiently small:

(i) $Q_{\varepsilon}$ is H\"older continuous (down to the scale $\e$) and up to $\partial \Lambda_\e$, in the sense that there exists some $\psi > 0$ (envisioned to be small) such that  1) $Q_\e$ is H\"older continuous in the usual sense for $z_\e, w_\e \in \Lambda_\e \setminus N_\psi(\partial \Lambda_\e)$: if $|z_\e - w_\e| < \psi$, then $|Q_\e(z_\e) - Q_\e(w_\e)| \lesssim \left( \frac{|z_\e - w_\e|}{\psi} \right)^\sigma$ and 2) if $z_\e \in N_\psi(\partial \Lambda_\e)$, then there exists some $w_\e^\star \in \partial \Lambda_\e$ such that $ |Q_\e(z_\e) - Q_\e(w_\e^\star)| \lesssim \left(\frac{|z_\e - w_\e^\star|}{\psi} \right)^\sigma$.

(ii) for any simply closed lattice contour $\Gamma_\e$,

\begin{equation}
\label{sint} 
| \oint_{\Gamma_\e} Q~dz 
| = |\sum_{h_\e \subseteq \Lambda_\e'} \oint_{\partial h_\e} Q~dz | 
\lesssim |\Gamma_\e| \cdot \e^\rho,
\end{equation}

with $\Lambda_\e', |\Gamma_\e|$ denoting the region enclosed by $\Gamma_\e$ and the Euclidean length of $\Gamma_\e$, respectively.

\end{defn}

\begin{remark}  (i)  Obviously any sequence of discrete holomorphic functions \textit{which also satisfy the H\"older continuity condition} are $\sigma$--holomorphic.

(ii) There are of order $|\Gamma_\e|/\e$ terms in a discrete contour integration but each term is multiplied by $\e$ and so in cases where $|\Gamma_\e| = O(1)$ (a contour of fixed finite length)  $|\Gamma_\e|$ need not be explicitly present on the right hand side of Equation \eqref{sint}.  We have introduced a more general definition as we shall have occasion to consider contours whose lengths scale with $\e$ (specifically they are discrete approximations to contours that are not rectifiable).  
(iii) From the assumption of H\"older continuity alone, we already have that $|\oint_
{\partial h_\e} Q~dz | \lesssim \e^{1+\sigma}$,
but on a moment's reflection, it is clear that this is quite far from what is necessary to 
provide adequate estimates for the integral around contours of \emph{larger scales} that are amenable to the
$\e \rightarrow 0$ limit.  
\end{remark}

We will now gather the necessary ingredients to establish that the (complexified) CCS--functions are $(\sigma, \rho)$--holomorphic.  The arguments here are certainly not new: various ideas and statements needed are already almost completely contained in \cite{stas_perc}, \cite{CL} and \cite{BCL2}. 

\begin{prop}\label{sfp}  
Let $\Lambda$ denote a conformal triangle with marked points (or prime ends) 
$B$, $C$, $D$ and let 
$\Lambda_{\e}$ denote an interior approximation (see Definition 3.1 of \cite{BCL2}) of $\Lambda$ with $B_{\e}, C_\e, D_\e$ the associated boundary points.  
Let $S_{\e}(z)$ denote the complex crossing function defined on $\Lambda_{\e}$.  
Then for all $\e$ sufficiently small, the functions $(S_\e: \Lambda_\e \rightarrow \mathbb C)$ are $(\sigma, \rho)$--holomorphic for some $\sigma, \rho > 0$.
\end{prop}

\begin{proof}
We will first establish, using some conformal mapping ideas, that $S_\e$ enjoys H\"older continuity up to the boundary; since arguments like this already appear in \cite{BCL2}, we will be brief.  Let us start with a pointwise statement:

\noindent{\bf Claim.} Suppose we have a point 
$A$ on the $[D,B]$ boundary, then we claim that there is some 
$\Delta^{\star} \equiv \Delta^\star(A)$ (with $1 \gg \Delta^{\star} \gg \e$)
and a connected set $N_{\Delta^{\star}} \subset \Lambda_{\e}$, also contained in the $\Delta^{\star}$ neighborhood of $A_{\e}$ and connected to $A_{\e}$, 
such that the following holds:
there exists some $\sigma > 0$ such that 
for any $z \in N_{\Delta^\star}$,
$$
|S_{\e}(z) - S_{\e}(A_{\e})| 
\lesssim
\left[ \frac{|z - A_{\e}|}{\Delta^{\star}}\right]^{\sigma}.
$$

\noindent \emph{Proof of Claim.}  
Let $z\in \Lambda_{\e}$ and consider $S_\e(z)$ to correspond to blue paths.  Then it is clear that if there is a yellow path starting on
$[D_{\e}, A_{\e}]$ and ending on $[A_{\e}, B_{\e}]$ which encircles $z$, then events contributing to  $S_\e(z)$ and $S_\e(A_\e)$ occur together
and there is no contribution to $|S_\e(z) - S_\e(A_\e)|$.
The power $(|z - A_{\e}|/\Delta^{\star})^{\sigma}$
corresponds to having the order of 
$|\log (|z - A_{\e}|/\Delta^{\star})|$ annuli (or coherent portions thereof)
connecting the two parts of the $[D_{\e}, B_{\e}]$ boundary with an independent 
chance of such a yellow circuit in each segment with uniformly bounded
probability.   Thus the principal task is to construct the reference scale 
$\Delta^{\star}$ in a manner which is uniform in $\e$.
While the entire issue is trivial when $|A-A_{\e}|$, $|B - B_{\e}|$ etc., are small compared to the distance between various relevant ``points'' on $\Lambda$, 
we remind the reader that under certain circumstances, the separation between these points and their approximates  may be spuriously large. Thus we turn to uniformization.

To this end, let $\varphi: \mathbb D \to \Lambda$
denote the uniformization map.  Let $X^{\prime}_{A}$ denote a crosscut neighborhood of $\varphi^{-1}(A)$ which does not contain any of the inverse images of the marked points $\varphi^{-1}(B), \dots$ nor, for $\e$ small, the inverse images of their approximates $\varphi^{-1}(B_{\e}), \dots$
but which \textit{does} (for $\e$ small) contain
$\varphi^{-1}(A_{\e})$. 
Next we set 
$X_A : = X_A^{\prime} \cap \varphi^{-1}(\Lambda_\e)$ so that
$$
\varphi(X_A) = \varphi(X_A^{\prime} \cap \varphi^{-1}(\Lambda_\e)).
$$
Note that $(\varphi_\e^{-1} \circ \varphi)(X_A)$ is itself a 
crosscut neighborhood of the image of $A_\e$ since $\Lambda_\e$ is an \emph{interior} approximation; here $\varphi_\e$ denotes the uniformization map associated with $\Lambda_\e$.

Next let $r_\Pi = r_\Pi(A_\e)$ be standing notation for the square centered at $A_\e$ of side $\Pi$.  Then, for $\Delta^{\star}$ sufficiently small, it is the case that $\varphi^{-1}(r_{\Delta^{\star}}) \subseteq X_{A}$ and it is worth observing that 
$\varphi_\e^{-1} (\partial (r_\Pi \cap \varphi(X_\delta)))$
is a crosscut containing $\varphi_\e^{-1}(A_\e)$ for all $\Pi \leq \Delta^\star$.

But now, it follows that there is a nested sequence of (partial) annuli, down to scale $|z - A_{\e}|$, contained inside $r_{\Delta^\star}$, within each of which there is a connected monochrome chain with uniform and independent probability separating $z$ from $A_\e$.  \qed

\bigskip

From the claim we have that corresponding to each boundary point of $\Lambda$, we have a neighborhood $\Delta^\star(z)$ in which we have H\"older continuity and it is certainly the case that $\partial \Lambda \subseteq \bigcup_{z \in \partial \Lambda} N_{\Delta^\star(z)}$, so by compactness there exist $z^{(1)}, \dots, z^{(k)}$ such that $\partial \Lambda \subseteq \bigcup_{\ell = 1}^k N_{\Delta^\star(z^{(\ell)})}$.  Adding a few $N_{\Delta(z)}$'s if necessary so that all neighborhoods have non--trivial overlap, this implies the existence of some $\psi > 0$ such that $N_\psi(\partial \Lambda) \subseteq \bigcup_{\ell = 1}^k N_{\Delta^\star(z^{(\ell)})}$ (here $N_\psi(\partial \Lambda)$ denotes the Euclidean $\psi$--neighborhood of $\partial \Lambda$).  In particular, $\psi \leq \Delta^\star(z^{(\ell)}), \ell = 1, \dots, k$, so if $\e \ll \psi$, and $z_\e \in N_\psi(\partial \Lambda)$, then $z_\e \in N_{\Delta^\star(z^{(\ell)})}$ for some $\ell$ and so $|S_\e(z_\e) - S_\e(z^{(\ell)}_\e)| \lesssim \left( \frac{|z_\e - z^{(\ell)}_\e|}{\psi} \right)^\sigma$.
For $z_\e, w_\e \in \Lambda_\e \setminus N_\psi(\partial \Lambda)$, $|z_\e - w_\e| < \psi$, so there are clearly of the order $\log(|z_\e - w_\e|/\psi)$ annuli surrounding both $z_\e$ from $w_\e$ and we obtain $|z_\e - w_\e| \lesssim \left( \frac{|z_\e - w_\e|}{\psi} \right)^\sigma$.

Finally, the statement concerning the behavior of discrete contour integrals of $S_\e$ can be directly found in \cite{stas_perc} for the triangular lattice (also c.f., \cite{beffara}) and in \cite{CL}, \S 4.3, for the extended models.
\end{proof}

\subsection{Cauchy Integral Estimate}
\label{CIEX}

We will start by establishing a multiplication lemma for an actual holomorphic function with a nearly--holomorphic function:

\begin{lemma}\label{panal}
Let $Q_{\e}$ be part of a $(\sigma, \rho)$--holomorphic sequence
as described in Definition \ref{nhf}.
above.
Let $\e > 0$ and suppose $\Gamma_{\e}$ is a discrete closed contour consisting of edges of hexagons at scale $\e$.  Let $q(z)$ be a holomorphic function on $\Lambda$ restricted to
$\Lambda_{\e}$ (both vertices and edges, all together regarded as a subset of $\mathbb C$).  Next let $1 \gg D \gg \e$ (both considered small).  Then for all $\e \geq 0$ sufficiently small
\[ |\oint_{\Gamma_{\e}} q \cdot Q_{\e}~dz| 
\lesssim (\|q\|_\infty \cdot \frac{\e^\rho}{D} + \|q\|_{C^1}\cdot D^\sigma)\cdot (|{\rm{Int}}(\Gamma_\e)| + |\Gamma_\e| \cdot D). \]
Here we remind the reader that $\|q\|_{C^1} = \|q\|_\infty + \|q'\|_\infty$.  Moreover, in the statement and upcoming proof of the lemma we also remind the reader that all integrals are regarded as taking place in the continuum.
\end{lemma}
\begin{proof}

Consider a square--like grid of scale $D$ and let $R_k$ denote the $k^{\text{th}}$ such square which has non--empty intersection with $\Lambda_\e$.  Next we let 
\[ \gamma_k := \partial(R_k \cap \mbox{Int}(\Gamma_\e)).\]
Note that $\gamma_k$ is not necessarily a single closed contour, but each $\gamma_k$ is a collection of closed contours.  See Figure \ref{OWW}.  It is observed that if $F$ is a function, then $\oint_{_{\Gamma_\e}} F ~dz= \sum_k \oint_{\gamma_k} F~dz$, where by abuse of notation, as mentioned above, each term on the righthand side may represent the sum of several contour integrals.  Next let us register an estimate within a single region bounded a $\gamma_k$, the utility of which will be apparent momentarily:

\begin{figure}[ ]
\centering
\includegraphics[scale=.25]{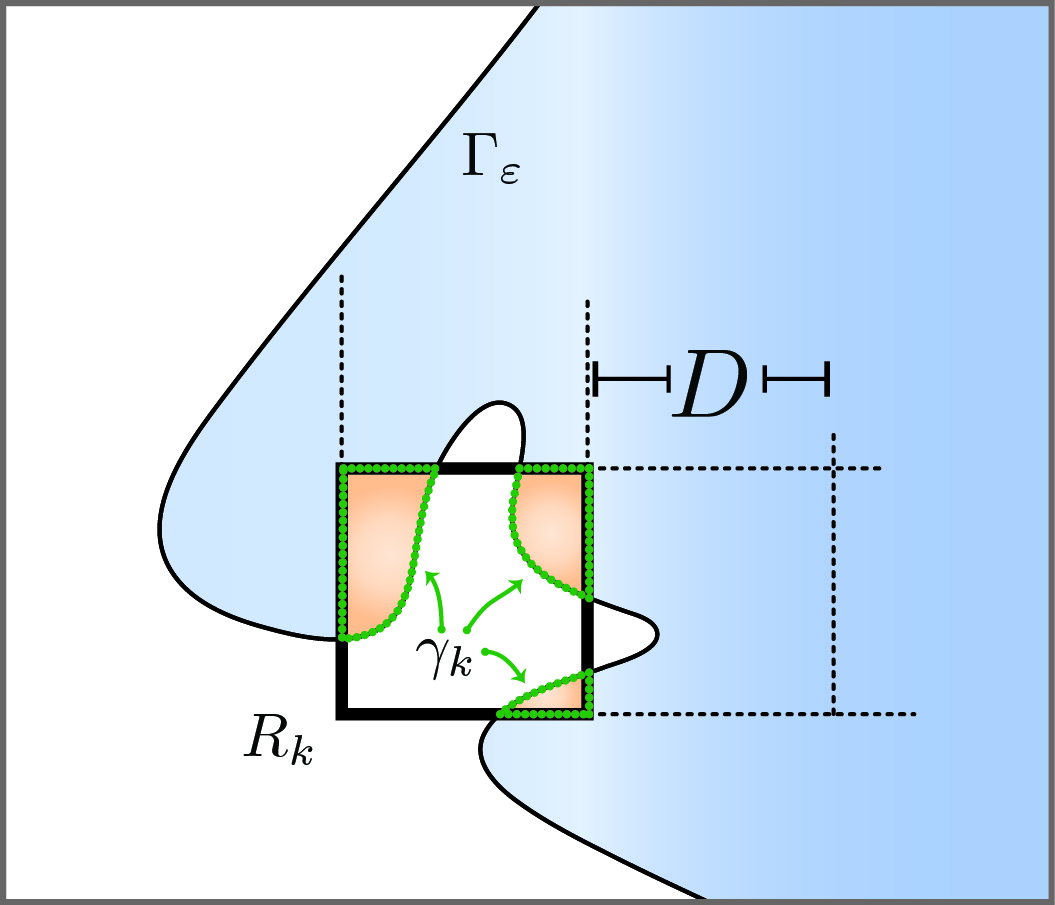}
\begin{changemargin}{.55in}{.55in}
\caption{{\footnotesize Aspects described in the proof of Lemma \ref{panal}.}}
\label{OWW}
\end{changemargin}
\vspace{-.3 in}
\end{figure}

\bigskip

\noindent \textbf{Claim.}  Let $z_k \in R_k$ (if $R_k$ intersects $\partial \Lambda_\e$ then choose $z_k$ in accordance with item (i) of the definition of $\sigma$--holomorphicity so that H\"older continuity of $Q$ can be assumed).  Then 
\begin{equation}
\label{eoh}
\oint_{\gamma_k} q \cdot Q~dz  = q(z_k)  \cdot   \oint_{\gamma_k} Q~dz + \mathcal E_k,\end{equation}
where 
\[ |\mathcal E_k| \lesssim |\gamma_k| \cdot \|q\|_{C^1} \cdot D^{1+\sigma}\]
and to avoid clutter, we omit the $\e$ subscript on the $Q$'s.  

\noindent \emph{Proof of Claim.}  Let us write 
\[ Q(z) = Q(z_k) + \delta Q(z)\]
and similarly, 
\[ q(z) = q(z_k) + \delta q(z). \]
We then have that 
\[ \oint_{\gamma_k} q \cdot Q~dz - q(z_k) \cdot \oint_{\gamma_k} Q~dz = \oint_{\gamma_k} \delta Q \cdot \delta q~dz + Q(z_k) \cdot \oint_{\gamma_k} \delta q~dz.\]
The second term on the right hand side vanishes identically by analyticity of $q$ whereas the integrand of the first term, by the assumed H\"older continuity of $Q$ and analyticity of $q$, can be estimated via {$\lesssim$}~$\|q\|_{C^1} \cdot D \cdot D^\sigma$ and the claim follows.
\qed

\bigskip

Therefore we may write 
\[ \oint_{\Gamma_\e} q \cdot Q~dz = \sum_k \oint_{\gamma_k} q \cdot Q~dz := \sum_k q(z_k) \cdot \oint_{\gamma_k} Q~dz + \sum_k \mathcal E_k, \]
where $z_k$ is a representative point in the region $R_k \cap \mbox{Int}(\Gamma_\e)$.  We divide the error on the righthand side into two terms, corresponding to \emph{interior} boxes -- which do not intersect $\Gamma_\e$, and boundary boxes -- the complementary set.

Let us first estimate the interior boxes.  Here, from the claim we have that the integral over each such box incurs an error of 
$\|q\|_{C^1} \cdot D^{2+\sigma}$ since here $|\gamma_k| \lesssim D$.  There are of the order $|{\rm{Int}}(\Gamma_\e)| \cdot D^{-2}$ interior boxes so we arrive at the estimate $\|q\|_{C^1}\cdot D^\sigma \cdot |{\rm{Int}}(\Gamma_\e)|$.  On the other hand, for boundary boxes, the contribution to the errors from the boundary boxes will certainly contain the original contour length $|\Gamma_\e|$.  To this we must add
{$\lesssim$}
 $D$ $\times$ [the number of boundary boxes]
corresponding to the ``new'' boundary of the boxes themselves that we  might have introduced by considering the boxes in the first place.  This is estimated as follows:

\bigskip

\noindent \textbf{Claim.}  Let $M(\Gamma_\e, D)$ denote the number of boundary boxes -- i.e., the number of boxes on the grid visited by
$\Gamma_{\varepsilon}$.
Then $M \lesssim |\Gamma_\e|/D$.

\noindent \emph{Proof of claim.} Since arguments of this sort have appeared in the literature (e.g., \cite{C1}, \cite{C2}) many times, we shall be succinct: we divide the grid into 9 disjoint sublattices each of which is indicated by its position on a $3 \times 3$ square.  Let $M_1, \dots, M_9$ denote the number of boxes of each type that are visited by 
$\Gamma_{\varepsilon}$.  We may assume without loss of generality that 
$\forall j$, $M_1 \geq M_j$.  Let us consider the coarse grained version of $\Gamma_\e$ as a sequence of boxes on the first lattice (visited by $\Gamma_\e$); revisits of a given box are not recorded until/unless a different element of the sublattice has been visited in--between.   Since the distance between each visited box is more than $D$ it follows that corresponding to each visited box the curve $\Gamma_\e$ must ``expend'' at least $D$ of its length, i.e., $|\Gamma_\e| \cdot D \geq M_1 \geq (1/9) \cdot M$ and the claim follows.  \qed

\bigskip

It is specifically observed that the additional boundary length incurred is at most comparable to the original boundary length.  In any case altogether we acquire an estimate of the order 
$ |\Gamma_\e| \cdot \|q\|_{C^1} \cdot D^{1 + \sigma}$.
 We have established 
\[ |\sum_k \mathcal E_k| \lesssim \|q\|_{C^1} \cdot D^\sigma \cdot (|{\rm{Int}}(\Gamma_\e)| + |\Gamma_\e| \cdot D).\]

Finally, by item ii) of ($\sigma, \rho$)--holomorphicity,
$$
\sum_k |q(z_k) \cdot \oint_{\gamma_k} Q~dz| 
\lesssim \|q\|_\infty \cdot \e^\rho \cdot (|{\rm{Int}}(\Gamma_\e)| \cdot D^{-1} + |\Gamma_\e|) .
$$
This follows from the decomposition similar to the estimation of the $\mathcal E_k$ terms with the first term corresponding to interior boxes and the second to boundary boxes.  The lemma been established.
\end{proof}

We can now immediately control the Cauchy integral of a $(\sigma, \rho)$--holomorphic function uniformly away from the boundary:

\begin{cor}\label{cauchy}
Let $Q_{\e}$ be part of a $(\sigma, \rho)$--holomorphic sequence
as described in Definition \ref{nhf}
above.  Let $G_{\varepsilon}(z)$ be given as the Cauchy--integral of $Q_{\varepsilon}$ 
 -- as in Eq.(\ref{FS}) -- over some (discrete Jordan) contour 
$\Gamma_{\e}$.  Let $z$ denote any lattice point in ${\rm Int}(\Gamma_{\varepsilon})$ 
such that 
$$
{\rm dist}(z, \Gamma_\e) \geq n^{-a_5}:=d_1
$$ 
for some $a_5 > 0$ and let $D \gg \e$ (both considered small).  Then for all $\e > 0$ sufficiently small, and any $d_{2} < d_{1}$, 

\begin{equation}
\label{fms}
\begin{split} 
|G_{\varepsilon}(z) - Q_{\varepsilon}(z)| &= |\frac{1}{2\pi i} \oint_{\Gamma_{\e}} (Q_{\varepsilon}(\zeta) - Q_{\varepsilon}(z))\cdot \frac{1}{\zeta - z}~d\zeta~|\\
&\lesssim \left(\frac{\e^\rho}{d_{2} D} + \frac{D^\sigma}{d_{2}^2}\right) \cdot (|{{\rm{Int}}}(\Gamma_\e)| + |\Gamma_\e| \cdot D) + 
\left(\frac{d_{2}}{d_{1}}\right)^{\sigma}
\end{split}
\end{equation}
\end{cor}
\begin{proof}
This is the adaptation of standard arguments from the elementary theory of analytic functions to the present circumstances.  Let $\gamma_{d_{2}}$ denote an approximately circular contour 
that is of radius $d_{2}$ and which is
centered at the point $z$.  Let $\Gamma_{\e}^{\prime}$ 
denote the contour $\Gamma_{\e}$ together with 
$\gamma_{d_{2}}$ -- traversed backwards -- and a back and forth traverse connecting the two.  We have, by Lemma \ref{panal}, that
$$
 |\frac{1}{2\pi i} \oint_{\Gamma_{\e}^\prime} 
 (Q_{\varepsilon}(\zeta))\cdot \frac{1}{\zeta - z}~d\zeta~|
 \lesssim \left(\frac{\e^\rho}{d_{2} D} + \frac{D^\sigma}{d_{2}^2}\right) \cdot (|{{\rm{Int}}}(\Gamma_\e)| + |\Gamma_\e| \cdot D)
$$
where, in the language of this lemma, 
we have used $\|q\|_\infty \lesssim d^{-1}_{2}$ and 
$\|q\|_{C^1} \lesssim d^{-2}_{2}$.
Thus we write
$$
G_{\e}(z) = \frac{1}{2\pi i}\oint_{\gamma_{d_{2}}}\frac{Q_{\e}(\zeta)}{\zeta - z }d\zeta + \mathcal E_{2},
$$
with $|\mathcal E_{2}|$ bounded by the right hand side of the penultimate display.  
So, subtracting $Q_{\e}(z)$ in the form
$$
Q_{\e}(z)  =  
\frac{1}{2\pi i}\oint_{\gamma_{d_{2}}}\frac{Q_{\e}(z)}{\zeta - z }d\zeta,
$$
we have that 
$$
|G_{\varepsilon}(z) - Q_{\varepsilon}(z)| \lesssim |\mathcal E_{2}|  +
\frac{1}{2\pi}\oint_{\gamma_{d_{2}}}\frac{|Q_{\e}(z) - Q_{\e}(\zeta)|}{\zeta - z }d\zeta
$$
and the stated result follows immediately from the H\"older continuity of 
$Q_{\e}$.
\end{proof}

By inputing information on $|\partial \Omega_n^\square|$, the required Cauchy--integral estimate now follows:

\noindent\textit{Proof of Lemma \ref{ci}.}  We first recall the statement of the lemma: 

\textit{Let $\Omega_n^\square$ and $S_n^\square$ be as in Proposition \ref{pr} so that 
\[|\partial \Omega_n^\square| \leq n^{\alpha(1-a_1)},\]
where $M(\partial \Omega) < 1 + \alpha$. 
For $z \in \Omega_n^\square$ (with the latter regarded as a continuum subdomain of the plane) let 
\begin{equation}
\label{FS}
F_n^{\square}(z) = \frac{1}{2\pi i} \oint_{\partial \Omega_n^\square} \frac{S_n^\square(\zeta)}{\zeta - z}~d\zeta.  
\end{equation}
Then for $a_1$ sufficiently close to 1 there exists $0 < \beta < \sigma, \rho$ and some $a_5 > 0$ such that for all $z \in \Omega_n^\square$ so that $\mbox{dist}(z, \partial \Omega_n^\square) > n^{-a_5}$,
$$
\left|S_n^\square(z) - F_n^{\square}(z) \right| \lesssim n^{-\beta}.
$$
}

By Proposition \ref{sfp}, we have that the functions 
$S_n^\square(z)$ (with $\e = n^{-1}$) have the 
$(\sigma, \rho)$--holomorphic property.   In addition, we shall also have to keep track of a few other powers of $\e$, which we now enumerate:
\begin{enumerate}[i)]
\item let us define $b_1 > 0$ so that in macroscopic units we have 
\[ |\partial \Omega_N^\square| \leq \e^{-\alpha (1-a_1)} := \e^{-\alpha b_{1}};\]
\item let us define
\[ d_2:= \e^s, \]
for some $s > 0$ to be specified later;  
\item finally, we define
\[ D := \e^t, \]
where the role of $D$ will be the same as in the proof of Lemma \ref{panal} (it is the size of a renormalized block).
\end{enumerate}

Plugging into Corollary \ref{cauchy} (again with $n^{-a_5} = d_1$) 
we obtain that 
\[\begin{split} \left|S_n^\square(z) - F_n^{\square}(z) \right| &\lesssim \left(\frac{\e^\rho}{d_2 D} + \frac{D^\sigma}{d_2^2}\right) \cdot (|{{\rm{Int}}}(\partial \Omega_n^\square)| + |\partial \Omega_n^\square| \cdot D) + \left(\frac{d_2}{d_1} \right)^\sigma\\
&\lesssim \left(\e^{\rho - (s+t)} + \e^{t\sigma - 2s} \right) \cdot (1 + \e^{-\alpha b_1 + t}) + \frac{\e^{s\sigma}}{d_1^\sigma}\\
&= \e^{\rho - (s+t)} + \e^{\rho - s - \alpha b_1} + \e^{t\sigma - 2s} + \e^{(1+\sigma)t - \alpha b_1 - 2s} + \frac{\e^{s\sigma}}{d_1^\sigma}.
\end{split}\]

With $\sigma$ fixed, the parameters $s, t > 0$ and $d_1$ can be chosen so that all terms in the above are positive powers of $\e$: set $t = \lambda \sigma$, where $\lambda \in (0, 1)$ so that $\sigma > \frac{1-\lambda}{\lambda}$.  This choice of $t$ implies that $(1 + \sigma)t > \sigma > t$.
Now let $s > 0$ and $b_1 > 0$ be sufficiently small so that $2s < t\sigma$ and $\alpha b_1 < t$ so altogether we have the last two terms are positive powers of $\e$. Next take $t$ and then $s$ and $b_1$ even smaller if necessary, we can also ensure that $\rho > s+t$ and $\rho > s + \alpha b_1$.  Finally, $d_1$ can be chosen to be some power of $\e$ so that $\e^s \ll d_1$.

\qed

\section{Harris Systems}
\label{SBT}

This last section is devoted to the proof of Theorem \ref{hreg}, although the construction may be of independent interest and find further utility.

\subsection{Introductory remarks}  For many purposes, the pertinent notion of distance -- or separation -- is Euclidean; in the context of critical percolation, what is more often relevant is the \textit{logarithmic} notion of distance:  how many scales separate two points.  These matters are relatively simple deep in the interior of a domain or in the presence of smooth boundaries.  However, for points in the vicinity of rough boundaries, circumstances may become complicated.  For certain continuum problems, including, in some sense, the limiting behavior of critical percolation, there is a natural notion for a system of increasing neighborhoods about a boundary point:   the preimages under uniformization of the logarithmic sequence of cross cuts
centered about the preimage of the boundary point in question.  
This device was employed implicitly and explicitly at several points in \cite{BCL2}.  In the present context, we cannot so easily access 
the limiting behavior we are approaching.  Moreover, 
in order to construct such a neighborhood sequence at the discrete level, 
it will be necessary to work directly with $\Omega_{n}$ itself.  

We will construct a neighborhood system for each point in $\partial\Omega_{n}$ by inductively exploring the entire domain via a sequence of crossing questions.  
Our construction demonstrates (as is \emph{a posteriori} clear from the convergence of $S_n$ to a conformal map) that various domain irregularities 
e.g., nested tunnels,
which map to a small region under uniformization are, in a well--quantified way, also unimportant as far as percolation is concerned.

\subsection{Preliminary Considerations}\label{dfss}

For completeness let us first recall the setting.  Let  $\Omega \subset \mathbb C$ be a simply connected domain with 
$\text{diam}(\Omega) < \infty$
and let 
$2\Delta$ denote the supremum of the radius of all circles which are contained in 
$\Omega$.  
Further, let 
$\mathcal D_{\Delta}$ denote a circle of radius $\Delta$ with the same center as a circle for which the supremum is realized.  We will denote by $\Omega_m$ any interior discretization of $\Omega$, e.g., one of the types discussed before; we use $n^{-1}$ to denote the lattice spacing.  For 
$\omega \in\partial \Omega_{m}$
we will define a sequence of segments the boundaries of which are paths beginning and ending on $\partial \Omega_{m}$.  As a rule, these segments separate 
$\omega$ from $\mathcal D_{\Delta}$.  The dimensions of these segments will be determined by percolation crossing probabilities analogous to the system of annuli (of which these are fragments) investigated by Harris in \cite{H}.  
Notwithstanding that the regions between segments do not actually form annuli, 
we will refer to the resultant objects as \textit{Harris rings} -- or, occasionally, ring segments, annular fragments, etc.  The ultimate goal will be to ensure that Harris rings have uniform upper and, to some extent, lower bounds on their crossing probabilities (among other properties).  Moreover, these represent the essence of what must be traversed by paths emanating from $\mathcal D_{\Delta}$ which reach to ``the essential vicinity'' of the point $\omega$.  Details will unfold with the construction.  The ultimate object will be called \textit{the Harris system stationed at} $\omega$.  

We will start with the preliminary considerations of the construction.   
Let $S_0(\omega)$ denote the smallest square 
(i.e., lattice approximation thereof) which is
centered at $\omega \in \partial \Omega_m$ and whose  boundary intersects $\mathcal D_{\Delta}$.  That is to say, the boundary is approximately 
tangent to $\partial \mathcal D_\Delta$.  We set $R_0(\omega) := S_0(\omega) \cap \Omega_m$.  Successive topological rectangles $A_1(\omega), \dots, A_{k}(\omega), \dots$ which may be \textit{envisioned} as the intersection of 
$\Omega_{m}$ with a nested sequence of square annuli  centered at $\omega$
will in practice 
be constructed via a non--trivial inductive procedure: 1) there will be deformations of the shape of the annular segments; 2) the sizes of the ``smaller squares'' (i.e., the location of the ``next'' boundary) will be determined by percolation crossing probabilities; 3) the basic shape will not always be a square centered at 
$\omega$.
Nevertheless, we will call these annular (ring) fragments.

The annular fragment $A_k$ will have four boundary pieces, forming a topological rectangle; opposing pairs of boundaries will be denoted as yellow and blue with blue corresponding to a portion of $\partial \Omega_{m}$.  The rectangle $A_{k}$ will constitute an arena for exclusive crossing type events e.g., yellow crossings between the yellow boundaries and blue crossings between the blues.  A good portion of our inductive procedure involves the refinement and coloring of the boundaries.  

All points on the yellow boundaries can be connected to $\omega$ via a (self--avoiding) path in the complement of $\partial \Omega_{m}$ and in the complement of the blue boundaries.  The outer and inner yellow boundaries may be -- somewhat loosely -- defined by the stipulation that all such paths to the outer boundary must first pass through the inner boundary.    Already, it is the case that all of $\partial \Omega_{m}\cap\partial A_{k}$ is blue; indeed, envisioning $A_{k}$ as the intersection of 
$\Omega_{m}$ with an actual square annulus, some of the blue boundary will be where $\partial \Omega_{m}$ cuts through such a ring.  

Key in the initial portion of the construction is that for some
$\vartheta$ with
 $0 < \vartheta < 1/2$, it will be the case that the probability of a yellow crossing between the yellow segments of the boundaries \emph{and} the probability of a blue crossing between the blue segments of the boundaries are both in excess of $\vartheta$ (and therefore less than $1-\vartheta$).  Eventually we will forsake the lower bound 
for the yellow crossings 
in favor of an ostensibly much smaller bound pertaining to geometric properties of the annular fragments which permit yellow crossings under tightly controlled conditions.  And, eventually, we may have to consider $\vartheta$ as a small parameter.  Nonetheless all quantities will be uniform in the ultimate 
progression of fragments and in the lattice spacing $n^{-1}$ for $n$ sufficiently large.

The essence of the geometric property which we will require of the Harris regularization scheme is that 
$\omega$ -- or its relevant vicinity -- can be connected to $\mathcal D_\Delta$ via a sequence of boxes housed within the ring fragments.  The size of the boxes
is uniform within a layer and
does not increase or decrease too fast between neighboring layers.
While the box sizes may be ``small'', this will only be relative to the characteristic scale of the layer via a numerical constant which is independent of $n$ and the particulars of the fragment.  Thus, the scale of the boxes may be considered comparable to the scale of the fragments to which they belong.

Dually, $\omega$ can be ``sealed off'' from $\mathcal D_{\Delta}$ by the independent events of separating paths which have an approximately uniform probability in each segment.  Thus we envision an orientation to our constructions leading from $\mathcal D_\Delta$ to $\omega$.  (Indeed, it is this orientation which permits us to choose the appropriate components to be colored yellow at various stages of the construction.)  Moreover, from these considerations, it emerges that only the first 
$O(\log n)$ of these segments are relevant for the problems under consideration.
If $\Omega_{m}$ has a smooth boundary this would, in fact, be all of them; under general 
circumstances, as it turns out, the configurations in the region beyond the 
first $O(\log n)$ segments have negligible impact on the percolation problem at hand. 

We will describe what is fully required in successive stages of increasing complexity, but before we begin, let us dispense with some geometric and some lattice details.  While the definitions and conventions which follow are certainly not all immediately necessary -- and possibly unnecessary for an \textit{understanding} of the overall scheme --  
we have elected to display them at the outset in a place where they are readily accessible.  
The reader is invited to skim these lightly and 
later, if required, refer back to these paragraphs.

\vspace{.25 cm}

\noindent \textbf{Preliminary definitions.}
\label{ldts}
The moral behind the upcoming definitions is that all lattice details should be resolved in as organic a way as possible via the definition of the percolation properties for the model at hand.
All of the models of interest for us have representations which are hexagonal based: each model provides some \emph{smallest independent unit} (abbreviated SIU) in the sense that such a unit (a subset of the lattice) can be stochastically configured independently of one another and any smaller subset is either empty or is correlated with some neighbor.  In the case of independent hexagonal tiling, the smallest such unit is simply a hexagon whereas for the generalized models in \cite{CL} the smallest unit can be either a single hexagon or a flower (which consists of 7 hexagons).  All notions of neighborhood, self--avoiding, etc., 
are now to be thought of in terms of the intrinsic definition of connectivity associated with the underlying percolation problem.  However, it is pointed out that in the case of the models introduced in \cite{CL}, path transmissions may take place over fractions of hexagons.

\begin{figure}[ ]
\centering
\includegraphics[scale=.3]{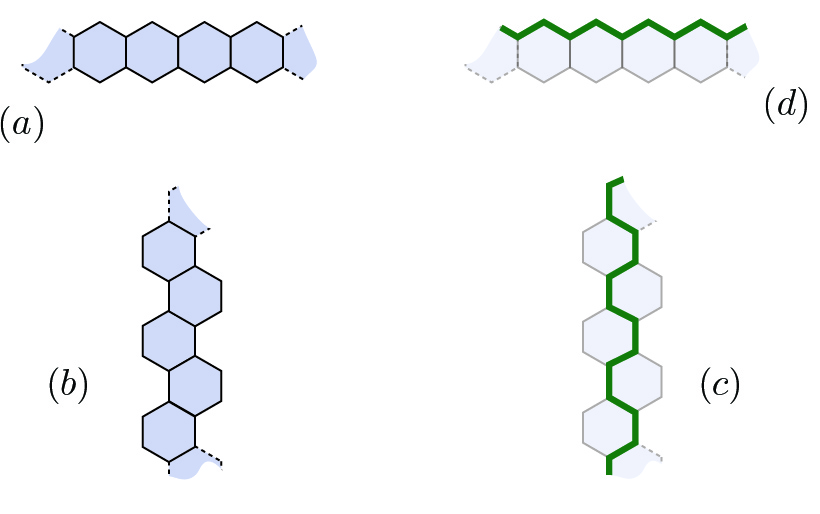}
\begin{changemargin}{.55in}{.55in}
\caption{{\footnotesize Various of notions of horizontal and vertical consistent with hexagonal tiling: $(a)$  a horizontal string $(b)$ a vertical string $(c)$ a vertical segment $(d)$ a horizontal segment. }}
\label{KXT}
\end{changemargin}
\vspace{-.3 in}
\end{figure}

\begin{enumerate}

\item  A \textit{string} of hexagons is a sequence of hexagons with no repeats each of which 
-- save the first -- has an edge in common with its predecessor.  Similarly, a \textit{segment} is a self--avoiding path  consisting of edges (which are the boundary elements of hexagons).  With apology, both objects will on occasion be referred to as paths; in all cases, the relevant notion will be clear from context and/or the distinction will be of no essential consequence.  Similarly, if 
$\mathcal A \subset \Omega_{m}$ is a connected set of hexagons 
then the boundary, $\partial \mathcal A$, could mean the 
\textit{hexagons} in $\overline{\Omega}_{m}\setminus \mathcal A$ which share an edge with an element of $\mathcal A$ or the segment consisting of the afore described shared \textit{edges}.  Again, the distinction will usually be unimportant but, if need arises will be made explicit.  

\item For $x, y\in \overline{\Omega}_{m}$, we will use 
$\mathscr P : x\leadsto y$ to denote a string in $\Omega_{m}$ connecting $x$ and $y$.  By contrast, the \textit{event} $\{x \leadsto y\}$ will mean, in the pure hexagon problem, the existence of a $\mathscr P:x\leadsto y$ such that each element of $\mathscr P$ is of the same color thereby forming a monochrome percolation connection (or transmission) between these points.  For the more general cases, 
$\{x \leadsto y\}$ also denotes this monochrome connectivity event but here some of the transmissions may take place through hexagon fragments (see, if desired, \cite{CL}, Figure 1).

\item We shall often use descriptions like horizontal and vertical, and this should be understood to mean the closest lattice approximation to either a horizontal or vertical segment.
The precise details are actually unimportant but
recall that, to be definitive, we have assumed that 
hexagons are oriented so that there are two vertical edges parallel to the $y$--axis.  For the moment, let us discuss the pure hexagon problem.  
  
  A horizontal string of hexagons
would consist of a sequence of neighboring hexagons 
the centers of which are horizontally aligned.  
Thus, each hexagon in the string shares only vertical edges in common with its neighbors or neighbor.  By contrast, a vertical string of hexagons will 
``zigzag'' a bit.  E.g., the $k^{\text{th}}$ hexagon in the sequence shares in common the segments to the left of the right vertical segment with the $(k \pm 1)^{\text{st}}$ hexagon. And similarly for the $(k + 2)^{\text{nd}}$ hexagon in the sequence.  More pertinently, a horizontal segment consists, e.g., of the tops of hexagons in a horizontal string with each hexagon in the string contributing (except, perhaps at the endpoints) the two edges above and adjacent to the vertical.  
A vertical segment alternates e.g., between left vertical edges on even hexagons and, for the odd hexagons, the entire left sides connect these (even vertical edges) together.  
Thus, horizontal segments are above or below horizontal strings and similarly vertical segments are to the left or right of vertical strings.  See figure \ref{KXT}.  Square boxes will be understood as already aligned with respect  to the fixed Cartesian axes with boundaries consisting of vertical and horizontal segments as described.  

There is a natural notion of neighboring segments, namely with a fixed \textit{string} in mind, the segment which is of (approximately) the same length that is on the other side of the string.  More formally, given a segment and associated  string, the next segment in the direction of the string is given by the symmetric difference over all hexagons in the string of the edges of each hexagon and then deleting the initial segment.  These observations will be useful when we define the concept of \textit{sliding}.  Notwithstanding, the ends of strings and segments must often be tended to on an individual basis as is the case of composite SIU.

For extended SIU -- flowers -- for now it is only important to 
state a few rules:  for strings, the entire SIU must be incorporated as a single unit.  It should be noted that this incorporation can have a variety of manifestations; e.g., in the case of horizontal segments, there are three levels at which the string may impinge.  However, all pure hexagons in a horizontal string must still have their centers horizontally aligned.  Similarly for the case of vertical strings.  
In conjunction with the above, 
segments are forbidden to pass through the flower and therefore must circumvent.  In the context of extended SIU, there will be some small amount of ``special considerations'' when we 
turn our 
discussion to distances and neighborhoods.  
\item If If $x = (x_1, x_2)$ and $y = (y_1, y_2)$ are SIU in $\Omega_m$
and $V\subseteq \Omega_{m}$ then we define the distance 
$d_{\infty}^{V}(x,y)$ to be the $L^{\infty}$-- distance, i.e., $\max\{|x_1 - x_2|, |y_1 - y_2|\}$, 
subject to the local connectivity constraint that they can actually be connected within the stated distance scale.  With precision:  if $b_{L}$ is notation for a box of side--length $2L$ then
$$
d_{\infty}^{V}(x,y) =
\min\{L\mid \exists b_L \text{ such that a string in $b_{L}$ contains $x$ and $y$}
\}.
$$
If no such $b_{L}$ exists -- i.e., if it is not the case that both $x$ and $y$ are in $V$, 
than we will regard $d_{\infty}^{V}(x,y)$ as being infinite.

If $\mathscr X$ and $\mathscr Y$ are sets in $V$, then $d_\infty^{V}(\mathscr X, \mathscr Y)$ is defined in the usual way, i.e., 
$$
d_\infty^{V}(\mathscr X, \mathscr Y) = 
\inf_{x \in \mathscr X\hspace{-2 pt},  y \in \mathscr Y}~d_\infty^{V}(x, y).  
$$

Once again, with apology, we will also use $d_{\infty}^{(\cdot)}(\cdot)$ to denote similar minded distances between pairs (or sets) of edges and, often enough, omit the superscript when it is clear from context.

\item Similar to the above item:  let $\Gamma$ denote a segment or a string.  Then 
$\|\Gamma\|_\infty$ denotes the $d_{\infty}^{\Gamma}$-- or $L^{\infty}$-- \textit{diameter} of $\Gamma$.

\item
 A \emph{topological rectangle} is a simply connected subset of $\Omega_{m}$ containing all its SIU that has four marked points on the boundary dividing the boundary into two pairs of ``opposing'' segments: 
 $(\mathscr L, \mathscr R)$ and $(\mathscr T, \mathscr B)$.  We envision such a rectangle as a stadium for a dual percolation crossing problem. 

\end{enumerate}

\begin{defn}[Sliding]\label{ddd}
Let $q$, $q^{\prime}$ denote points in 
$\overline{\Omega}_{m}$.  Let $z_{1}, z_{2}\in \partial \Omega_{m}$ and let 
$\Gamma$ denote a segment in $\Omega_{m}$, connecting
$z_{1}$ to $z_{2}$ and which separates 
$q$ and $q^{\prime}$.  Thus, $\Omega_{m}$ is disjointly decomposed into the \textit{connected components} $\mathscr C_{\Gamma}(q)$ and $\mathscr C_{\Gamma}(q^{\prime})$ of $q$ and 
$q^{\prime}$ respectively with the sets 
$\mathscr C_{\Gamma}(q)$ and $\mathscr C_{\Gamma}(q^{\prime})$ having 
$\Gamma$ as a common portion of their boundaries.  If necessary, it is assumed that $\Gamma$ is such that $\mathscr C_{\Gamma}(q)$ -- and hence 
$\mathscr C_{\Gamma}(q^{\prime})$ do not contain any partial SIU.  Moreover, to avoid spurious complications, it is assumed that the $L^{\infty}$ separation between $\Gamma$ and $q^{\prime}$ is at least a few units in excess of the $\ell$ to be discussed below.

We now define the \textit{sliding} of $\Gamma$ by $\ell$ units in the direction of 
$q^{\prime}$.  
In essence, this is the construction of a new segment -- which we denote by 
$\Gamma^{\prime}$ -- that is $\ell$ units closer to $q^{\prime}$ and, correspondingly $\ell$ units further from $q$.  The segment 
$\Gamma^{\prime}$ also connects some $z_{1}^{\prime}$ and 
$z_{2}^{\prime}$ on $\partial \Omega_{m}$  and, most importantly, also separates $q$ from $q^{\prime}$.  The segment $\Gamma^{\prime}$ in most cases can be envisioned as a displacement of 
$\Gamma$ with some natural adjustments.  We reiterate that, as was the case above, the microscopic details are not essential and may be omitted in a preliminary reading.

Let $\mathbf h_{\Gamma, q}$ denote the 
hexagons in $\mathscr C_{\Gamma}(q)$ which have at least one edge in
$\Gamma$ and for 
$h\in \mathbf h_{\Gamma, q}$, let 
$N_{\ell, \mathscr C(q^{\prime})}(h)$
denote the radius $\ell$ neighborhood of $h$ using the 
$d_{\infty}^{\mathscr C_{\Gamma}(q^{\prime})}$-- distance:
$$
N_{\ell, \mathscr C_{\Gamma}(q^{\prime})}(h) = 
\{
h^{\prime}\mid d_{\infty}^{\mathscr C_{\Gamma}(q^{\prime})}(h,h^{\prime}) \leq \ell
\}.
$$
Next we have 
$$
\tilde{N}_{\ell, \mathscr C_{\Gamma}(q^{\prime})}(\Gamma) = 
\underset {h\in\mathbf h_{\Gamma, q}}{\cup}N_{\ell, \mathscr C_{\Gamma}(q^{\prime})}(h)
$$
this set along with the completion of any partial SIU it contains may be regarded as the relevant $k$--neighborhood of $\Gamma$ and will be denoted as above without the tilde.   The ``$\ell$--slide'' of $\Gamma$ will be a subset of the \textit{edge} boundary of 
$N_{\ell, \mathscr C_{\Gamma}(q^{\prime})}(\Gamma)$.

Indeed, let us consider
$$
\tilde{\Gamma}^{\prime} = 
\partial N_{\ell, \mathscr C_{\Gamma}(q^{\prime})}(\Gamma)
\setminus (\partial \Omega_{m}\cup \Gamma)
$$
where in the above, both notions of boundary refer to edge boundaries.  It is evident that $\tilde{\Gamma}^{\prime}$ consists of one or more segments 
each of which must begin and end on the domain boundary.
Let us denote these segments by $\Gamma_{1}^{\prime}, \Gamma_{2}^{\prime}, \dots$; 
one of these will be selected as \textit{the} new segment 
$\Gamma^{\prime}$.

We now claim that exactly one $\Gamma^{\prime}_{j}$ separates $q$ from $q^{\prime}$.  Indeed let us reconsider the amalgamated $\tilde{\Gamma}^{\prime}$:
since $q^{\prime} \notin \overline{N}_{\ell, \mathscr C_{\Gamma}(q^{\prime})}(\Gamma)$ (because of the assumed distance between $q^{\prime}$ and $\Gamma$ in excess of $\ell$)
it is clear that 
$\tilde{\Gamma}^{\prime}$
separates $q^{\prime}$ from $\Gamma$ and hence certainly separates 
$q$ and $q^{\prime}$.  Thus any self--avoiding path of edges
$\Omega_{m}$ connecting $q$ and $q^{\prime}$ must intersect
$\tilde{\Gamma}^{\prime}$
and the above claim is now readily established. 

Indeed, consider the set
$\mathscr C_{\Gamma}(q^{\prime})
\setminus N_{\ell, \mathscr C_{\Gamma}(q^{\prime})}(\Gamma)$
which consists of connected components $K_{1}$, $K_{2}$, etc.  The edge boundaries of these $K_{j}$'s consist of an edge segment of $\partial \Omega_{m}$
joined together with an edge segment from $\tilde \Gamma^{\prime}$ the latter of which 
correspond to the aforementioned $(\Gamma^{\prime}_{j})$'s.  If $q^{\prime} \in \Omega_{m}$ then 
we observe that only one of these $K_{j}$ can contain it and this corresponds to the 
$\Gamma_{j}^{\prime}$ selected.  If 
$q^{\prime} \in \partial \Omega_{m}$, then, similarly, only one of the $K_{j}$ uses edges of this hexagon 
(where we again invoke the fact that the separation between $q^{\prime}$ and 
$\Gamma$ is in excess of $\ell$)
and \textit{this} corresponds to the 
$\Gamma_{j}^{\prime}$ selected. 
\end{defn}

Let us use $A$ as notation for the connected region bounded by $\Gamma$, $\Gamma^{\prime}$ and the appropriate portion of $\partial \Omega_{m}$.  Then it is noted that to within a few lattice spacings (due to inherent discreetness and the possible effects SIU) all points on $\Gamma^{\prime}$ are the same 
$d_{\infty}^{A}$-- distance from $\Gamma$.  However, it should be noted that due to the possibility of pockets that had been sealed off by the now discarded $\Gamma_{j}^{\prime}$ there might well be points in $A$ considerably further from $\Gamma$ than $\Gamma^{\prime}$.  Moreover, the contrast between 
$\Gamma^{\prime\prime}$ and $\Gamma^{\prime}$ with the former constructed via an $(\ell + 1)$--neighborhood slide will be regarded -- by definition --
as the (precursory) addition, modulo the aforementioned discreet irregularities, of a single layer to $A$.

\subsection{Preliminary Constructions}
Using the above described sliding procedure, we will start with some initial segment which separates $\omega$ from
$\mathcal D_{\Delta}$ (or some representative point therein)
and consider the sequence of slides indexed by $\ell$ in the direction of $\omega$.
We stop this procedure when certain criteria are satisfied within the region (which is a topological rectangle) whose boundary consists of the initial segment, the current slide and the relevant portions of the boundary.  This region will be referred to as a \textit{ring fragment} and the criteria will pertain to crossing probabilities within and certain geometric properties of the region itself.    When the requisite criteria are satisfied, we will refer to the current slide (or certain modifications thereof) as the 
\textit{successor} segment of the initial segment.  It is this successor segment
 which then plays the role of the ``initial'' segment when the next successor is to be defined.  

We remark that the overall construction is somewhat intricate due to the modifications to which we have already alluded.  Indeed, 
the full assemblage will actually require an \textit{inductive} procedure.  Our expository methodology is as follows: we will first describe the one--step procedure, i.e., constructing an acceptable successor segment from some given segment and then describe the full logical inductive procedure in a last subsubsection; further, the one--step procedure may in itself be complicated so we have therefore broken the construction into three stages which we will call the $S$--construction, the $Q$--construction and the $R$--construction.  We will describe them in order as they require more and more detailed control on successive Harris segments.  In all cases we will refer to the running (current) initial segment as 
$Y_{I}$ and the constructed successor segment as 
$Y_{F}^{(\cdot)}$ where $(\cdot)$ indexes the relevant modification.
At each stage, we need not make any specific claim as to the nature of 
$Y_{I}$ however, it may be generally assumed that 
$Y_{I}$ is such that it \textit{already} satisfies the criteria $(\cdot)$.  

Finally, we remark that for simple domains e.g., convex domains or domains bounded only by straight line segments, the vast majority of the forthcoming is 
unnecessary:  here we create successors as just described stopping when crossing probability within the region acquires a desired value (the simplified S--construction) -- or when a \emph{fixed} scale determined by the domain has been reached -- and then divide the region into boxes the scale of the region itself (the simplified  R--construction).  All of the up and coming  technology concerns the possibility of (multi--scale) fjords, needle--like tunnel structures (on multi--scales) and similar minded impediments.   

\subsubsection{The $S$--Construction \rm{(Harris ring fragments)}}  The initial stage is the $S$--construction, which starts with a square centered at the specified $\omega  \in \partial \Omega_m$.  Recall 
the square $S_{0}(\omega)$ centered about $\omega$
and note that since $\partial \Omega_m \cap \partial \mathcal D_\Delta = \emptyset$ whereas $\partial \mathcal D_\Delta$ has non--trivial intersection with $\partial S_0(\omega)$, we can declare the first yellow segment of the boundary of the first ring fragment to be, simply, the connected component of $\partial S_0(\omega) \cap \partial \mathcal D_\Delta$ in $\Omega_m$. We denote this yellow boundary portion by $Y_{0}$.  

Consider a a slide of $Y_0$ which we temporarily denote $\mathbb X_S$ and let $S = \mathscr C_{\mathbb X_S}(\omega)$.  For immediate use \textit{and} for future reference, let us define some auxiliary objects:
 
\noindent (a) \hspace{3 pt} 
Let $\mathbf P_{\omega}$ be the set of all paths 
$\mathscr P : \omega \leadsto Y_{0}$ defined as usual: they are self--avoiding paths (strings) consisting of hexagons with $\mathscr P \subseteq \Omega_m \cup \{\omega\}$, and $\mathscr P \cap \partial \Omega_m = \{\omega\}$.

\noindent (b) \hspace{3 pt}  Let $\mathbb Y_{S} \subseteq \mathbb X_S$ be such that any $z \in \mathbb Y_{S}$ is the last point in $S$ for \emph{some} 
$\mathscr P\in \mathbf P_{\omega}$.  Note that in general these could be SIUs.

\noindent (c) \hspace{3 pt} Let $Y_{S} \subset \mathbb Y_{S}$ be the set of points 
(edges) in
$\mathbb Y_{S}$ that can be reached from $\omega$ by a portion of a path 
$\mathscr P \in \mathbf P_{\omega}$ which lies in the complement of 
$\mathbb Y_{S}$.  

\noindent
In general a string $\mathscr P$ may be regarded as a sequence of neighboring hexagons, a sequence of centers of neighboring hexagons (points of the original triangular lattice) connected by straight line segments (bonds of the original triangular lattice) or the corresponding sequence of edges separating the elements of the hexagon description (bonds of the honeycomb lattice dual to the aforementioned bonds).  To avoid excessive clutter, here and in the future, we will be non--specific about such matters when they are of no actual consequence to the main argument.
Some illustrations which indicate generic differences 
between the sets
$Y_S$ and $\mathbb Y_S$ can be found in the \textit{inserts} of Figure \ref{TLB}.    

\begin{figure}[ ]
\centering
\includegraphics[scale=.19]{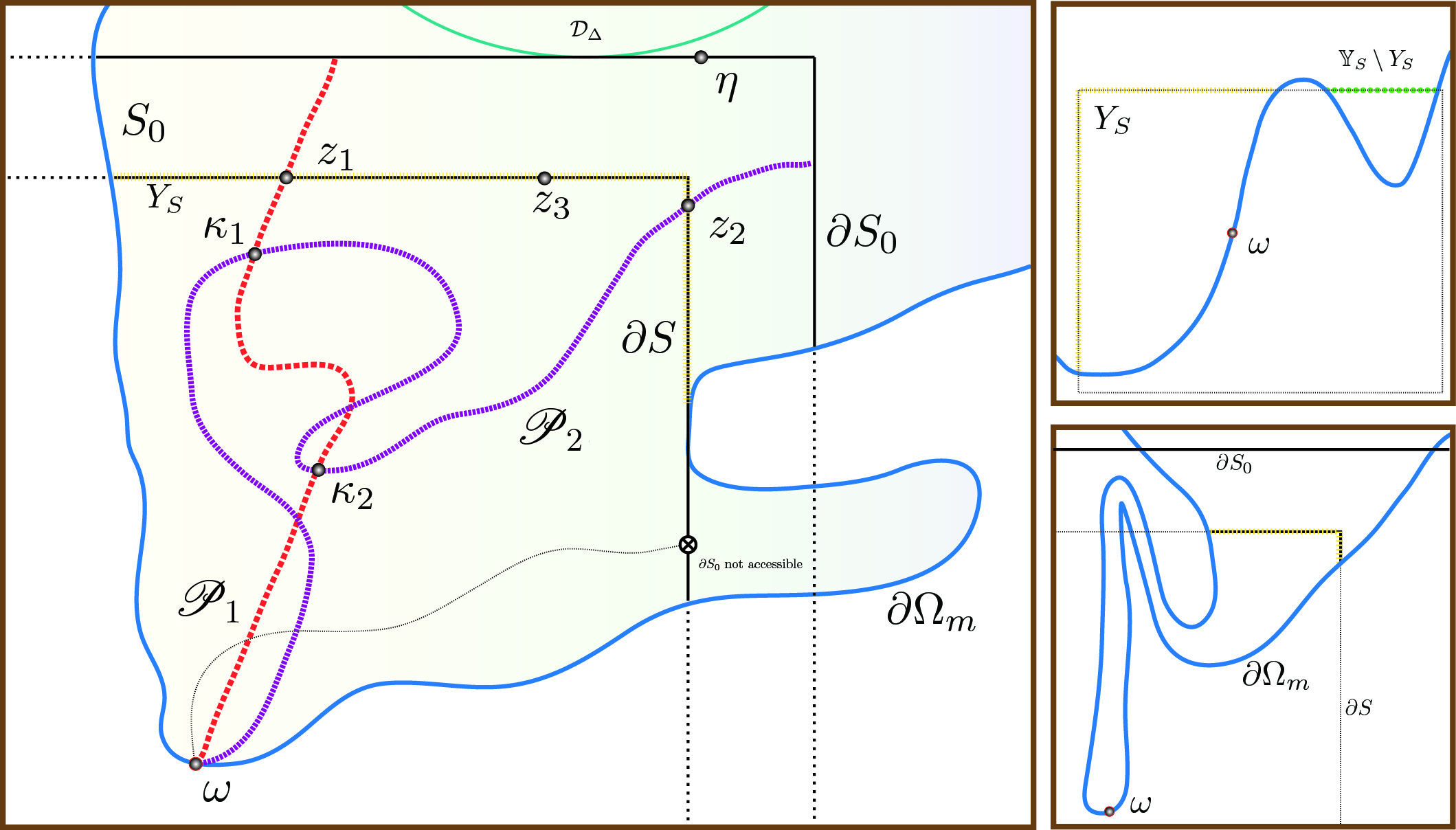}
\begin{changemargin}{.55in}{.55in}
\caption{{\footnotesize Illustration of the principal concepts for the proof that $Y_{S}$ is a connected set.
In the example illustrated, the loop $\mathscr L$ is completed using portions of $\partial S_{0}$.
Top insert:  $\mathbb Y_S \neq Y_S$.
Bottom insert:  $\mathbb Y_S = Y_S$.}}
\label{TLB}
\end{changemargin}
\vspace{-.3 in}
\end{figure}

We shall refer to these $Y_{S}$ as \textit{yellow segments}.  We will now establish some elementary topological properties of these yellow segments $Y_{S}$;
in particular, we claim that as \textit{segments}, the $Y_S$ are in fact well--defined:

\bigskip

\noindent {\bf Claim.}  The set $Y_S$ 
as described above is non--empty and consists of a single 
connected component.  Moreover, $Y_{S}$ separates $\omega$ from $Y_{0}$, i.e., every path in $\mathbf P_\omega$ contains at least one element of $Y_S$.

\bigskip

\noindent{\emph{Proof of Claim.}} To establish that $Y_{S}$ is non--empty, let 
$\mathscr P\in \mathbf P_{\omega}$ and let $z\in\mathbb Y_{S}$ denote the 
\textit{first} element of $\mathbb Y_{S}$ encountered by $\mathscr P$.  
Further, let $\mathscr P_{\text{I}}$ denote the (initial) portion of $\mathscr P$ till it reaches 
$z$.  Now, by definition of the set there is some other 
$\mathscr P^{\prime}\in \mathbf P_{\omega}$ (which is possibly the same as $\mathscr P$)
such that $z$ is the last element of $\mathscr P^{\prime}$ which is in $S$.  
We denote by $\mathscr P_{\text{F}}^{\prime}$ the (final) portion of 
$\mathscr P^{\prime}$ after it has passed through $z$.  Now consider the path 
$\mathscr R$ which may be informally described as 
$\mathscr P_{\text{I}}\cup\{z\}\cup\mathscr P_{\text{F}}^{\prime}$.  
It is seen that $\mathscr R$ satisfies the requisite properties of a path described in
item $($c) above for the point $z$ and so, indeed, $Y_{S}$ is not empty.   

An argument along these lines also demonstrates that $Y_{S}$ separates 
$\omega$ from $Y_{0}$:  again letting 
$\mathscr P \in \mathbf P_{\omega}$, we can consider $z$ as above.  The same argument then shows that $z \in Y_S$ and thus we conclude that \textit{every} path in 
$\mathbf P_{\omega}$ contains an element of $Y_{S}$.  

We now turn to the central portion of the claim, namely that $Y_{S}$ is a single component.  To this end, let us consider the ancillary paths which are along the lines of $\mathscr R$ above.  This set of paths will be called 
$\mathbf P_{S}$ and each $\mathscr P \in \mathbf P_{S}$
has the property that it has a unique encounter with $\mathbb Y_{S}$ -- which is necessarily a single element $z$ in 
$Y_{S}$.   
So, in particular, each $\mathscr P\in \mathbf P_{S}$ divides into
a $\mathscr P_{\text{I}}$ -- before the encounter with $Y_{S}$, a 
$\mathscr P_{\text{F}}$ -- after the encounter with $Y_{S}$, and a point $z$ in 
$Y_{S}$ itself.  

Foremost, it is clear by construction of paths along the lines of $\mathscr R$ 
that any $z \in Y_{S}$ has a $\mathscr P\in \mathbf P_{S}$ which goes through $z$. 
We further define $V_{S}$ to be the connected component of $\omega$ relative to the 
$\mathscr P_{\text{I}}$'s just described:
$$
V_{S} = \{
v\in \Omega_{m}\mid v\in\mathscr P_{\text{I}} 
\text{ for some } \mathscr P_{\text{I}} \subset \mathscr P \in \mathbf P_{S}
\}
$$
Similarly, we define $U_{S}$ to be the points reachable by some 
$\mathscr P_{\text{F}} \subset \mathscr P \in \mathbf P_{S}$.  Formally, we define both sets as disjoint from $Y_{S}$ and adjoin $\omega$ to $V_{S}$.  

Now it is noted that $\partial \Omega_m$ may in general divide $\partial S$ into many components so we are showing that only one of these contains $Y_{S}$.  
We refer the reader to Figure \ref{TLB} for the up and coming argument.  
We must demonstrate that any  
$z_{1}, z_{2} \in Y_{S}$ must be in the same component.  
We may as well assume that $z_{1} \neq z_{2}$ and that they are not neighbors, since otherwise it is immediate 
that they are in the same component.

Consider paths  $\mathscr P_1, \mathscr P_2$, in $\mathbf P_{S}$
which
contain $z_1$ and $z_2$, respectively.  The initial portions of these paths 
$\mathscr P_{\text{I},1}$ and $\mathscr P_{\text{I},2}$ may very well intersect, so let 
$\kappa_{1}$ denote the last such point along $\mathscr P_{\text{I},1}$ and similarly for $\kappa_{2}$; if there are no such point along the paths proper, then we define $\kappa_{1} = \kappa_{2} = \omega$.  Nevertheless, we certainly allow  the possibility that $\kappa_{1} = \kappa_{2}$.
Assuming though that 
$\kappa_{1} \neq \kappa_{2}$, there is (for the sake of definitiveness) a portion of 
$\mathscr P_{\text{I},1}$  connecting $\kappa_{2}$ to $\kappa_{1}$ since, by definition, $\kappa_{1}$ was the last such ``collision''  point along $\mathscr P_{\text{I},1}$.
Thus, starting at $z_{2}$, moving (backwards) along $\mathscr P_{2}$ till $\kappa_{2}$ is reached then moving (forwards) along $\mathscr P_{1}$ through $\kappa_{1}$ and on to $z_{1}$ we have achieved a path
$z_{2}\leadsto z_{1}$ which is entirely in $V_{S}$.  This latter path we regard as one part of a loop.

On the other side of $Y_{S}$, in $U_{S}$, we follow a similar procedure and define, in an analogous fashion, the points $\eta_{1}$ and $\eta_{2}$.  Again, if no such points exist along the paths proper,
then we join these paths together using the the relevant portion of $Y_{0}$.  This constitutes another path 
$z_{1}\leadsto z_{2}$ which is disjoint from the first.  This is the second part of the loop the entirety of which 
-- including $z_{1}$ and $z_{2}$ -- 
is denoted by $\mathscr L$.  Now, with the possible exception of $\omega$ itself, 
$\mathscr L$ is disjoint from $\partial \Omega_{m}$. Thus, $\text{Int}(\mathscr L) \subset \Omega_{m}$
and so the entire portion of $Y_S$ in between $z_{1}$ and $z_{2}$ lies within.  

Let $z_{3}$ be one such point in between.  Consider any path residing in 
$U_{S}$
from $z_{3}$ to the second half of the loop.  (Such a path may be acquired by attempting a path 
$z_{3}\leadsto \partial S_{0}$ -- in $U_{S}$ -- until obstructed.)  Joining this path with the relevant portions of, 
$\mathscr P_{\text{F},1}$ and/or $\mathscr P_{\text{F},2}$
 we now \textit{have} a path  $z_{3}\leadsto \partial S_{0}$ in 
$U_{S}$.  Performing a similar construction in $V_{S}$ we get a path $z_{3}\leadsto \omega$ in $V_{S}$ and putting these two together we have constructed a  \textit{bona fide} $\mathscr P_{3} \in \mathbf P_{S}$ 
which contains $z_{3}$.  We have therefore demonstrated $z_{3} \in Y_{S}$ as desired. \qed
 
%
%
%
%

%
%
%
%

\bigskip
Next we may consider a further successor segment $Y_{S'}$ of $Y_S$ and the corresponding component of $\omega$, $S'$. Then we have the following ``partial ordering'' property:
\bigskip

\noindent{\bf Claim.}
Let $S$, $S^{\prime}$, $Y_{S^{\prime}}$ etc. be as described.  Then $Y_{S}$ separates $Y_{S^{\prime}}$ from $Y_{0}$ and $Y_{S^{\prime}}$ separates $Y_{S}$ from $\omega$.  A similar result holds for intermediate $Y$-- segments.  
E.g., if $Y_{S^{\dagger}}$ is an earlier successor of $Y_{0}$ than $Y_{S}$, then 
$Y_{S^{\dagger}}$ separates $Y_{0}$ from $Y_{S}$
or 
if $Y_{S^{\dagger}}$ is an intermediate successor segment between 
$Y_{S^{\prime}}$ and $Y_S$, then $Y_{S^{\dagger}}$
separates $Y_{S}$ from $Y_{S^{\prime}}$ etc.

%

\bigskip

\noindent{\emph{Proof of Claim.}}  This follows readily from the ideas invoked in the proof of the previous claim.  Consider $\mathscr P : Y_{S^{\prime}} \leadsto Y_{0}$ and define
$\mathscr P_{\text{F},S^{\prime}}$ to be the portion of this path after its final exit from 
$Y_{S^{\prime}}$ -- which takes place at some $z^{\prime}$.  Now adjoin to this an
``initial'' path, $\mathscr P_{\text{I},S^{\prime}}$ (of the type described in the proof of the previous claim) so that $\mathscr P_{\text{I},S^{\prime}} \cup z^{\prime} \cup \mathscr P_{\text{F},S^{\prime}}$ is a path in 
$\mathbf P_{S^{\prime}}$.  
It is asserted that $\mathscr P_{\text{I},S^{\prime}}$ can \textit{not} meet $Y_{S}$: indeed, supposing to the contrary, we could adjoin its portion from $\omega$ to $Y_{S}$ with an appropriate 
$\mathscr P_{\text{F},S}$ -- and a $z\in Y_{S}$.  Now since the latter two take place, essentially, in $S^{c}$, this composite object would represent a path $\omega \leadsto Y_{0}$ which circumvents $Y_{S^{\prime}}$ altogether which is impossible since $Y_{S^{\prime}}$ separates $\omega$ from $Y_{0}$ by a result of the previous claim.  

However, the path $\mathscr P_{\text{I},S^{\prime}} \cup z^{\prime} \cup \mathscr P_{\text{F},S^{\prime}}$
is a path from $\omega$ to $Y_{0}$ so it must intersect $Y_{S}$ somewhere; by the above, this has to be in 
$\mathscr P_{\text{F},S^{\prime}}$ which, it is recalled was part of the original $\mathscr P$.  
With regards to intermediate successors, e.g., if $Y_{S^{\dagger}}$ 
is an earlier successor of $Y_{0}$ than $Y_{S}$, we observe that as a consequence of their mutual construction, 
$S \subset S^{\dagger}$,
and the preceding argument may be taken over directly.  
Further generalities discussed in the statement of this claim are handled similarly.
\qed

\bigskip

Given $Y_0$ and its successor $Y_S$, the blue boundary is defined to be the portions of $\partial \Omega_m$ connecting the endpoints of 
$Y_0$ and $Y_S$. The \emph{topological rectangle} formed by the blue boundaries and $Y_0, Y_S$ will be denoted by an $A$ (an annular fragment).

The size of the annular fragment 
will be adjusted, if possible, to keep the probability of a yellow crossing
$Y_{0} \leadsto Y_{S}$ in the range $(\vartheta, 1-\vartheta)$ with $\vartheta$ a constant of order unity 
independent of $n$, $A$, etc.,
with a constraint to be described below (but eventually to be taken as ``small'').  
This will be attempted by using the sliding scale construction described in Definition \ref{ddd}; in essence, we advance by single lattice units till the crossing probability achieves a value in the desired range.  Thus, we arrive at a sequence of 
ring fragments $A_{1}, A_{2}, \dots$ with, essentially, uniform crossing probabilities of both types.  Unfortunately,  as will emerge (and as is not hard to envision) for the general cases, the procedure may fail.  Then we must engage the more complicated constructions described in the rest of this subsubsection.  

For the time being, let us assume then that $A_{1}, A_{2}, \dots, A_{k-1}$
have been constructed with the desired properties so far mentioned; we will go about constructing $A_{k}$.  
In the present stage of the construction it is, as part of the assumption, the case that we have already acquired the yellow boundaries $Y_{0}, \dots Y_{k-1}$.  We investigate the sequence $X_{k,1}, X_{k,2}, \dots$ of successively progressing slides starting with $Y_{k-1}$ (in the direction of $\omega$).  This necessarily leads to the consideration of a double sequence of temporary boundaries -- the slides themselves -- and the \textit{associated} yellow segments, 
$Y_{k,1}, Y_{k,2}, \dots$; these are both temporary especially the first (temporary temporary)
compared with the second which is more legitimate (legitimate temporary).  It is noted that
$Y_{k,\ell} \subseteq X_{k,\ell}$ where, it is emphasized, the inclusion is generically strict.  
These form two sequences of temporary topological rectangles.  
Elements of the first sequence will be denoted by
$A_{k,\ell}$ which has yellow boundaries 
$Y_{k-1}$ and $Y_{k,\ell}$ and those of the second by 
$\tilde{A}_{k,\ell}$ which has in the stead of the second yellow boundary, the successor set $X_{k,\ell}$.  Our pertinent crossing questions actually concern the ring fragments 
$A_{k,1}, A_{k,2}, \dots$ (which are of the \textit{legitimate} nature).  Again, for ease of exposition, let us assume that the first 
$\ell - 1$ renditions have yellow crossing probabilities in excess of $1-\vartheta$, i.e., suppose
$$
\mathbb P(\mathscr P: Y_{k-1}\leadsto Y_{k,\ell -1}) = \kappa > 1-\vartheta.
$$
Then, we claim, that for $\vartheta$ chosen appropriately, the yellow crossing in 
$\tilde{A}_{k,\ell}$ is not too small:
$$
\mathbb P(\mathscr P: Y_{k-1}\leadsto X_{k,\ell})  > \vartheta.
$$
Indeed, conditional on the existence of a yellow crossing of the type described in 
$A_{k,\ell -1}$, it is only necessary to attach one more (SI) unit of yellow to achieve the desired connection up to $X_{k,\ell}$.  In the usual hexagon tiling problem, this 
occurs with probability 1/2; in general, to include the models introduced in 
\cite{CL}, let us say that this occurs with probability $r$.  Then we have
$$
\mathbb P(\mathscr P: Y_{k-1}\leadsto X_{k,\ell}) \geq (1-\vartheta)\cdot r 
\geq \vartheta
$$
if $\vartheta$ satisfies
\begin{equation}
\label{ERT}
\vartheta < \frac{r}{1+r}.
\end{equation}

It may be \textit{envisioned} that the process proceeds smoothly till eventually we get an $\ell$ large enough so that the desired yellow crossing probability falls into the range
$(\vartheta, 1 - \vartheta)$ -- even if this occurs, we are far from done.  
However, we have already allowed for an obstruction which is akin to an
``uncontrolled approximation'', namely the difference between the (temporary, temporary)
yellow boundaries $X_{k,\ell}$ and the (legitimate, temporary) yellow boundaries $Y_{k,\ell}\subseteq X_{k,\ell}$.  As is clear from previous discussions -- and has, in fact been a basis of the derivation -- it is the legitimate, temporary yellow boundary that must actually be considered.  In particular, we must account for the possibility that
the yellow crossing probability in 
$\tilde{A}_{k,\ell}$ is in $(\vartheta, 1-\vartheta)$ -- or even lies above 
$1 - \vartheta$ -- but when we replace the yellow boundary
$X_{k,\ell}$ with the smaller $Y_{k,\ell}$, the crossing probability in
${A}_{k,\ell}$ falls below $\vartheta$.  

These circumstances would tend to occur if 
$\omega$ lies deep inside a narrow tunnel which opens into a wider central region
where the sliding is currently taking place; when the slide reaches the mouth of the tunnel, huge portions of the segment can be lost.
In this case, the appropriate action is in fact to slide backwards: indeed, let us note that this discontinuity is engendered by the fact that not all legitimate crossings from $Y_{k-1}$ to $Y_{k, \ell - 1}$ can be extended by one SIU to a path to $Y_{k, \ell}$.  On the other hand, it is the case that all points on $Y_{k, \ell}$ are one SIU away from some point in $Y_{k, \ell - 1}$; this asymmetry is inherent in the definition of successors.  Thus, continuity can be ensured by backsliding $Y_{\ell}$:

\begin{lemma}[\rm{Sliding Scales.}]\label{slid}
Let $Y_k, Y_{k, \ell}$, etc., be as described such that $\mathbb P(Y_{k-1} \leadsto Y_{k, \ell}) < \vartheta$.  Consider $F_0^{(k)} \equiv Y_{k, \ell}, F_1^{(k)}, F_2^{(k)}, \dots$ successive backward slides (i.e., in $\mathscr C_{Y_{k-1}}(\mathcal D_\Delta)$) of $Y_{k, \ell}$.  Let the corresponding yellow segments (relative to $Y_{k-1}$)  $G_j^{(k)} \subseteq F_j^{(k)}$ be constructed as above so that all of $G_j^{(k)}$ is accessible from $Y_{k-1}$.  Then it is the case that successive crossing probabilities do not increment too fast, i.e., 
$$
\mathbb P(Y_{k-1} \leadsto G_{j-1}^{(k)} \mid Y_{k-1} \leadsto G_j^{(k)}) > r.
$$
\end{lemma}
\begin{proof}
First note that by definition $F_0^{(k)} \equiv G_0^{(0)}$ and each SIU in $G^{(k)}_j$ is connected to \emph{some} SIU in $G^{(k)}_{j-1}$ and so a crossing to $G_{j-1}^{(k)}$ necessarily implies a crossing to $G_{j}^{(k)}$.  (We re--remind the reader that here we are sliding \emph{backwards}, so the preceding are exactly the desired connectivity relation.)  It follows, by attaching one more yellow SIU, that $\mathbb P(Y_{k-1} \leadsto G_{j-1}^{(k)} \mid Y_{k-1} \leadsto G_j^{(k)}) > r$.  
\end{proof}

For later purposes, we will need a slightly more complex procedure than just sliding backwards.  Supposing that there is a jump in crossing probabilities as described e.g., at the $2\ell^{\mbox{\tiny th}}$--stage, let us consider e.g., the segment $Y_{k, \ell}$; by the previous separation claim, this segment separates $Y_{k-1}$ from $Y_{k, 2\ell}$.  We backslide as described above, however, we use $Y_{k, \ell}$ as a ``barrier'': i.e., we only consider the connected component of our neighborhoods in the region bounded by $Y_{k, \ell}, Y_{k, 2\ell}$ and the relevant portions of $\partial \Omega_m$.  In point of fact, we will not exactly use $Y_{k, \ell}$ but a certain modification thereof.  Nevertheless, the above procedure works for any $Z_k$ (replacing $Y_{k, \ell}$ as a barrier) which separates $Y_{k-1}$ from $Y_{k, 2\ell}$ and has the property that $\mathbb P(Y_{k-1} \leadsto Z_k) > 1 - \vartheta$.  Notwithstanding, we will still call the backward slides $G_j^{(k)}$'s (suppressing the $Z_k$ dependence).

In any case, on the basis of the above lemma, we may now assert that there exists an $m$ such that 
$$
\mathbb P(Y_{k-1} \leadsto G_m^{(k)}) \in (\vartheta, 1- \vartheta).
$$
Indeed, the result of the lemma remains true under the modified procedure and so if it is the case that $\mathbb P(Y_{k-1} \leadsto G_{j-1}^{(k)}) < \vartheta$, then  
$$ 
r \cdot \mathbb P(Y_{k-1} \leadsto G_j^{(k)}) < \vartheta, 
$$
and hence $\mathbb P(Y_{k-1} \leadsto G_j^{(k)}) < 1 - \vartheta$ if $r$ satisfies Equation \eqref{ERT}.  This leads to the existence of the stated $m$.

We will then promote $G_m^{(k)}$ into $Y_k$.  As alluded to before, there are additional modifications to be performed on $Y_k$, but nevertheless we are finished with the $S$--construction.


\subsubsection{The $Q$--Construction \rm{(effective regions)}}
The $S$--construction is not sufficient to capture certain irregularities which may be present in the domain $\Omega$ -- nor to achieve our purpose.  These problems manifest themselves on two levels: the successive yellow regions may be vastly different in length, as can be caused by a narrow tunnel suddenly leading to a wide region.
On a more subtle level, there are cases where the $S$--construction yields consecutive yellow regions which are of comparable size but the ``effective'' yellow region where the crossing would  actually take place is in fact much smaller.  This is indicative of ``pinching'' of 
$\partial \Omega_m$ in the vicinity of the current segment.  In any case, the problem here is that, in essence, the process is proceeding too quickly.  Thus we will reduce the relevant scales in order to slow the growth of the evolving neighborhood sequence and possibly perform some further ``backwards'' steps.
To a \emph{first approximation}, given $Y_I$, a successor $Y_F$ is not utilizable for us if it is the case that it is too large relative to the separation between $Y_I$ and $Y_F$.  In this case we will instead consider future segments grown around some ``effective region'' determined by some \emph{sub}segment of $Y_F$.  The purpose of this subsubsection is to modify the segment $Y_F$ into some $Y_F^{(e)}$ (representing the effective region) so that it has the correct aspect ratio relative to $Y_{I}$ ($\equiv Y_I^{(e)}$).  We will make the notion of ``effective region'' precise but first we need a proposition concerning crossing probabilities of more general topological rectangles:

\begin{prop}\label{btt}
Let $A$ denote a topological rectangle with sides $(\mathscr L, \mathscr R)$ and $(\mathscr T, \mathscr B)$, respectively.  Consider the following notion of aspect ratio: with 
$$
a :=  d_\infty^{A}(\mathscr T, \mathscr B),~~~b := d_\infty^{A}(\mathscr L, \mathscr R),
$$
we set 
$$
B = B_A = b/a.
$$ 
Let $\eta > 0$, then for any $A ~(= A_n)$, there exists some $B(\eta)$ such that if $B_{A} > B(\eta)$, then for the critical percolation problems of interest in this work,
$$
\mathbb P(\mathscr L \leadsto \mathscr R) < \eta
$$
with the above uniform in $n$.
\end{prop}
\begin{proof}
Let $x_*$ be a point on $\mathscr T$ 
and $y_{*}$ a point on $\mathscr B$ where the infimum defining $a$ is realized and consider a box, denoted by $G_{a}$, of side $2a$ containing these points \textit{and} a path $\mathbf p_{a}$ in $A$ connecting them.  It is noted that $\mathbf p_{a}$ separates $\mathscr L$ and $\mathscr R$.  Now consider a box, denoted by $G_{b}$ which is of side $2\lambda b$ for
$\lambda \lesssim 1$ with the same center as the aforementioned box; it is assumed that $\lambda b > a$.    
By the definition of the $d_{\infty}^{A}$-- distance, it is the case that any path in $A$ connecting $\mathscr L$ to $\mathscr R$ must have a portion outside $G_{b}$ -- although $\mathscr L$ and $\mathscr R$ may themselves lie within.  It is now claimed that any circuit in the annulus $G_b\setminus G_a$ separates $\mathscr L$ from $\mathscr R$.  Indeed,
let $\mathscr P: \mathscr L \leadsto \mathscr R$ denote a path in A.  It may be assumed e.g., that 
$$
\mathscr P: \mathscr L \leadsto \partial G_{b} \leadsto \mathbf p_{a} \leadsto \mathscr R
$$
(the other case is identical).  Therefore a portion of $\mathscr P$ connects 
$\partial G_b$ to $\partial G_a$ and thus meets the circuit as claimed.  

Now, by standard critical 2D percolation arguments (now so--called RSW arguments) there are of the order $\log (b/a)$ independent chances of creating a circuit in $G_b\setminus G_a$ thus preventing a crossing between $\mathscr L$ and $\mathscr R$ and the desired result is established.
\end{proof}

Henceforth, to avoid inconsequentials, we shall assume that $\vartheta$ has been chosen small enough so that all relevant $B$'s are certainly greater than two or three.  We are now ready to prove:
 
\begin{lemma}[Aspect Ratio Estimate.]\label{er}
Let $Y_I$ denote some (initial) segment and $Y_F$ a successor segment which is constructed as described previously.  Let $A_{F}$ denote the topological rectangle of relevance -- of which these are two of the sides -- and within which the crossing problem of current interest takes place.  In particular it is assumed that $\mathbb P(Y_I \leadsto Y_F) \in (\vartheta, 1- \vartheta)$.  Moreover, it is assumed that $Y_F$ is obtained from $Y_I$ by the direct sliding construction, i.e., without invoking the backsliding as was featured in Lemma \ref{slid} and the discussions immediately following.
Let us denote the separation distance between $Y_I$ and $Y_F$ by $J_F$, i.e., 
$J_F := d_\infty^{A_{F}}(Y_I, Y_F)$.
Then there exists some $B$ with $1 < B < \infty$ and a modification of $Y_F$, which we denote by $Y_F^{(e)}$, such that 
$$
B^{-1} \cdot J_F \leq \|Y_F^{(e)}\|_\infty \leq (3B + 2) \cdot J_F + c^{\prime} \leq \kappa'B \cdot J_F
$$
for some constant $c^{\prime}$ and $\kappa'$
and 
$$
\theta'' \leq \mathbb P(Y_I \leadsto Y_F^{(e)}) \leq 1 - \theta'',
$$
where $\theta'' = \vartheta - \vartheta^2$. 
\end{lemma}

\begin{proof}
Let $B$ be chosen from Proposition \ref{btt} corresponding to $\eta = \vartheta^2$.  Let us suppose that $\|Y_F\|_\infty > 3B \cdot J_F$. We define $x_\ell$ and $x_r$ as the extreme ``left'' and ``right'' endpints of $Y_F$ and similarly define $y_\ell$ and $y_r$ for $Y_I$.  Consider the path
 $\mathscr B_{\ell, F}$ which is the ``left'' boundary of $A_{F}$ and
which is a portion of
$\partial \Omega_m$ that starts at $y_\ell$ and proceeds to $x_\ell$.  To avoid future clutter, we shall omit the $F$ subscript and, when clear from context write 
$\mathscr B_{\ell} := \mathscr B_{\ell, F}$.  
Similarly we have the other blue boundary which is the path $\mathscr B_{r}$ starting with $x_r$ and proceeding to $y_r$.

The essence of the argument is captured in the simplified version where $Y_I$ and 
$Y_F$ are (essentially) horizontal regions.  In order to expedite the overall process, we shall first treat this case.  The general case -- which can certainly be omitted on a preliminary reading -- will be attended to subsequently.  

\emph{Simplified Version.}   Note that due to the constraint on all $d_\infty^{A_{F}}$-- distances and envisioning $B \gg 1$, $Y_I$ and $Y_F$ must have considerable ``overlap''.  Let us now define an auxiliary curve $\Gamma_\ell$: in this simplified setup, $\Gamma_\ell$ is defined to be the leftmost perpendicular extension from $Y_F$ to $Y_I$ which does not intersect $\mathscr B_{\ell}$ inside the region bounded by $Y_I$ and $Y_F$ and similarly define $\Gamma_{r}$.  (It is noted that despite the nomenclature, $\Gamma_{\ell}$ may actually be to the ``right'' of 
$\Gamma_{r}$.)  See Figure \ref{PFG} which will also be useful for the up and coming.  

\begin{figure}[ ]
\centering
\includegraphics[scale=.27]{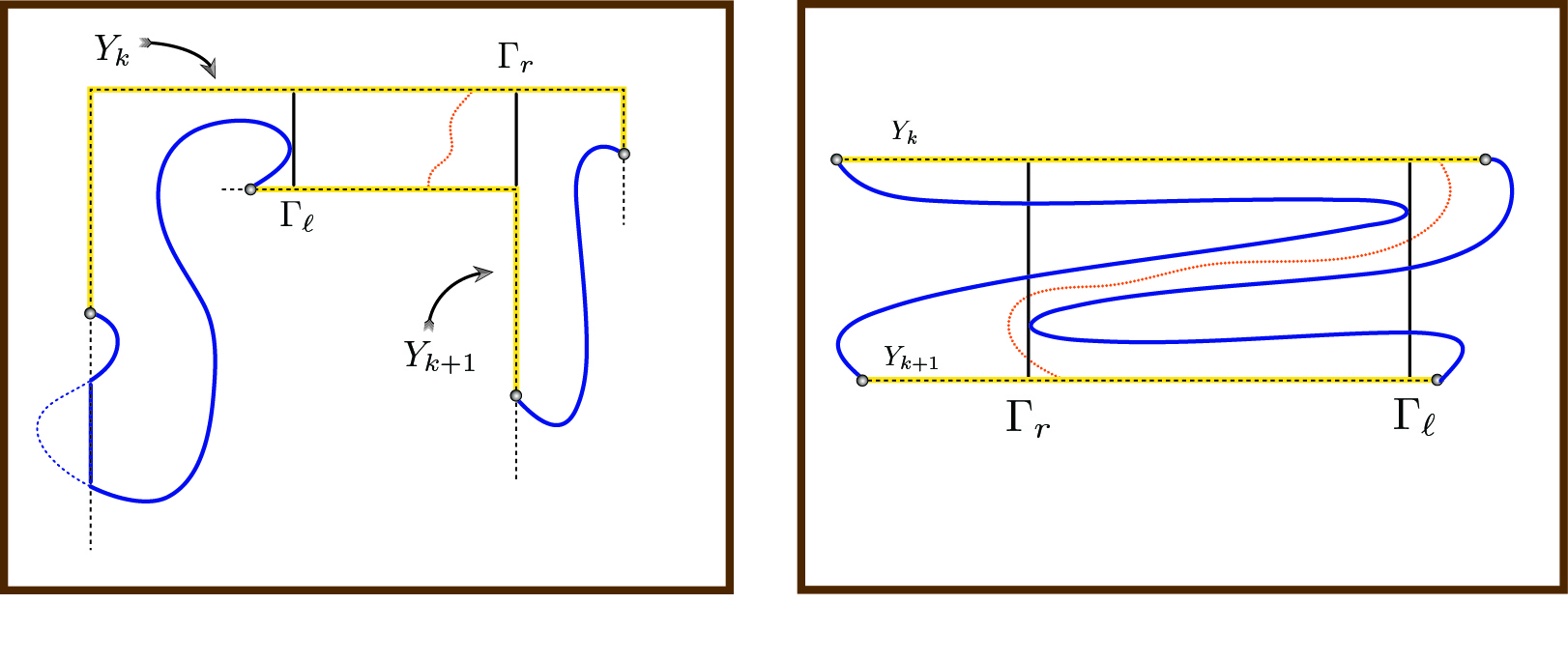}
\begin{changemargin}{.55in}{.55in}
\caption{{\footnotesize Setup described in the proof of Lemma \ref{PFG} under simplified assumption that segments are fragments of squares.  The constant $B$ here is envisioned to be about 2.  Left:
$\Gamma_{\ell}$ and $\Gamma_r$ situated as anticipated;
$Y_{k} \leadsto Y_{k+1}$ is too likely.  Right:  $\Gamma_{\ell}$ and $\Gamma_r$
switched from anticipated placement; here 
$Y_{k} \leadsto Y_{k+1}$ is too unlikely.}}
\label{PFG}
\end{changemargin}
\vspace{-.3 in}
\end{figure}

We claim that the displacement (along $Y_F$) between 
$\Gamma_\ell$ and $\Gamma_r$ cannot exceed $B \cdot J_F$:  
first if $\Gamma_\ell$ were to the left of $\Gamma_r$, then the relevant yellow crossing probability in $A_F$ would be bounded from below by the easy way crossing of the topological rectangle
bounded by $\Gamma_\ell$, $\Gamma_r$, $Y_I$ and $Y_F$.
By the definition of $B$ this would exceed $1 - \vartheta^2$, while by assumption, $A_F$ was constructed so that this yellow crossing probability could not exceed $1 - \vartheta$.
On the other hand, if $\Gamma_\ell$ were to the right of $\Gamma_r$, then any yellow crossing would now be forced to traverse 
the \textit{hard} way along a topological rectangle with aspect ratio exceeding $B$, 
which again by definition of $B$ would be less than $\vartheta^{2}$ which is again contrary to the assumption concerning $A_F$. 

Next, if possible, we will extend the portion of $Y_F$ which is in between $\Gamma_\ell$ and $\Gamma_r$ by an amount from $Y_{F}$ no more than an additional $B \cdot J_F$ on each side.  Obviously, we do this only if space is available, otherwise, e.g., $\|Y_F\|_\infty$ is \textit{already} less than $3B \cdot J_F$.  We denote as necessary 
the bounding vertical segments at the end of the extensions by $\gamma_\ell$ and $\gamma_r$.  In any case, these bounding segments $\gamma_\ell, \gamma_r$, should they exist will hit $\mathscr B_\ell \cup \mathscr B_r$ and we will denote the portions of $\gamma_\ell$ and $\gamma_r$ which connect $Y_F$ to $\mathscr B_\ell$ and $\mathscr B_r$ by $\tau_\ell$ and $\tau_r$, respectively.  The portion of $Y_F$ between $\tau_\ell$ and $\tau_r$ together with the $\tau$'s themselves constitute our \emph{effective region} $Y_F^{(e)}$.  (If no reconstruction is required, then we may consider $\tau_\ell = \tau_r = \emptyset$ and set $Y_F^{(e)} \equiv Y_F$.)  Notice in particular that $\|\tau_\ell\|_\infty$ and $\|\tau_r\|_\infty$ cannot exceed $J_F$.

We now argue that the yellow crossings actually occur between the effective regions with high ``conditional'' probability.  Consider the event  $\mathscr D: = \{Y_I \leadsto Y_F\} \setminus \{Y_I \leadsto Y_F^{(e)}\}$.  We first claim that the event $\mathscr D$ implies the existence of yellow \emph{and} blue crossings in fragments of aspect ratio at least $B$ and therefore has probability less than $\vartheta^{4}$.   

For simplicity, we will suppose that $\Gamma_{\ell}$ is the the left of $\Gamma_{r}$.  In this case -- assuming for non--triviality that $Y_F \neq Y_F^{(e)}$ -- it may be further supposed without any additional loss of generality that $Y_F\setminus Y_F^{(e)}$ contains points to the right of $\Gamma_{r}$ from which the yellow crossing producing the event $\mathscr D$ emanates.  
Then, the fragment between $\gamma_{r}$ and $\Gamma_{r}$
has aspect ration $B^\prime$ with $B^{\prime} \geq B$.  This fragment must be traversed the hard way by a yellow path in order that 
(this version of) the event 
$Y_F\setminus Y_F^{(e)}\leadsto Y_I$ occurs.  But in addition, a similar sort of blue path must occur to separate the lowest version of said yellow path from the yellow boundary region between $\Gamma_{r}$ and $\gamma_{r}$; the intersection of these two path events has probability less than 
$\vartheta^{4}$ and so here the desired result has been established.  Other cases are handled similarly. 

The stated upper bound on the crossing probability follows more easily: it suffices to observe that the probability of a yellow connection between $Y_I$ and $\gamma_\ell$ or $\gamma_r$ by the choice of $B$ is bounded by $\vartheta^2$.

\bigskip

\emph{The General Case.} The key point here is to have a tangible construction of $\Gamma_r$ and $\Gamma_\ell$ along with the associated $\tau_r$ and $\tau_\ell$, at which point the argument proceeds along the lines of the simplified case.  We let $x_r, \dots, \mathscr B_r$ be as defined before.  Consider the composite curve consisting of $Y_F, \mathscr B_r$, the portion of $\partial \Omega_m$ connecting $y_r$ to $y_\ell$ (not containing $\omega$) and $\mathscr B_\ell$.  The curve described forms a closed loop whose interior we denote by $\mathscr C$ and note particularly that $Y_I$ must lie in $\mathscr C$.  

We emphasize that by definition of the sliding procedure, 
the $d_{\infty}^{A_{F}}$-- distance of \textit{every} element of 
$Y_F$ to \textit{some} element of $Y_{I}$ is within a few lattice spacings of $J_{F}$  
(small discrepancies may arise due to the extended SIU's and certain lattice details).
Thus let $J_{F}^{\prime} = J_{F} + c$ -- for some constant $c$ of order unity -- denote the maximal such distance. 
Let $x_\alpha$ and $x_\beta$ denote two points along $Y_F$ with $x_\alpha$ to the right of $x_\beta$ and such that $d_\infty^{A_{F}}(x_\alpha, x_\beta) > 2 \cdot J_F^{\prime}$; thus the boxes of side $2J_{F}^{\prime}$ (i.e., ``radius'' $J_F^{\prime}$) centered about these points are disjoint.  Note that $Y_I$ must visit both of these boxes; for simplicity, let us, if necessary, shrink these boxes by a few lattice spacings (no more than $c$) so that $Y_I$ can only visit the surface of these boxes.  We denote the boxes centered at $x_\alpha$ and $x_\beta$ by $\mathfrak b_\alpha$ and $\mathfrak b_\beta$, respectively.

Our claim is that (regarding $Y_I$ as moving from $y_r$ to $y_\ell$) $Y_I$ must visit 
$\mathfrak b_\alpha$ \emph{before} it visits $\mathfrak b_\beta$.  Indeed, supposing to the contrary that $Y_I$ visits $\mathfrak b_\beta$ first, then it is noted that the curve consisting of the portion of $Y_I$ from $y_r$ till it hits $\mathfrak b_\beta$ and portions of $\partial \mathfrak b_\beta$ (plugging into $Y_F$) would divide $\mathscr C$ into two disjoint components, one containing $x_\alpha$ (on its boundary) and the other containing $y_\ell$ (on its boundary). Moreover, the former component contains all of 
$\mathfrak b_{\alpha}\cap \mathscr C$.
Thus $Y_I$ could not possibly then proceed to visit $\partial \mathfrak b_\alpha$ without crossing the aforementioned separating curve and/or $\partial \mathscr C$.

Now consider a sequence of successive non--overlapping boxes $\mathfrak b_1, \mathfrak b_2, \dots \mathfrak b_q$ centered on points in $Y_F$ oriented from right to left, so that $\mathfrak b_1$ is centered at $x_r$.
All these boxes have diameter approximately $2J_{F}^{\prime}$ which have been tuned, as above, so that $Y_{I}$ just visits their surfaces.   Let us look at the last such box which intersects $\mathscr B_r$ and denote it by $\mathfrak b_s$.  Due to the above claim, it is clear that $\mathfrak b_1, \dots, \mathfrak b_{s-1}$ are all intersected by $\mathscr B_r$ and, moreover, within each box, the relevant portion of $\mathscr B_r$ separates $Y_F$ from $Y_I$ (at least twice).  Indeed, first consider paths $\widetilde{\mathscr B}_r$ and $\widetilde{\mathscr B}_\ell$ outside of $\Omega_m$ which do not intersect the $\mathfrak b$ boxes and also connect $x_r$ to $y_r$ and $x_\ell$ to $y_\ell$, respectively.  Next consider the domain bounded by $Y_I \cup Y_F \cup \widetilde{\mathscr B}_r \cup \widetilde{\mathscr B}_\ell$; the curve $\mathscr B_r$ (and later also $\mathscr B_\ell$) must lie within this domain.  Now on each $\mathfrak b_j$, there are two portions of $\partial \mathfrak b_j$ connecting $Y_I$ to $Y_F$ and \emph{both} of these separate $y_r$ from $\mathfrak b_s$ within the domain just described. 

We will now define $\Gamma_r$ to be the portion of $\partial \mathfrak b_{s+1}$ 
which is inside $A_{F}$ and connects $Y_F$ to $Y_I$.
(There are, in fact at least two two choices; to be definitive, we choose, starting from $Y_{F}$, the``rightmost''.)   As before, we move from $\mathfrak b_s$ to $\mathfrak b_{s - k}$ where $\mathfrak b_{s-k}$ is the nearest box whose $d_\infty^{A_{F}}$-- distance from $\mathfrak b_s$ is greater than $B \cdot J_F$.  The topological statements above about $\mathscr B_r$ then permits us to define $\tau_r$ as any (lattice) path between the center of $\frak b_{s-k}$ and $\mathscr B_r$ within the box.  Note in the above that we have tacitly assumed $s \geq 1$, otherwise, e.g., $\tau_r = \emptyset$.  Finally, $\Gamma_\ell$, etc., are defined analogously.
With these tangible definitions of $\Gamma_r$, $\Gamma_\ell$, $\tau_r$ and $\tau_\ell$, the proof concludes \emph{mutatis mutantis} along the lines of the simplified case.  Notice that again, $\|\tau_\ell\|_\infty, \|\tau_r\|_\infty \leq  J_F + c \approx J_{F}$
and so the diameter of $Y_{F}^{(e)}$ is bounded above as in the statement of this lemma.

As for the stated lower bound on the diameter, let us first note, supposing
$Y_{F}^{(e)} = Y_{F}$,
that since $\|Y_F\|_\infty \geq d_\infty^{A_{F}}(\mathscr B_\ell, \mathscr B_r)$, if it were the case that $\|Y_F\|_\infty < B^{-1} \cdot J_F$, then by the choice of $B$ we would conclude that the blue crossing probability is in excess of $1 - \vartheta^2$, which is impossible.  Therefore, $\|Y_F\|_\infty \geq B^{-1} \cdot J_F$.  On the other hand, if we actually had to perform the effective regions construction, then we clearly have $\|Y_F^{(e)}\|_\infty \geq B \cdot J_F$.
\end{proof}

We will now address the cases where a successful construction of a $Y_F$ \emph{requires} a backsliding construction, as featured in Lemma \ref{slid} and the discussions thereafter.  The relevant barrier segment $Z_F$ will be constructed by a modification of $Y_{F, L}$, where $L$ is the smallest integer larger than $J_F/2$.  

We observe that by construction, each box of size $L$ around an element of $Y_{F, L}$ is visited by $Y_I$, so with the meaning of $\mathfrak b_j$ the same as in the above proof, we let $\mathfrak b_\alpha$ be the leftmost box which is not intersected by $\mathscr B_r$ and $\mathfrak b_\beta$ the rightmost box which is not intersected by $\mathscr B_\ell$.  If it is the case that $\alpha < \beta$ (considering the orientation to be from right to left) then we perform effective region construction on $Y_{F, L}$ (here we clearly have $\partial \Omega_m$ sufficiently close to $Y_{F, L}$ so that the construction can be performed with $\|\tau_\ell\|_\infty$ and  $\|\tau_r\|_\infty$ remaining well--controlled) and set $Z_F$ to be the resulting modified segment.

We are left with the case that $\alpha > \beta$; here we will set $Z_F \equiv Y_{F, L}$.  But let us note the following fact (which is only non--trivial if $\beta - \alpha > B$): if $w_{F, L} \subseteq Y_{F, L}$ is any connected subsegment which traverses more than $B$ of the $\mathfrak b_j$ boxes between $\mathfrak b_\alpha$ and $\mathfrak b_\beta$, then 
$$
\mathbb P(Y_I \leadsto w_{F, L}) > 1 - \vartheta^2,
$$
where the above crossing is allowed to take place anywhere in the region bounded by $Y_I, Y_{F, L}$ and the relevant portions of $\partial \Omega_m$.  Indeed this follows since to prevent such a connection entails a long way blue crossing with aspect ration in excess of $B$.  

\begin{prop}\label{er2}
In the cases where $Y_F$ must be obtained from $Y_I$ via a backsliding procedure as described following Lemma \ref{slid} and above, there exists an $m$ such that $G_m^{(F)}$ has the following properties:

1) $\mathbb P(Y_I \leadsto G_m^{(F)}) \in (\vartheta, 1 - \vartheta)$;

2) there is some constant $\kappa > 0$ such that $\|G_m^{(F)}\|_\infty \leq \kappa B \cdot L$.
\end{prop}

\begin{proof}
Let us temporarily denote by $\widetilde Y_F$ the first successor segment such that $\mathbb P(Y_I \leadsto \widetilde Y_F) < \vartheta$ (i.e., the suppressed second index is approximately equal to $2L$).  First, if necessary, we perform effective region like construction on $\widetilde Y_F$: again, we reiterate that if the relevant $\Gamma_\ell$ lies to the right of $\Gamma_r$, then the relevant $\tau_\ell$ and $\tau_r$ are well controlled and the diameter of the effective region obeys the bounds stated in Lemma \ref{er} (relative to $\widetilde J_F$ which is the relevant separation distance between $Y_I$ and $\widetilde Y_F$).  On the other hand, if $\Gamma_\ell$ is to the left of $\Gamma_r$, then the observation immediately preceding the statement of this proposition implies that the separation between $\Gamma_\ell$ and $\Gamma_r$ cannot exceed $B \cdot L$ and, if necessary, we may perform an additional effective region construction.

In any case, as in Lemma \ref{er}, we have manufactured a $\widetilde Y_F^{(e)}$ whose diameter does not exceed $\kappa'B \cdot L$.  Now we will carry out the backsliding procedure starting with $\widetilde Y_F^{(e)}$ with the barrier $Z_F$ as described above.  We observe that after approximately $L$ iterations, we will have reached $Z_F$.  We will not do \emph{more} than an additional $\kappa'B \cdot L$ iterations (in order to achieve the crossing probability criterion) because if so,
we would have subsumed a portion of $Z_{F}$ at least this large (or perhaps all of it)   and, by the observation immediately preceding the proposition,
the crossing probability from $Y_I$ would exceed $1 - \vartheta^2$.  Thus, $m$ is certainly less than $(\kappa' B + 1) \cdot L$.
It follows that $\|G_m^{(F)}\|_\infty$ is less than $\|\widetilde Y_F^{(e)}\|_\infty + c \cdot m$, where the constant $c$ corresponds to increase in diameter due to the $d_\infty$--neighborhood construction.  We thus have that $\|G_m^{(F)}\|_\infty$ is less than or equal to $\kappa B \cdot L$ for some constant $\kappa$, as stated.
\end{proof}

We can now set $Y_F$ to be equal to the $G_m^{(F)}$ existentiated above and set $J_F$ to be equal to $L$.  For future reference we will adopt $\kappa$ as the constant bounding the diameter of all such regions.

\subsubsection{The $R$--Construction \rm{(percolative boxes)}}
Our objectives will eventually be achieved  by establishing the existence of monochrome percolation connections between e.g., (the vicinity of) $\omega$ and the central region 
$\mathcal D_\Delta$.  This will be accomplished by 
establishing a grid of contiguous boxes within each of the successor ring fragments.  This entails a suitable modification of $Y_{F}^{(\cdot)}$, which will then \textit{define} the final ring fragment of interest.  Most pertinently, as will emerge later, the sizes of the boxes in neighboring fragments will differ by at most a uniform scale factor.

\begin{remark}
We remark that in the forthcoming use of boxes, as far as connectivity properties are concerned, all notions are inherited from the standard 2D square site lattice, e.g., boxes are connected if they share an edge in common and, dually, $*$--connected means an edge or corner in common.
\end{remark}

Suppose now that $Y_I, Y_F$ and $J_F$ are such that the conclusion of Lemma \ref{er} holds (particularly, $\mathbb P(Y_I \leadsto Y_F) \in (\theta'', 1 - \theta'')$).  Let us tile the region enclosed by $Y_I$ and $Y_F$ (i.e., $A_F$) by boxes of size $2^{-r} \cdot J_F$ for some $r > 0$.  The boxes under consideration will include all boxes whose closure intersects the closure of $A_{F}$; the precise value of $r$ will be specified later.
We remark that these boxes will be a super--grid for boxes of the ultimate small scale which will be $2^{-2r} \cdot J_{F}$.  We start with the larger scale:

\bigskip

\noindent{\bf Claim.}  There exists some fixed $r_0 \equiv r(\theta'') > 0$ such that if $r > r_0$ then among the aforementioned boxes of scale $2^{-r} \cdot J_{F}$ there is a connected path 
between $Y_I$ and $Y_F$ by boxes which do not intersect the blue boundaries.

\bigskip

\noindent \emph{Proof of Claim.}  Let $\beta$ denote any box which intersects $\mathscr B_{\ell}$.  Then we assert, for $r$ chosen suitably large, that $\beta$ or any box in its $*$--connected neighborhood does not meet $\mathscr B_{r}$.  Indeed, supposing to the contrary, let us consider the $r-2$ annuli of doubling sizes 
$2^{-r}\cdot J_F\times [4, \dots,  2^{r-1}]$ with $\beta$ at the center.  Since we have started at $4\times 2^{-r}\cdot J_F$ the inner ring always contains the $*$--neighborhood of $\beta$ and since we have ended at $\frac{1}{2} \cdot J_{F}$, it cannot be the case that \textit{both} $Y_{F}$ and $Y_{I}$ penetrate (any of) these annuli.  
In each such annulus, by weak scale invariance of critical percolation, a blue circuit independently exists with probability at least some $\lambda$ and each such blue circuit would connect $\mathscr B_{\ell}$ to $\mathscr B_{r}$ in $A_{F}$, thereby preventing a yellow crossing.  But then for $r$ large enough, the yellow crossing probability would be too small; indeed, 
$$
\theta ^{\prime\prime} \leq \mathbb P(Y_I \leadsto Y_F) \leq (1-\lambda)^{r-2}.
$$
The above \emph{defines} $r_{0}$:  $(1-\lambda)^{r_{0}-2} := \theta^{\prime\prime}$ and so the assertion has been established.  Now the claim will (eventually) follow:  let $y_{r}$ and $y_{\ell}$ denote the left and right endpoints of $Y_{I}$ and moving from 
$y_{\ell}$ to $y_{r}$ along $Y_{I}$ let $\beta_{\ell}$ denote the last (rightmost) box intersected by $Y_{I}$ which also intersects $\mathscr B_{\ell}$.  We similarly define 
$\beta^{\prime}_{\ell}$ along $Y_{F}$ and $\beta_{r}, \beta_{r}^{\prime}$ relative to $\mathscr B_R$.
We further define $\mathbf Y_{I}$, $\mathbf Y_{F}$, $\mathbf B_{\ell}$ and 
$\mathbf B_{r}$:  the set $\mathbf Y_{I}$ consists of those boxes which meet the portion of 
$Y_{I}$ which is after its last exit from $\beta_{\ell}$ and before its first entrance 
into $\beta_{r}$.  Similarly, let $\mathbf B_{\ell}$ denote those boxes which meet the portion of 
$\mathscr B_{\ell}$ after its last exit from $\beta_{\ell}$ and before its first entrance into 
$\beta_{\ell}^{\prime}$.  Similarly for $\mathbf Y_{F}$ and $\mathbf B_{r}$.

It is seen that save for the four corners mentioned, all these sets are disjoint:
the $\mathbf Y$ pair are ``well separated'', the $\mathbf B$ pair are (not quite as well) separated due to the above assertion and, e.g., $\mathbf Y_{I}$ and $\mathbf B_{\ell}$ are disjoint save for $\beta_{\ell}$
essentially by definition.  Now these sets, which are $*$--connected objects (although not necessarily $*$--connected paths) certainly form the boundary of a ``crossing problem arena'' -- we temporarily delete from consideration all other boxes which meet $\partial A_{F}$.  The desired crossing exists because, due to the first assertion, we have, for example, that the $*$--connected boundary of $\mathbf B_{\ell}$ is disjoint from $\mathbf B_{r}$.
\qed
%
\begin{remark}
For the benefit of the forthcoming, the above display should be modified to read 
$(1 - \lambda)^{r} < \vartheta^{2} \ll \theta^{\prime\prime}$ -- where in addition, we have now stipulated $\vartheta \sim \theta^{\prime\prime} \ll 1$.  This will set our scales.  Thus we envision $2^{r}$ as comparable to $B$ but, in any case, we certainly have
$$
2^{r} \geq B
$$
\end{remark}
We will now define the current modification to $Y_F$:

\noindent\textbf{Definition of $Y_F^{(b)}$.}  
Consider now the grid of scale 
$2^{-2r} \cdot J_{F}$ -- considerably smaller than the previous grid.  Here it is stipulated that 
$r$ (with $r > r_{0}$) has been chosen so as to satisfy
$$
(1 - \lambda)^{r-1} < (\theta^{\prime\prime})^{2},
$$
with $\lambda$ being the probability of a monochrome circuit in an annulus where the outer scale is twice the inner scale.
Again we consider the (closed) boxes which have non--zero intersection with the 
region $A_{F}$ -- and for the immediate future, no others.  From this set we delete all boxes which have non--empty intersection with $Y_F$ \textit{and their $*$--neighboring boxes}.  

The remaining boxes may now consist of several components.  Nevertheless, we may temporarily consider all boxes with at least one edge not shared by another box in the collection to be a boundary box.  These ``exposed'' box--edges form closed circuits -- the edge boundaries of the components.  Among these circuits there is to be found a unique path which begins on $\mathscr B_{\ell}$, ends on $\mathscr B_{r}$, have 
no other encounters with either of the  $\mathscr B$'s and which is disjoint from the old 
$\mathbf Y_{I}$.  This special path is characterized as follows; we consider a larger set of  paths which are allowed to use edges of \textit{any} box.  These paths have the same general description used above: they connect $\mathscr B_{\ell}$ to $\mathscr B_{r}$ and consist of whole (box) edges save, possibly, for the first and last wherein occurs the only encounter with $\mathscr B_{\ell}\cup\mathscr B_{r}$.  Such a set is obviously non--empty since $\mathscr B_{\ell}$ is connected to $\mathscr B_{r}$ by boxes (even of the larger scale) and the paths are partially ordered with respect to their separation from $Y_{I}$.  The special path is the one that is ``furthest'' from $Y_{I}$.

One final modification is required.  Consider the boxes which share at least one edge (or fragment thereof) with the aforementioned path.  By the ordering considerations of the construction, each edge belongs to a unique box (the collection of which in fact form a $*$--connected chain).  Moving from left to right along the path, we define, as in the context of the larger grid, the box $\mathfrak b_\ell$ to be the rightmost of the boxes which intersect 
$\mathscr B_{\ell}$ and similarly for $\mathfrak b_r$.  We now consider the portion of the path connecting $\mathfrak b_{\ell}$ to $\mathfrak b_{r}$ and adjoin to it the box edges (partial or otherwise) of $\mathfrak b_{\ell}$ and $\mathfrak b_{r}$ so that overall we have reconstituted a connection between $\mathscr B_{\ell}$ and $\mathscr B_{r}$.  This final step should be enacted by choosing the path along the box edges to be ``as far away'' from $Y_I$ as possible.

The definition is essentially complete, the above constructed path constitutes the segment
$Y_{F}^{(b)}$.  The region $A_{F}^{(b)}$ is the topological rectangle with boundaries 
$Y_{I}$, $Y_{F}^{(b)}$ and the appropriate portions of $\mathscr B_{\ell}$ and $\mathscr B_{r}$ which connect these segments.    Finally, we use the notation
$$
b_F := 2^{-2r} \cdot J_F
$$
for the scale -- sidelength -- of the individual boxes.  

\begin{remark}
\label{bkd}
Let us note for future reference that by virtue of having used $*$--connected neighborhoods, it is the case that there is a full layer of connected boxes separating $Y_F$ from $Y_F^{(b)}$.  The fact that there is a full layer is clear; the fact that this layer is connected follows from the observation that the boxes which intersect $Y_F$ were, certainly, $*$--connected, so the layer represents a portion of the dual circuit surrounding this set.
\end{remark}

We now observe that $Y_F^{(b)}$, $A_{F}^{(b)}$ etc., have various anticipated properties:
\begin{prop}\label{bcs}
Consider the boundary boxes of $A_{F}^{(b)}$ that share portions of their boundary with $Y_{F}^{(b)}$.  Then these boxes form a $*$--connected cluster -- and hence $Y_{F}^{(b)}$ itself is connected.  Moreover, each such box is connected by a (connected) path of full boxes to boxes which intersect $Y_{I}$ with a path length --  measured in number of boxes -- which does not exceed some universal constant $\textsc l \in (0, \infty)$.
\end{prop}
\begin{proof}
It is noted that by construction, $Y_{F}^{(b)}$ is constituted from a non--repeating sequence of attached box edges all of which are complete save the first and last.  Therefore, it is manifestly connected and, moreover, the boxes which contribute boundary elements to this segment form a $*$--connected chain running from
$\mathfrak b_{\ell}$ to $\mathfrak b_{r}$, inclusive.  We denote this chain by 
$\mathbf y_{F}$ and define
$\mathbf y_{I}$,
$\mathbf{b}_{\ell}$ and
$\mathbf{b}_{r}$ in the same fashion as their upper case counterparts were defined using the larger scale boxes.  The region bounded by $\mathbf y_{F}$
and appropriate elements of the other sets 
$\mathbf y_{I}$,
$\mathbf{b}_{\ell}$ and
$\mathbf{b}_{r}$
form a topological rectangle along with associated boundaries so that opposing pairs are well separated (as in the claim for the larger scale boxes) and hence any box in $\mathbf y_{F}$
can be connected to some box in $\mathbf y_{I}$.  

Finally, concerning the length of the connection, let us provide a crude estimate: it is known that the diameter of $Y_F^{(b)}$ is bounded by $\|Y_F\|_\infty \leq \kappa B \cdot J_F$ (by Lemma \ref{er} and Proposition \ref{er2}); moreover, every point on $Y_F^{(b)}$ is separated from some point on $Y_I$ by $d_\infty$--distance some constant times $\kappa B \cdot J_F$.  Thus, a box which is a reasonable multiple of $B \cdot J_F$ centered about $Y_F$ -- let us call it $\mathbb B$ -- contains all of $Y_F$ and paths between every point on $Y_F$ and some point on $Y_I$ \emph{and} a substantial neighborhood (e.g., larger than $J_F$) of these paths.  

Now consider any ``microscopic'' path $\mathscr P$ in $A_F$ between a point on $Y_F$ and some point on $Y_I$ contained in $\mathbb B$ and let us denote by $\mathfrak b_\mathscr P$ the set of small boxes visited by $\mathscr P$.  Of course, $\mathfrak b_\mathscr P$ may contain boundary boxes from $\mathbf b_\ell$ and/or $\mathbf b_r$; however, since these are well--separated, the $*$--connected closure of $\mathfrak b_\mathscr P$ manifestly contains a connected path of full boxes between $\mathbf y_F$ and $\mathbf y_I$ (indeed, the $*$--connected closure of e.g., $\mathbf b_\ell$ itself consist of full boxes).  Finally, it is certainly the case that \emph{any} such path cannot possibly consist of more boxes than the total number of boxes in the aforementioned bounding set $\mathbb B$.  Since both $\mathbb B$ and the small boxes have area proportional to a universal constant times $J_F^2$, the result follows for some $\textsc l$.
\end{proof}

\bigskip

Our next claim is that the replacement of $Y_F$ by $Y_F^{(b)}$ does not substantially alter the yellow crossing probability:

\noindent \textbf{Claim.}  Suppose $\mathbb P(Y_I \leadsto Y_F) \in (\theta'', 1 - \theta'')$, then 
$$
\theta'' \leq \mathbb P(Y_I \leadsto Y_F^{(b)}) < \frac{1- \theta''}{(1 - (\theta'')^2)^2}
\approx 1 - \theta'' + 2(\theta'')^2.
$$

\noindent\emph{Proof of Claim.}  The lower bound is immediate since a crossing from $Y_I$ to $Y_F$ certainly implies a crossing to $Y_F^{(b)}$.  The upper bound follows from an easy continuation of crossings argument: indeed, consider e.g., the ``left'' endpoint of $Y_F$ whose grid box we surround by an annulus with inner scale $2^r \cdot b_F$ and outer scale $2^r$ times this inner scale.  Let us consider the event $\mathbb G$ that the rightmost crossing from $Y_I$ to $Y_F^{(b)}$ is outside the inner square defining the above annulus.  Then, we have at least $r$ independent chances -- by setting up independent annuli -- to continue the rightmost crossing to $Y_F$.  This continuation therefore happens with (conditional) probability in excess of $1 - (\theta'')^2$ by the choice of $r$ and so  
$$
\mathbb P(\{Y_I \leadsto Y_F \mid \{Y_I \leadsto Y_F^{(b)}\} \cap \mathbb G) > 1 - (\theta'')^2.$$

On the other hand, let $\mathbb O_Q$ denote the event of a yellow circuit somewhere inside the
$b_{F}\cdot2^{r} \times b_{F}\cdot2^{2r}$
annulus.  We note that should the event $\mathbb O_Q$ occur, then, indeed, the rightmost crossing \emph{is} to the right of this inner square and so by the choice of $r$, $\mathbb P(\mathbb G) \geq \mathbb P(\mathbb O_Q) > 1 - (\theta'')^2$ 

The desired inequality now follows:  
\[
\begin{split}
\mathbb P(Y_I \leadsto Y_F) &\geq \mathbb P(Y_I \leadsto Y_F \mid \{ Y_I \leadsto Y_F^{(b)}\} \cap \mathbb G) \cdot \mathbb P(\{Y_I \leadsto Y_F^{(b)}\} \cap \mathbb G)\\
&> (1 - (\theta'')^2)^2 \cdot \mathbb P(Y_I \leadsto Y_F^{(b)}).
\end{split}
\]
Here to arrive at the last line we have used the FKG inequality.
\qed

\begin{cor} \label{ddf}
Consider the curve, denoted by $\overline{Y}_{F}^{(b)}$ which consists of the outer edges (i.e., closer to $Y_F$) of the boxes which share an (inner) edge with a box in 
$\mathbf{y}_{F}$ and attached to $\mathscr B_{\ell}$ and $\mathscr B_{r}$ by an analogous procedure as was used for $Y_{F}^{(b)}$.
  Then
$$
\mathbb P(\overline{Y}_{F}^{(b)} \leadsto Y_F^{(b)}) > 1 - (\theta'')^2.
$$
\end{cor}
\begin{proof}
This follows directly from the above argument using the observation that 
$\overline{Y}_{F}^{(b)}$ separates $Y_{F}$ from $Y_{F}^{(b)}$.
\end{proof}
\begin{remark}
We remark that we should think of the above claim ``dually'', i.e., that the complementary bound 
$$
\theta'' - c_2(\theta'')^2 < \mathbb P(\mathscr B_\ell' \leadsto \mathscr B_r') \leq 1 - \theta''
$$
(for $\theta^{\prime\prime}$ small and some constant $c_{2}$)
holds for blue crossings; here $\mathscr B_\ell, \mathscr B_r$ now denote the appropriate ``left'' and ``right'' blue boundaries of the topological rectangle formed by $Y_I$ and $Y_F^{(b)}$.  Indeed, as will unfold below, all that is actually used for yellow crossings is the ability to cross these regions via paths inside the boxes with a probability which is independent of the details of the region.  
\end{remark}

\subsection{Induction}\label{fcont}
We are now ready to assemble all ingredients and describe the full inductive procedure.  Our goal is to show that there are box connections between all successive yellow segments \emph{and} that the box scale of neighboring layers do not differ by too much.

First let us set out the base case for our induction.  Let $Y_0:= \partial S_0(\omega) \cap \partial \mathcal D_\Delta$ be as described before.  If 
$\mathbb P(Y_0 \leadsto \omega) \geq \vartheta$, then stop. (Indeed, this is one of the stopping criterion for our iteration and if it actually occurs on the first step, then clearly the lattice spacing is too large to be worthy of any detailed consideration.)  Otherwise, let us perform the $S, Q, R$--constructions to yield some $\widehat Y_1$ so that $\mathbb P(\widehat Y_1 \leadsto Y_0) \in (\theta'', 1 - \theta'')$.  Recall that here $\theta''$ may differ from $\vartheta$ due to the $O(\vartheta^2)$ errors in the $Q, R$--construction sections.  It is noted that since $\mathcal D_\Delta$ has a radius half the maximal amount, for judicious choice of $\vartheta$, the ancillary constructions, in particular the $Q$--construction, are actually not necessary.

We will assume the existence of segments
\[P_0 \equiv \widehat Y_0, ~P_1 \equiv \widehat Y_1, ~\dots,~ P_{k-1} \equiv \widehat Y_{k-1}, ~T_k \equiv \widehat Y_k.\]
(Here the $\widehat~$ denotes modifications to the original segments $Y_\ell$ from the $Q, R$--constructions.)  The following are our \emph{inductive hypotheses}:

0) $T_k$ is the box construction version of some $Y_k$ and there is some $\overline{Y}_{k}$ which is the $d_{\infty}$-- neighborhood boundary of $T_{k}$ such that 

\indent \indent i) $\mathbb P(\overline{Y}_{k} \leadsto T_k) \geq 1 - (\theta'')^2$ (by Corollary \ref{ddf}); 

\indent \indent ii) $d_\infty(T_k, Y_k) > b_k$ (by Remark \ref{bkd}); 

1) the following conclusions hold for $P_{k-1}$ and $T_k$:

\indent \indent i) $\mathbb P(P_{k-1} \leadsto T_k) \in (\theta'', 1 - \theta'')$;

\indent \indent ii) all of $T_k$ can be connected to 
 $P_{k-1}$ via boxes of size $b_k$ completely unobstructed by $\partial \Omega_m$ and further the number of boxes required for such a connection is not in excess of {\textsc l} (by Proposition \ref{bcs});

\indent \indent iii) it is the case that (see the effective regions construction subsubsection)
\[B^{-1} \cdot J_k \leq \|P_{k-1}\|_\infty \leq 2^{2r + 1} (\kappa B) \cdot J_k,~~~B^{-1} \cdot J_k \leq \|T_k\|_\infty \leq \kappa B \cdot J_k.\]

2) the conclusions of 1.ii) and 1.iii) also hold for (prior) successive segments $P_{\ell-1}, P_\ell$, $1 \leq \ell \leq k-1$; also a weakened version of 1.i) holds for these prior segments: $\mathbb P(P_{j-1} \leadsto P_j) \leq 1 - \theta''$ (which we think of as a lower bound on the blue crossings).

Here the letter $P$ denotes what is considered a \emph{permanent} yellow segment whereas $T$ denotes a \emph{temporary} segment.  
Let us note that we will not require the lower bound on yellow crossing probabilities per se except for the last layer consisting of $P_{k-1}, T_k$.  

\bigskip

\textbf{Step 1.}  We first construct some segment $Y_{k+1}$ by sliding $T_k$ 
towards $\omega$ and backsliding if necessary so that 
$$
\mathbb P(T_k \leadsto Y_{k+1}) \in (\vartheta, 1 - \vartheta).
$$
%
%

\textbf{Step 2.} Next we perform the effective regions construction and the box construction to yield $\widehat Y_{k+1}$ so that $\mathbb P(T_k \leadsto \widehat Y_{k+1}) \in (\theta'', 1-\theta'')$, $\|\widehat Y_{k+1}\|_\infty \leq \kappa B \cdot J_{k+1}$ and there are suitable box connections between $T_k$ and $\widehat Y_{k+1}$ with boxes of scale $b_{k+1} = 2^{-2r} \cdot J_{k+1}$ (see Lemma \ref{er}, Proposition \ref{er2} and Proposition \ref{bcs}).  Notice that we also immediately have from the bound on the crossing probabilities and Lemma \ref{er} that
$$ 
\|T_k\|_\infty \geq B^{-1} \cdot J_{k+1},~~~\|Y_{k+1}\|_\infty \geq B^{-1} \cdot J_{k+1}. 
$$

\textbf{Step 3.} We have now satisfied items 1.i) and 1.ii) of the inductive hypothesis for the layer $(T_k, \widehat Y_{k+1})$.  

\textbf{Step 4.} 
Now we will verify item 1.iii). First we claim that $J_{k+1} > b_k /2$.  

Indeed, recall that there is a connected neighborhood of boxes separating $T_k$ from $Y_k$ (again see Remark \ref{bkd}).  The (outer) boundary of this
neighborhood is $\overline Y_{k}$ which by item 0.i) is connected to $T_k$
with too high a probability to consider stopping the process until at least some portion of the evolving neighborhood boundary reaches past it.  However, by the nature of the neighborhood sliding construction the entire evolving boundary  pushes through this curve coherently.
So if the construction of $Y_{k+1}$ does not entail a backsliding, then the result immediately follows, in fact without the factor of two.
Now if a backsliding were required, then recall that we used $Y_{k+1, L_{k+1}}$ as a barrier (here $L_{k+1}$ is the closest integer to $J_{k+1}/2$, as described in the discussions preceding Proposition \ref{er2}) and so the claim follows by item 0.ii).  

Item 1.iii) has now been verified since we now have $\|T_k\|_\infty \leq \kappa B \cdot J_k \leq 2^{2r + 1} (\kappa B) \cdot J_{k+1}$.   All induction hypotheses have been verified, so we set $P_k \equiv T_k$ and  $T_{k+1} \equiv \widehat Y_{k+1}$.
\bigskip

The induction can now be continued towards $\omega$, starting with $P_k$ and $T_{k+1}$, provided that $\mathbb P(T_{k+1} \leadsto \omega)< \vartheta$ and $k + 1 \leq \Gamma \cdot \log n$, with $\Gamma$ as in the statement of Theorem \ref{hreg} -- otherwise we stop.  

\begin{remark}
We note particularly that from item 3 in step 4 of the induction, the percolating boxes are connected going from one layer to the next.

Also, for notational convenience, in the statement of Theorem \ref{hreg}, we have reverted back to using $\vartheta$ (so $\vartheta$ there corresponds to $\theta''$ here).
\end{remark}

\subsection{A Refinement}

We will require one additional property of these Harris systems.  First let us define some terminology:

\begin{defn}\label{hrp}
Let $\Omega \subset \mathbb C$ be a bounded, simply connected domain and $\Omega_m$ some interior discretization of $\Omega$.
For $\omega \in \Omega_m$, consider the inductive construction as described, yielding $P_1, P_2,$ etc., until the crossing probability between $\omega$ and the last $P_\ell$ is 
less than
$\vartheta$ -- indicating that we have approximately reached the unit scale --  
or until we have succeeded a sufficient number of segments.  

We remind the reader that we refer to the topological rectangles formed by successive $P_\ell$'s as \emph{Harris rings} and the amalgamated system of these segments around $\omega$ the \emph{Harris system stationed at $\omega$}.  

\end{defn}

For our purposes we will also need to show that for $n$ sufficiently large, 
for the marked point corresponding to $A$, the relevant Harris segments have endpoints lying in the anticipated boundary regions:

\begin{lemma}\label{ped}
Let $\Omega \subseteq \mathbb C$ be a bounded simply connected domain with marked boundary prime ends $A, B, C, D \in \partial \Omega$ (in counterclockwise order) and suppose $\Omega_m$ is an interior approximation to $\Omega$ with $A_m, B_m, C_m, D_m \in \partial \Omega_m$ approximating $A, B, C, D$.  Consider the hexagonal tiling problem studied in \cite{stas_perc} or the flower models introduced in \cite{CL} (in which case we assume the Minkowski dimension of $\partial \Omega$ is less than 2) and the Harris system stationed at $A_m$.  Then there is a number $v_{A}$ such that for all $m$ 
sufficiently large, all but $v_{A}$ of the Harris segments form conduits from
$[D_{m}, A_{m}]$ to $[A_{m}, B_{m}]$.
More precisely, under uniformization, there exists some $\eta > 0$ such that all but $v_{A} = v_{A}(\eta)$ of these segments begin and end in the $\eta$--neighborhood of the pre--image of $A$.  
\end{lemma}

\begin{proof}
Let $\varphi: \mathbb D \rightarrow \Omega$ be the uniformization map with $\varphi(0) = z_0$ for some $z_0 \in \mathcal D_\Delta$ and let $\zeta_A, \zeta_B, \zeta_C, \zeta_D$ denote the pre--images of $A, B, C, D$, respectively.  
Let $\eta > 0$ denote any number smaller than e.g., half the distance separating any of these pre--images.  Let $N_{\eta}(\zeta_A)$ denote the
$\eta$--neighborhood of $A$ and let 
$\{r_{d}, r_{b}\}$ denote the pair
$N_{\eta}(\zeta_A)\cap \partial \mathbb D$
with $r_{d}$ in between $\zeta_{A}$ and $\zeta_{D}$ and $r_{b}$ between
$\zeta_{A}$ and $\zeta_{B}$.  Similarly, about the point $\zeta_{C}$ we have
$N_{\eta}(\zeta_A)\cap \partial \mathbb D := \{s_{d}, s_{b}\}$.

We denote by $\mathscr G_{d}$ the continuum crossing probability from
$[\zeta_{A},r_{d}]$ to $[s_{d}, \zeta_{C}]$ 
(with $(\mathbb D; \zeta_{A}, \zeta_{C},$ $s_{d}, r_{d})$ regarded as a conformal rectangle) and similarly $\mathscr G_{b}$
for the continuum crossing probability from
$[r_{b}, \zeta_{A}]$ to $[\zeta_{C}, s_{b}]$.  It is manifestly clear that these are non--zero since all relevant cross ratios are finite.  

Now consider $\Omega$ as a conformal polygon with (corresponding) marked points (or prime ends)  $A, R_{b}, B, S_{b}, \dots R_{d}$ (corresponding to $\zeta_A, r_b, \zeta_B, s_b, \dots, r_d$) and
$\Omega_{m}$ with marked boundary points
$A_{m}, \dots R_{d_{m}}$ the relevant discrete approximation.  
It is emphasized, perhaps unnecessarily, that this is just 
$\Omega_{m}$ with $A,B,C,D$ and with four additional boundary points marked and added in.  It follows by conformal invariance and convergence to Cardy's Formula that the probability of a crossing in $\Omega_{m}$ from
$[R_{b_{m}}, A_{m}]$ to $[C_{m}, S_{b_{m}}]$ converges to 
$\mathscr G_{b}$ and similarly for the crossings from 
$[ A_{m}, R_{d_{m}}]$ to $[S_{d_{m}}, C_{m}]$.

We shall need an additional construct, denoted by $\Phi_{m}$ which is best described as the intersection of three events:
(i)  a yellow connection between $\partial \mathcal D_{\Delta}$ and 
$[R_{d_{m}}, S_{d_{m}}]$, (ii) a similar connection between 
$\partial \mathcal D_{\Delta}$ and 
$[R_{b_{m}}, S_{b_{m}}]$ and (iii) a yellow circuit in 
$\Omega_{m}\setminus \mathcal D_{\Delta}$.  It is observed that the intersection of these three events certainly implies a crossing between $[R_{d_{m}}, S_{d_{m}}]$ and $[R_{b_{m}}, S_{b_{m}}]$.

It is noted that item (iii) has probability uniformly bounded from below since 
$\mathcal D_{\Delta}$ is contained in a circle twice its size.  As for the other two, we must return to the continuum problem in $\mathbb D$.  
Let $\mathcal E \subseteq \mathbb D$ denote the preimage of 
$\mathcal D_{\Delta}$ under uniformization with corresponding evenly spaced boundary points 
$p_{1}, p_{2}, p_{3}$ and  $p_{4}$.  
Let us pick an adjacent pair of points -- conveniently assumed to be 
$p_{1}$ and $p_{2}$ -- which may be envisioned as approximately 
facing the $[s_{d}, r_{d}]$ segment of $\partial \mathbb D$.
We now connect $r_{d}$ and $p_{1}$ with a smooth curve in $\mathbb D$
and similarly for $s_{d}$ and $p_{2}$.  It is seen that these two lines along with the $[s_{d}, r_{d}]$ portion of 
$\partial \mathbb D$ and the $[p_{1}, p_{2}]$ portion of 
$\mathcal E$ are the boundaries of a conformal rectangle.  We let 
$\mathscr L_{d}$ denote the continuum crossing probability from
$[p_{1}, p_{2}]$ to $[s_{d}, r_{d}]$ within the specified rectangle.  

We perform a similar construct involving 
$p_{3}$, $p_{4}$, $s_{b}$ and $r_{b}$ and denote by 
$\mathscr L_{b}$ the corresponding continuum crossing probability.  Thus, as was the case above, in the corresponding subsets of $\Omega_{m}$, it is the case that as 
$m\to\infty$, the probability of observing yellow crossings of the type corresponding to the aforementioned crossings in (i) and (ii) tend to 
$\mathscr L_{d}$ and $\mathscr L_{b}$, respectively.  
(While of no essential consequence, we might mention that at the discrete level, the relevant portions of $\partial \mathcal D_{\Delta}$ may be \textit{defined}
to coincide with the inner approximations of the subdomains we have just considered.)

Let us call $\mathbb G_{m}$ the intersection of all these events:
$\Phi_{m}$ and the pair of 
$[R_{b_{m}}, R_{d_{m}}]\leadsto [S_{b_{m}}, S_{d_{m}}]$ crossings (corresponding to $\mathscr G_b$ and $\mathscr G_b$).  Then we have, uniformly in $m$ for $m$ sufficiently large, 
$$
\mathbb P(\mathbb G_{m}) \geq \sigma
$$
for some $\sigma = \sigma(\eta) > 0$.  

We next make the following claim:  

\bigskip

\noindent {\bf Claim.} Consider the event that there is a blue path beginning and ending on $\partial \Omega_{m}$ that seperates $A_{m}$ from 
$\mathcal D_{\Delta}$.  Then, if the event $\mathbb G_{m}$ also occurs, it must be the case that (modulo orientation) the path begins on $[R_{b_{m}}, A_{m}]$ and ends on $[A_{m}, R_{d_{m}}]$.

\bigskip

\noindent \emph{Proof of Claim.}  To avoid clutter, we will temporarily dispense with all $m$--subscripts.  Note that since $A_m, R_{b_m}, B_m, C_m, D_m, R_{d_m}$ divide the boundary into six segments, there are $\frac{1}{2}\cdot 6 \cdot 7 = 21$ cases to consider and, therefore, twenty to eliminate.  Let us enumerate the cases:

$\circ$  A crossing from $[C,A]$ to $[R_{b}, C]$ or from $[A,C]$ to 
$[C, R_{d}]$  (5 cases): each possibility is prevented by (at least) one of the yellow crossings between the segments in $[R_{d}, R_{b}]$ 
and $[S_{b}, S_{d}]$.

$\circ$  Corner cases, e.g., at the $D$ corner, $[C, D]$ to $[D, R_{d}]$ (4 cases): recalling that the blue path must separate $\mathcal D_{\Delta}$ and $A$, these are obstructed by the yellow circuit about $\mathcal D_{\Delta}$ which is connected to the \textit{opposite} $R\cdot S$ boundary, which in this example corresponds to $[R_{b}, S_{b}]$.  (We note that these circuits are constructed precisely to prevent the possibility of connections ``sneaking'' through $\mathcal D_\Delta$.)

$\circ$  An $[R_{d}, R_{b}]$ segment connected to a $[B, C]$ or a $[C, D]$ segment (4 cases): these are prevented by the yellow crossing from $[S_{d}, R_{d}]$ to $[R_{b}, S_{b}]$.  

$\circ$  Diagonal (same to same) paths, e.g., $[C, D]$ to 
$[C, D]$ (6 cases): recalling the separation clause, these are obstructed by the connection of the circuit around $\mathcal D_{\Delta}$ and its connection to whichever -- or both -- 
$R\cdot S$ segment which is \textit{not} where the blue path begins and ends.  In this example this corresponds to $[R_b, S_b]$.

$\circ$  Finally, $[D, C]$ to $[C, B]$: this is the same as the previous case.  

\noindent The claim is proved.
\qed

\bigskip

With the above in hand, the rest of the proof of this lemma is immediate.  Let 
$v_{A}^{\prime}$ denote the number of Harris segments in the system stationed at $A_{m}$ which do \textit{not} begin on $[R_{b_{m}}, A]$ and end on 
$[A, R_{d_{m}}]$.  (I.e., the twenty cases treated above.)  Letting $\mathbb B_{m}$ denote the event of a blue circuit of the type described in the claim, we have
$$
1 - \sigma \geq 1 - \mathbb P(\mathbb G_{m}) \geq \mathbb P(\mathbb B_{m})
\geq 1 - (1 - \vartheta)^{v_{A}^{\prime}}
$$
which necessarily implies $v_{A}^{\prime}$ is bounded above (independently of $m$) by the ratio $\log \sigma/\log(1 - \vartheta)$.  Clearly, $v_A^\prime \geq v_A$ as in the statement of the lemma so the result has been established.

\end{proof}

\noindent \textbf{Acknowledgements.}

The authors acknowledge the hospitality of the Fields Institute during the Thematic Program on Dynamics and Transport in Disordered Systems where this work was initiated.

I.~B.~was partially supported by an NSERC Discovery Grant 5810--2009--298433.

L.~Chayes was supported by the NSF under the auspices of the grants DMS--08--05486 and  PHY--12--05295.

H.~K.~L.~was supported by the NSF under the grant DMS--10--04735 and by the Olga Taussky--John Todd Instructorship at Caltech.

The authors also acknowledge some detailed useful comments from a source who will remain anonymous.

\end{document}